\documentclass[11pt,letterpaper,twoside,reqno,nosumlimits]{amsart}

\usepackage{etoolbox}
\patchcmd{\section}{\scshape}{\bfseries}{}{}
\makeatletter
\renewcommand{\@secnumfont}{\bfseries}
\makeatother

\usepackage[dvipsnames]{xcolor}
\patchcmd{\section}{\normalfont}{\normalfont\color{MidnightBlue}}{}{}
\patchcmd{\subsection}{\normalfont}{\normalfont\color{MidnightBlue}}{}{}

\makeatletter
\def\subsubsection{\@startsection{subsubsection}{3}%
\z@{.5\linespacing\@plus.7\linespacing}{-.5em}%
{\normalfont\bfseries}}
\makeatother

\usepackage{fancyhdr}
\usepackage{amsmath,amsfonts,amsbsy,amsgen,amscd,mathrsfs,amssymb,amsthm,mathtools,tensor}

\usepackage{todonotes}
\usepackage[stable]{footmisc}
\usepackage{amsthm}
\usepackage{amsmath}
\usepackage{amssymb}
\usepackage{dsfont}
\usepackage{amsfonts}
\usepackage{graphicx}
\usepackage{color}
\usepackage[bf,SL,BF]{subfigure}
\usepackage{url}
\usepackage{epstopdf}
\usepackage{pdfsync}
\usepackage[colorlinks]{hyperref}
\usepackage[all]{xypic}
\usepackage{comment}
\usepackage{mathabx}
\usepackage{upgreek}

\hypersetup{
    linkcolor=blue,
}

\newlength{\fixboxwidth}
\usepackage{epsfig,amsbsy,graphicx,multirow}

\usepackage{ algorithm, algorithmic}

\renewcommand{\algorithmiccomment}[1]{\bgroup\hfill//~#1\egroup}

\usepackage[all]{xy}
\usepackage{accents}

 \numberwithin{equation}{section}

\def\P{\mathbb{P}}
\def\R{\mathbb{R}}

\def\E{\mathbb{E}}

\def\X{{\bf\mathcal{X}}}

\def\L{\mathcal{L}}

\def\restrict#1{\raise-.5ex\hbox{\ensuremath|}_{#1}}

\def\<{\big\langle}
\def\>{\big\rangle}

\def\Var{\operatorname{Var}}

\def\diam{\operatorname{diam}}

\def\argmin{{\operatorname{argmin}}}

\def\supp{{\operatorname{support}}}

\def\dim{{\operatorname{dim}}}

\def\Center{\operatorname{Center}}

\def\MINIBALL{\operatorname{MINIBALL}}

\definecolor{red}{rgb}{0.9, 0, 0}

\newtheorem{Theorem}{Theorem}[section]

\newtheorem{Remark}[Theorem]{Remark}

\newtheorem{Problem}{Problem}

\graphicspath{{images/}}

\begin{document}
\title[UQ of the 4th kind]{Uncertainty Quantification of the 4th kind; optimal  posterior accuracy-uncertainty tradeoff with the minimum enclosing ball}

\date{\today}

\author[\tiny Bajgiran]{Hamed Hamze Bajgiran}
\address{Hamed Hamze Bajgiran, Caltech,  MC 9-94, Pasadena, CA 91125, USA}
 \email{hhamzeyi@caltech.edu}

\author[Batlle]{Pau Batlle}
\address{Pau Batlle, Caltech, MC 305-16, Pasadena, CA 91125, USA}
 \email{pau@caltech.edu}

\author[Owhadi]{Houman Owhadi{$^*$}}
\address{Houman Owhadi, Caltech,  MC 9-94, Pasadena, CA 91125, USA}
 \email{owhadi@caltech.edu}
 
 \author[Samir]{Mostafa Samir}
 \address{Mostafa Samir, Beyond Limits}
  \email{mibrahim@beyond.ai}

\author[Scovel]{Clint Scovel}
\address{Clint Scovel, Caltech,   MC 9-94, Pasadena, CA 91125, USA}
 \email{clintscovel@gmail.com}

\author[Shirdel]{Mahdy Shirdel}
\address{Mahdy Shirdel, Beyond Limits, 400 N Brand Blvd, Glendale, CA 91203, USA}
 \email{mshirdel@beyond.ai}

\author[Stanley]{Michael Stanley}
\address{Michael Stanley, Carnegie Mellon, Baker Hall 132, Pittsburgh, PA 15213, USA}
 \email{mcstanle@andrew.cmu.edu}

\author[Tavallali]{Peyman Tavallali}
\address{Peyman Tavallali, Jet Propulsion Laboratory}
 \email{peyman.tavallali@jpl.caltech.edu}

\thanks{$^\dagger$Author list in alphabetical order.\\
$\quad$ {$^*$}Corresponding author: \href{mailto:owhadi@caltech.edu}{owhadi@caltech.edu}  }

\maketitle

\begin{abstract}
Uncertainty quantification (UQ) is, broadly, the task of determining appropriate uncertainties to model predictions. There are essentially three kinds of approaches to Uncertainty Quantification:  (A) robust optimization (min and max), (B) Bayesian (conditional average) and (C) decision theory (minmax). Although (A) is robust, it is unfavorable with respect to accuracy and data assimilation. (B) requires a prior, it is generally non-robust (brittle) with respect to the choice of that prior and posterior estimations can be slow. Although (C) leads to the identification of an optimal prior, its approximation suffers from the curse of dimensionality and the notion of loss/risk used to identify the prior is one that is averaged with respect to the distribution of the data. We introduce a 4th kind which is a hybrid between (A), (B), (C), and hypothesis testing. It can be summarized as, after observing a sample $x$,   (1) defining a likelihood region through the relative likelihood  and (2) playing a minmax game in that region to define optimal estimators and their risk.
The resulting method has several desirable properties: (a) an optimal prior is identified after measuring the data and the notion of loss/risk is a  posterior one, (b) the determination of the optimal estimate and its risk can be reduced to computing the  minimum enclosing ball of the image of the likelihood region under the quantity of interest map (such computations are fast and do not suffer from the curse of dimensionality).
The method is characterized by a  parameter in $ [0,1]$  acting as an assumed lower bound on the rarity of the observed data (the relative likelihood). When that parameter is near $1$, the method produces a posterior distribution concentrated around a maximum likelihood estimate (MLE)   with tight but low confidence UQ estimates. When that parameter is near $0$, the method produces a maximal risk posterior distribution with high confidence UQ estimates. In addition to navigating the accuracy-uncertainty tradeoff, the proposed method addresses the brittleness of Bayesian inference by navigating the robustness-accuracy tradeoff associated with data assimilation.
\end{abstract}

\section{Introduction}
\label{sec_new_intro}
The past century has seen a steady increase in the need of estimating and
 predicting complex systems and making (possibly critical) decisions with limited
information \cite{OwhadiScovelMachineWald}. These decisions are currently being formed based on increasingly complex models with imperfectly known parameters  estimated based on available (limited) data whose distribution depends on the unknown/imperfectly known parameters of the model (if the model is well specified, i.e., if the distribution of the data belongs to the parametric family of distributions represented by the model).
Making  decisions and assessing the risk of these decisions requires identifying  methods for data assimilation (estimating the parameters of the model based on data) and quantifying the risk/uncertainties of these decisions/parametric models.
Such UQ methods are not unique, and they essentially differ through assumptions made on the generation of the true parameter of the model.
In all inference/UQ methods, there is a tradeoff between robustness and accuracy \cite{owhadi2017qualitative}, and  these assumptions lead to the accuracy of the underlying method when they hold true but also to their lack of robustness when they do not hold true.
In this paper, we introduce a new and rigorous UQ method that navigates (in a Pareto optimal manner) this tradeoff between accuracy and robustness in data assimilation and UQ for parametric models.

\begin{figure}[htp]
    \centering
    \includegraphics[width = 0.7\textwidth]{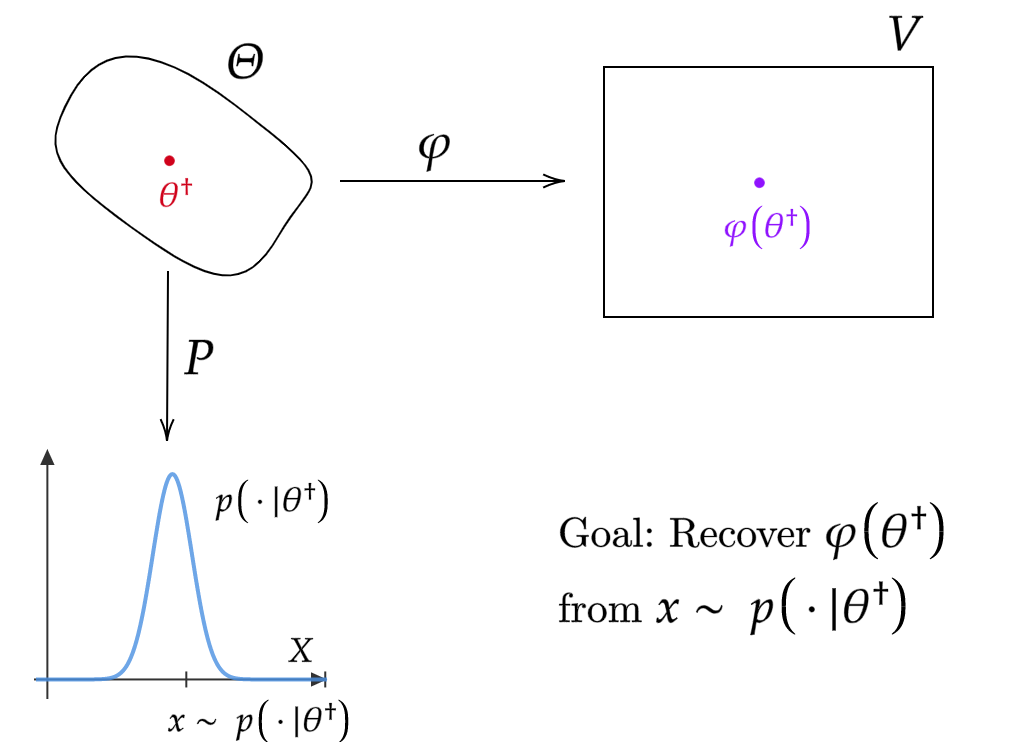}
    \caption{The Uncertainty Quantification (UQ) problem. Here, $\Theta$ is the space of parameters, $\theta^\dagger$ is the true unknown parameter, $\varphi$ a quantity of interest, and $P$ is the physical model determining the distribution $p$ from which the data $x$ is observed}
    \label{fig:UQSituation}
\end{figure}

\subsection{The problem}
To describe this method we formalize the underlying UQ problem as follows (see Fig.~\ref{fig:UQSituation}). Given a parameter space $\Theta$ and a quantity of interest $\varphi\,:\, \Theta \rightarrow V$
we seek to estimate  $\varphi(\theta^\dagger)$, where $\theta^\dagger\in \Theta$ is an unknown parameter, based on the observation of some data $x\in X$
sampled from a probability distribution $p(\cdot|\theta^\dagger)$ (given by our model $P$) depending on the unknown parameter $\theta^\dagger$.
Note that if our goal is to recover $\theta^\dagger$ itself, then we can let $\varphi$  be the identity function.
 A simple example (detailed in Sec.~\ref{sec_coin}) is to recover the probability $\theta^\dagger$ that a coin lands on heads, given the observation  $x=(x_1,\ldots,x_n)\in \{H,T\}^n$ of $n$ tosses of that coin.  Note that this general setup combines parametric uncertainty ($\theta^\dagger$ is unknown) with aleatoric uncertainty (the $x$ data is a sample from a random variable whose distribution depends on $\theta^\dagger$), and  they need be merged to estimate  $\varphi(\theta^\dagger)$ and quantify the uncertainty/risk of the estimation.

\subsection{The three main approaches to UQ}
There are currently three main approaches (detailed in Sec.~\ref{sec_intro}) to addressing this UQ problem. The {\bf worst case} (robust optimization) approach is (if $\varphi$ is real-valued) to compute,
 the minimum and maximum possible value of $\varphi(\theta)$ over all possible values the parameter $\theta\in \Theta$.
Although the data may be incorporated through empirical distribution inequalities \cite{owhadiscovelexball2015},
the worst case  approach is  conservative and, due to its lack of assumptions on the generation of $\theta^\dagger$, it is at the robust end of the tradeoff between accuracy and robustness. Indeed this approach is simply based on the observation that
\begin{equation}
\varphi(\theta^\dagger)\in \big[\min_{\theta\in \Theta}\varphi(\theta), \max_{\theta\in \Theta}\varphi(\theta)\big]\,.
\end{equation}

The {\bf Bayesian approach} is to assume that $\theta^\dagger$ is a sample from a prior distribution $\pi$ on $\Theta$, then  estimate $\varphi(\theta^\dagger)$ and quantify the uncertainty of that estimation by computing the posterior distribution of $\theta^\dagger$ given the data $x$.
Writing $d(x)$ for the estimation of $\varphi(\theta^\dagger)$ ($d\,:\, X\rightarrow V$), the Bayesian decision theoretic variant of the Bayesian approach is to introduce a loss/cost
\begin{equation}
\mathcal{\L}(\theta,d)=\E_{x\sim P(\cdot|\theta)}\E\big[\|\varphi(\theta)-d(x)\|^2\big]
\end{equation}
 for the choice of the estimator $d$ if the true value of unknown parameter is $\theta$, assume that $\theta^\dagger$ is sampled from a known prior distribution $\pi$ and identify an optimal estimator $d_\pi$ as a minimizer
 \begin{equation}\label{eqkjedjdhed}
d_\pi=\text{argmin}_{d}\E_{\theta\sim \pi}\big[\mathcal{\L}(\theta,d)\big]\,,
\end{equation}
   of the $\pi$-averaged loss $\E_{\theta\sim \pi}\big[\mathcal{\L}(\theta,d)\big]$, whose value at the minimum defines the risk of that estimator.
Due to the strength of the assumption that $\theta^\dagger$ is sampled from a known prior distribution, the Bayesian approach is at the accurate end of the tradeoff between accuracy and robustness, in particular, it is brittle to the choice of prior \cite{owhadi2017qualitative,owhadi2015brittleness, owhadi2015brittleness2,owhadi2016brittleness}.
The {\bf game/decision theoretic} approach formulates the underlying UQ problem as a zero-sum game in which $\theta$ is chosen by an adversarial player (Player I) seeking to maximize the loss $\mathcal{\L}(\theta,d)$ and $d$ is chosen by Player II seeking to minimize that loss.
As in classical game theory \cite{VNeumann28}, identifying a Nash equilibrium requires lifting this game by letting Player I randomize the selection of $\theta$ according to some mixed strategy/prior distribution $\pi$ on $\Theta$ and considering the average loss,
\begin{equation}\label{eqgame}
\mathcal{\L}(\pi,d)=\E_{\theta\sim \pi, x\sim P(\cdot|\theta)}\E\big[\|\varphi(\theta)-d(x)\|^2\big]\,, \quad \pi \in \mathcal{P}(\Theta)\,, d\,:\, X\rightarrow V\,.
\end{equation}
A saddle point $(\pi^*,d_{\pi^*})$ for \eqref{eqgame} is then identified by letting $d_\pi$ be the best Bayesian response \eqref{eqkjedjdhed} to $\pi$ and $\pi^*$ be a maximizer of the average-loss
$ \E_{\theta\sim \pi}\big[\mathcal{\L}(\theta,d_\pi)\big]$, i.e.,
\begin{equation}
\pi^* \in \text{argmax}_{\theta \in \Theta}\E_{\theta\sim \pi}\big[\mathcal{\L}(\theta,d_\pi)\big]\,.
\end{equation}
Although
 this approach achieves a balance in the accuracy/robustness tradeoff by relaxing the assumption 
that  $\theta$  is sampled from a known distribution, it does not explicitly  enable a navigation of that tradeoff. Furthermore, (1) the
  numerical approximation of an optimal mixed strategy for Player II suffers from the curse of dimensionality, and (2) $d_\pi$ is the best response to an data-averaged notion of risk rather than a data-given notion of risk.

\begin{figure}[h]
    \centering
    \includegraphics[width= \textwidth]{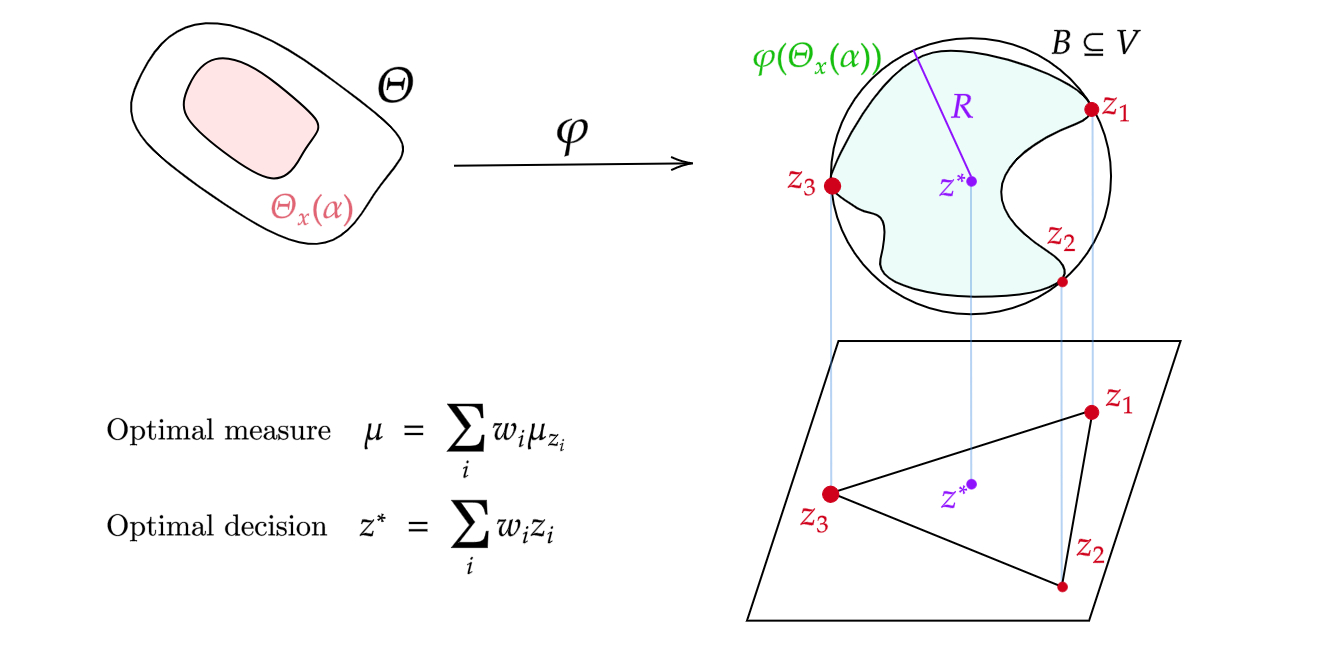}
    \caption{Example
    of the minimum enclosing ball $B$ about the image $\varphi(\Theta_{x}(\alpha))$
 (in green) with radius $R$ and center $d=z^{*}$. An optimal discrete measure $\mu:=\sum{w_{i}\updelta_{z_{i}}}$ ($z_i=\varphi(\theta_i)$) on the range of $\varphi$ for the maximum variance problem  is characterized by the fact that it is supported on
 the intersection of $\varphi(\Theta_{x}(\alpha))$ and $\partial B$ and $d=z^{*}=\sum{w_{i}z_{i}}$ is
 the center of mass of the measure $\mu$. The size of the solid red balls indicate the size of the  corresponding weights $w_{i}$.}
    \label{figminball}
\end{figure}

\subsection{Our new approach to UQ}
In this paper, we present a {\bf new approach} that does not suffer from weaknesses present in previous UQ methods such as brittleness and curse of dimensionality and that
explicitly navigates the tradeoff between accuracy and robustness in the estimation of the quantity of interest.
Motivated by the fact that the main cause of brittleness in inference is the possible rarity of the observed data \cite{owhadi2017qualitative,owhadi2015brittleness, owhadi2015brittleness2,owhadi2016brittleness}, the first step of this approach  is to make the hypothesis that the parameter $\theta$ that has generated the data is such that the data is not rare and
bound the probability that this hypothesis is false.
To describe this, given the observation $x$, for $\alpha\in [0,1]$ let $\Theta_x(\alpha)$ be the set of parameters $\theta\in \Theta$ whose relative likelihood
\begin{equation}\label{eqrellik}
\bar{p}(x|\theta):=\frac{p(x|\theta)}{\sup_{\theta'}p(x|\theta')}\,
\end{equation}
exceeds the threshold $\alpha$, i.e.,
\begin{equation}\label{eqnormrarset}
 \Theta_{x}(\alpha):= \bigl\{\theta \in \Theta: \bar{p}(x|\theta)\geq
   \alpha\bigr\}\,
   \end{equation}
 and let $ \beta_{\alpha}$ be the maximum (over $\theta\in \Theta$) probability that $\theta$ does not belong to $\Theta_x(\alpha)$ when $x$ is randomized according to the model $p(\cdot|\theta)$, i.e.,
\begin{equation}
\label{betadef0}
    \beta_{\alpha}:= \sup_{\theta \in \Theta}
P\Big(\bigl\{x' \in X: \theta \notin \Theta_{x'}(\alpha)\bigr\}\Big|\theta\Big)\,.
\end{equation}
$\beta_{\alpha}$ is interpreted as the significan/p-value of the hypothesis that $\theta^\dagger \in \Theta_x(\alpha)$. In particular,
for $\alpha$ close to one $\Theta_x(\alpha)$ concentrates around the Maximum Likelihood Estimators of $\theta^\dagger$ and the probability $ \beta_{\alpha}$ that the hypothesis is true goes to zero (which corresponds to accurate side of the tradeoff between accuracy and robustness).
For $\alpha$ close to zero, $\Theta_x(\alpha)$  stretches over the whole set $\Theta$, and the probability $ \beta_{\alpha}$ that the hypothesis is true goes to one (which corresponds to the robust side of the tradeoff).
The next step of this approach is to employ the game/decision theoretic approach with $\Theta$ replaced by the smaller set $\Theta_x(\alpha)$, i.e., replace \eqref{eqgame} with
\begin{equation}\label{eqgame2}
\mathcal{\L}(\pi,d)=\E_{\theta\sim \pi, x\sim P(\cdot|\theta)}\E\big[\|\varphi(\theta)-d(x)\|^2\big]\,, \quad \pi \in \mathcal{P}(\Theta_x(\alpha))\,, d\,:\, X\rightarrow V\,,
\end{equation}
 compute a saddle point $(\pi^{\alpha},d^\alpha)$ for \eqref{eqgame2} ($d^\alpha=d_{\pi^\alpha}$),  identify the optimal estimator as $d^\alpha$ and its
 risk/uncertainty $\mathcal{R}(d^\alpha)$ as the value  of the game:
 \begin{equation}\label{Rstar0}
 \mathcal{R}(d^\alpha):=\mathcal{\L}(\pi^\alpha,d^\alpha)\,.
 \end{equation}
\subsection{Main result}
One of our main results (Theorems \ref{thm_duala} and \ref{pioptreduced}) is that this optimal decision and its associated risk/uncertainty (defined as the value of the game at the Nash equilibrium)
can be identified as the center and the radius of the smallest ball enclosing the image of $\Theta_x(\alpha)$ under $\varphi$ (see Fig.~\ref{figminball}).
Furthermore, we present rigorous and practical algorithms (Algorithms \ref{Alg: ODAD General} and \ref{Alg: Miniball})\footnote{Python implementation for these algorithms can be found in \url{https://github.com/JPLMLIA/UQ4K}} with approximation accuracy guarantees
 for computing
that minimum enclosing ball based on the observation that optimal mixed strategies (priors) $\pi$ for Player I can be restricted to be supported at a maximum of $\dim(V)+1$ points located on the boundary of that ball.

\subsection{Coin toss}
\label{sec_coin}
At the cost of some forward referencing, we will now describe an application of our proposed problem to the estimation of the probability that a coin lands on heads based on the observation of $n$ independent tosses of that coin.

\subsubsection{$n$ tosses of a single coin.}
In this example, we estimate the probability that a biased coin lands on heads from the observation of $n$ independent tosses of that coin.
Specifically, we consider flipping a coin $Y$  which has an unknown probability $\theta^{\dagger}$ of coming heads $(Y=1)$ and
probability $1-\theta^{\dagger}$ coming up tails $(Y=0)$.
Here $\Theta:=[0,1]$, $X=\{0,1\}$, and the model $P:\Theta \rightarrow \mathcal{P}(\{0,1\})$ is
$P(Y=1|\theta)= \theta$ and
$P(Y=0|\theta)= 1-\theta$.
We toss the coin $n$ times generating a sequence
of i.i.d.~Bernoulli variables $(Y_{1},\ldots, Y_{n})$ all with the same unknown parameter
 $\theta^{\dagger} \in [0,1]$,
and let $x:=(x_{1},\ldots, x_{n})\in \{0,1\}^{n}$ denote the outcome of the experiment.
Let $h=\sum_{i=1}^{n}{x_{i}}$ denote the number of heads observed and $t = n - h$
 the number of tails.
Then the model for the $n$-fold toss is
\begin{equation}
\label{tossmodel0}P(x|\theta)=\prod_{i=1}^{n}{\theta^{x_{i}}(1-\theta)^{1-x_{i}}}=\theta^{h}
(1-\theta)^{t}
\end{equation}
and, given an observation $x$, the MLE is $\theta = \frac{h}{n}$ so that the relative likelihood \eqref{eqrellik} is\footnote{Although the fact that $\bar{p}(x|0)=\bar{p}(x|1)=0$ violates our positivity assumptions (described in Sec.~\ref{sec_intro}) on the model in our framework, in this case this technical restriction can be removed, so we can still use this example as an illustration.}
\begin{equation}
\label{tossmodel}
\bar{p}(x|\theta)=\frac{\theta^{h}
(1-\theta)^{t}}{\bigl(\frac{h}{n}\bigr)^{h}
\bigl(\frac{t}{n}\bigr)^{t}}  \, .
\end{equation}
We seek to estimate $\theta$, so let $V=\R$ and let the quantity of interest $\varphi:\Theta \rightarrow V$  be the identity function $\varphi(\theta)=\theta$. In this case, given $\alpha \in [0,1]$, the likelihood region
\begin{equation}
\Theta_{x}(\alpha) = \Bigl\{\theta \in [0,1]: \frac{\theta^{h}
(1-\theta)^{t}}{\bigl(\frac{h}{n}\bigr)^{h}
\bigl(\frac{t}{n}\bigr)^{t}} \geq
   \alpha\Bigr\}\,
\end{equation}
 constrains the support of priors to points with relative likelihood larger than $\alpha$.
Using Theorem \ref{pioptreduced} with $m=dim(V)+1=2$, one can compute a saddle point $(\pi^\alpha,d^\alpha)$ of the game \eqref{eqgame2}
as
\begin{equation}
\pi^\alpha=w \updelta_{\theta_{1}}+(1-w) \updelta_{\theta_{2}}\text{ and } d^\alpha=w\theta_{1}+(1-w) \theta_{2}
\end{equation}
where $w,  \theta_{1}, \theta_2$  maximize the variance
\begin{equation}
\begin{cases}
    \text{Maximize}&\quad w \theta_1^2 + (1 - w) \theta_2^2 - \left( w \theta_1 + (1 - w) \theta_2 \right)^2 \\
    \text{over}& \quad 0 \leq w \leq 1,\quad \theta_1, \theta_2 \in [0, 1] \\
    \text{subject to}& \quad \frac{\theta_i^{h}
(1-\theta_i)^{t}}{\bigl(\frac{h}{n}\bigr)^{h}
\bigl(\frac{t}{n}\bigr)^{t}} \geq
   \alpha , \qquad i = 1, 2\, .
\end{cases}
\end{equation}

\begin{figure}[h]
    \centering
    \includegraphics[width = 0.8\textwidth]{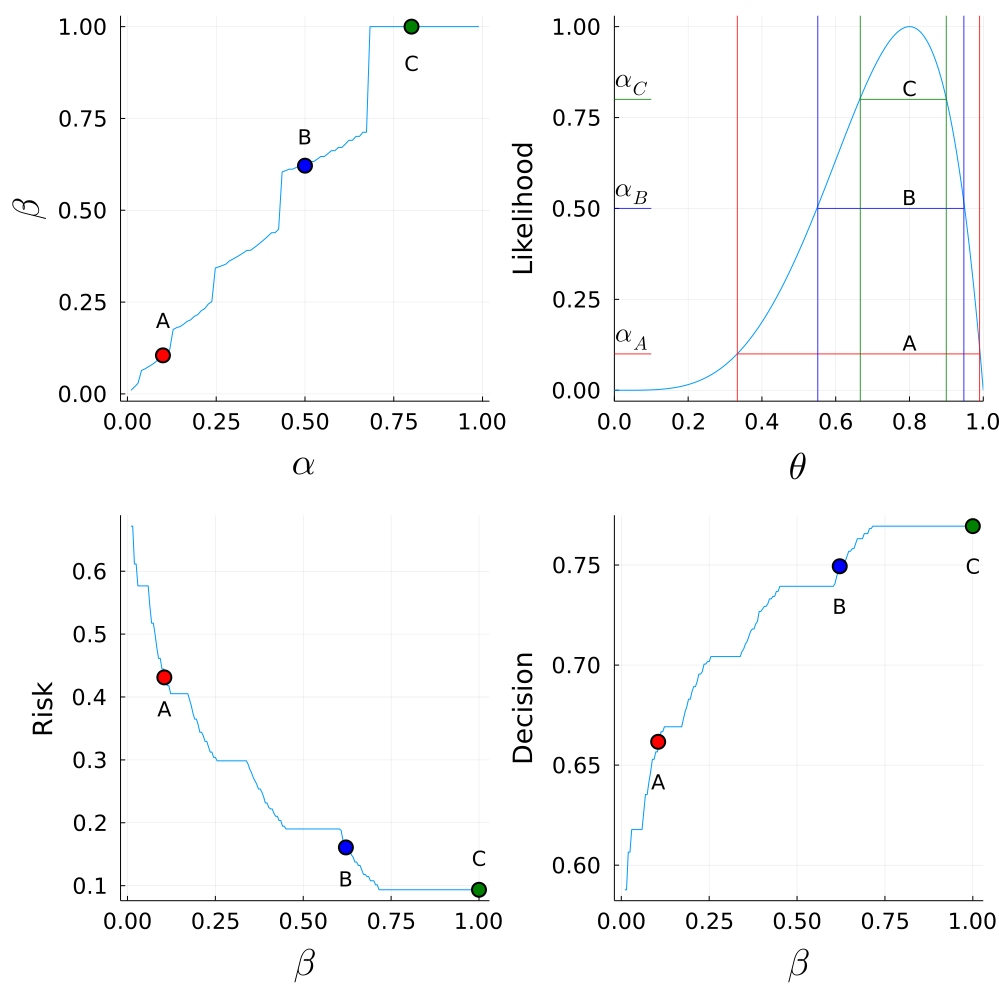}
    \caption{$\alpha - \beta$ relation, likelihood level sets, risk value and decision for different choices of $\alpha$ (and consequently $\beta$) for the 1 coin problem after observing 4 heads and 1 tails. Three different values in the $\alpha-\beta$ curve are highlighted across the plots}
    \label{fig:coins1d}
\end{figure}

 Equation \eqref{betadef0} allows us to compute $\beta \in [0,1]$ as a function of $\alpha\in [0,1]$.
The solution of the optimization problem can be found by finding the minimum enclosing ball of the set $\Theta_x(\alpha)$, which in this  $1$-D case is also subinterval of the interval $[0,1]$.
 For
  $n=5$ tosses resulting in   $h = 4$ heads and $t = 1$ tails, Figure \ref{fig:coins1d} plots
 (1) $\beta$, the relative likelihood, its level sets and minimum enclosing balls as a function of  $\alpha$, and (2) The risk $\mathcal{R}(d^\alpha)= $\eqref{Rstar0}  and optimal decision $d^\alpha$ as a function of $\beta$. Three different points in the $\alpha - \beta$ curve are highlighted. Note that as $\alpha$ goes from $0$ to $1$, the relative likelihood region $\Theta_x(\alpha)$ gets smaller (it shrinks towards the MLE),
 the optimal estimator $d^\alpha$ goes from the center of the worst case interval to the MLE estimate, the risk (variance) of the estimator shrinks
 (which corresponds to an increase in accuracy),
  but the confidence $1-\beta(\alpha)$ in that risk (the probability $\beta(\alpha)$ that $\theta^\dagger\in \Theta_x(\alpha)$) also shrinks towards zero (which corresponds to a loss of robustness).

\begin{figure}[h]
    \centering
    \includegraphics[width = 0.95\textwidth]{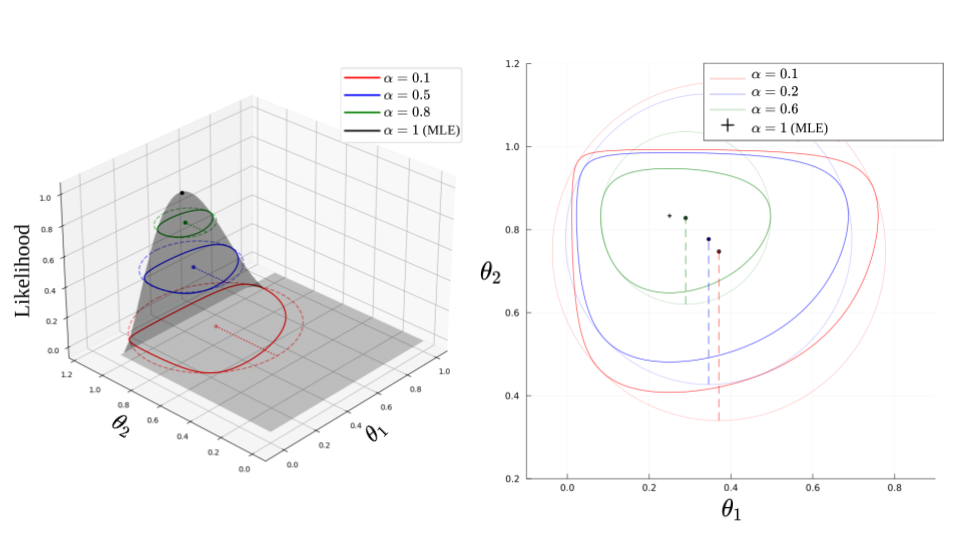}
    \caption{2D likelihood level sets and minimum enclosing balls for different values of $\alpha$, visualized as level sets of the likelihood function (left) and projected onto a 2D plane (right)}
    \label{fig:2dcoinviews}
\end{figure}

\subsubsection{$n_1$ and $n_2$ tosses of two  coins.}

We now consider the same problem with two independent coins with unknown probabilities $\theta^\dagger_1, \theta^\dagger_2$.  After tossing each coin $i$ $n_i$   times, the observation $x$ consists of
   $h_i$ heads and $t_i$ tails for each $i$,
  produce a 2D relative likelihood function on $\Theta = [0,1]^2$ given by
\begin{equation}
    \bar{p}(x|\theta_1, \theta_2) = \frac{\theta_1^{h_1}
(1-\theta_1)^{t_1}}{\bigl(\frac{h_1}{n_1}\bigr)^{h_1}
\bigl(\frac{t_1}{n_1}\bigr)^{t_1}} \frac{\theta_2^{h_2}
(1-\theta_2)^{t_2}}{\bigl(\frac{h_2}{n_2}\bigr)^{h_2}
\bigl(\frac{t_2}{n_2}\bigr)^{t_2}}\, .
\end{equation}

Figure \ref{fig:2dcoinviews} illustrates the level sets $\bar{p}(x|\theta_1, \theta_2) \geq \alpha$ and their corresponding bounding balls for $h_1 = 1, t_1 = 3, h_2 = 5, t_2 = 1$ and different values of $\alpha \in [0,1]$.

\subsection{Structure of the paper}
This article is organized as follows: In Sec.~\ref{sec_intro}, we formalize the UQ problem and review the three previous approaches to the problem, emphasizing the limitations addressed with our method. In Sec.~\ref{sec:UQ4k}, we introduce a new kind of uncertainty quantification based on the minimum enclosing ball, and in Sections \ref{sec:compfram} and \ref{sec_miniball}  we introduce the computational framework and our minimum enclosing ball algorithms. Sec.~\ref{sec:example} presents numerical illustrations of the efficacy and scope of our approach. Sec.~\ref{sec_generalloss} generalizes the loss and rarity assumptions.
Sec.~\ref{sec:thmproofs} presents supporting theorems and proofs.

\section{Previous approaches to UQ}
\label{sec_intro}
We begin by formalizing the UQ problem introduced in the previous section. Let $\varphi:\Theta \rightarrow V$ be a quantity of interest, where $V$ (the space of predictions) is a finite-dimensional vector space and
$\Theta$ (the space of parameters) is a compact set.
Let $X$ (the space of data) be a measurable space and write $\mathcal{P}(X)$ for the set of probability distributions on $X$.
Consider a  model $P:\Theta \rightarrow \mathcal{P}(X)$ representing the dependence of the distribution of a data point
$x\sim P(\cdot|\theta)$ on the value of the parameter $\theta\in \Theta$. Throughout, we use $\|\cdot\|$ to denote the Euclidean norm.
We are then interested in solving the following problem.

\begin{Problem}\label{pb1}
Let $\theta^\dagger$ be an unknown element of $\Theta$.
Given an observation $x\sim P(\cdot|\theta^\dagger)$ of data,  estimate $\varphi(\theta^\dagger)$ and quantify the  uncertainty (accuracy/risk) of the estimate.
\end{Problem}
We assume that we can write a probability density function for any $P(\cdot|\theta)$. More formally, we assume that $P$ is a dominated model with positive densities, that is,
 for each $\theta \in \Theta$, $P(\cdot|\theta)$ is
  defined by a (strictly) positive density $p(\cdot|\theta):X \rightarrow \R_{>0}$
with respect to a measure $\nu \in \mathcal{P}(X)$,
such that, for each  measurable subset $A$ of $X$,
\begin{equation}
P(A|\theta)=\int_{A}p(x'|\theta)d\nu(x'), \quad \theta \in \Theta.
\end{equation}
\subsection{The three main  approaches to UQ}\label{secthreemain}
Problem \ref{pb1} is a fundamental Uncertainty Quantification (UQ) problem, and there are essentially three main approaches
for solving it. We now describe them when $V$ is a Euclidean space with the $\ell^{2}$ loss function.

\subsubsection{Worst-case}
In a different setting, essentially where the set $\Theta$ consists of probability measures,
the OUQ framework \cite{owhadi2013optimal} provides a worst-case analysis for providing rigorous uncertainty bounds.  In the setting of this paper,
in the absence of data (or ignoring the data $x$), the  (vanilla) worst-case (or robust optimization) answer is
to estimate $\varphi(\theta^\dagger)$ with the minimizer $d^* \in V$ of the worst-case error 
\begin{equation}
\label{worstcase}
\mathcal{R}(d):= \max_{\theta \in \Theta}{\bigl[\|\varphi(\theta)-d\|^2\bigr]}\,.
\end{equation}
In that approach, $(d^*,\mathcal{R}(d^*))$ are therefore identified as the center and squared radius of the minimum enclosing ball of $\varphi(\Theta)$.

\subsubsection{Bayesian}
 The  (vanilla) Bayesian (decision  theory) approach (see e.g.~Berger \cite[Sec.~4.4]{berger2013statistical}) is to assume that $\theta$ is sampled from a prior distribution $\pi \in \mathcal{P}(\Theta)$, and approximate $\varphi(\theta^\dagger)$ with the minimizer $d_\pi(x)\in V$ of the Bayesian posterior risk 
\begin{equation}
\label{risk_bayes}
\mathcal{R}_{\pi}(d):= \E_{\theta \sim \pi_{x}}{\bigl[\|\varphi(\theta)-d\|^2\bigr]}, \quad d \in V,
\end{equation}
associated with the decision $d \in V$,
where
\begin{equation}
\label{eq_conditional0}
  \pi_{x}:=\frac{p(x|\cdot)\pi}{\int_{\Theta}{p(x|\theta)d\pi(\theta)}}
\end{equation}
is the posterior measure
determined by the likelihood $p(x|\cdot)$, the prior $\pi$ and the observation $x$.
The minimizer $d_\pi(x)$ of \eqref{risk_bayes} is the posterior distribution mean
 \begin{equation}\label{eqkjhedegdudyg}
 d_\pi(x):=\E_{\theta \sim \pi_{x}}[\varphi(\theta)]\,
 \end{equation}
  and  the uncertainty is quantified by the posterior variance
\begin{equation}
\label{postvarrisk_bayes}
\mathcal{R}_{\pi}(d_\pi(x)):= \E_{\theta \sim \pi_{x}}\big[ \|\varphi(\theta)-d_\pi(x)\|^2\big]\,.
\end{equation}

\subsubsection{Game/decision theoretic}
 The Wald's game/decision theoretic approach is to  consider  a two-player zero-sum game where player I selects   $\theta \in \Theta$, and player II selects a decision function $d:X \rightarrow V$ which estimates the quantity of interest
 $\varphi(\theta)$ (given the data $x\in X$), resulting in the loss
\begin{equation}
\label{LalphadefW0l2}
 \L(\theta,d):=\E_{x \sim P(\cdot|\theta)}\bigl[\|\varphi(\theta)-d(x)\|^2\bigr], \qquad    \theta \in \Theta,\,\,
   d:X \rightarrow  V  ,
\end{equation}
for player II. Such a game will normally not have a saddle point, so following von Neumann's approach \cite{VNeumann28}, one randomizes both players' plays to identify a Nash equilibrium. To that end, first observe that, for the quadratic loss considered here (for ease of presentation), because of the convexity of the loss in $d$, only the choice of player I needs to be randomized. Letting
$\pi \in \mathcal{P}(\Theta)$ be a probability measure randomizing the play of player I, we consider the lift
\begin{equation}
\label{LalphadefW0lifted}
 \L(\pi,d):=\E_{\theta \sim \pi}\E_{x \sim P(\cdot|\theta)}\bigl[\|\varphi(\theta)-d(x)\|^2\bigr], \qquad    \pi \in \mathcal{P}(\Theta),\,\,
   d:X \rightarrow  V,
\end{equation}
of the game \eqref{LalphadefW0l2}.
A minmax optimal estimate of $\varphi(\theta^\dagger)$ is then obtained by identifying a Nash equilibrium (a saddle point) for \eqref{LalphadefW0lifted}, i.e. $\pi^* \in \mathcal{P}(\Theta)$ and $d^*:X \rightarrow  V$ satisfying
\begin{equation}
 \L(\pi,d^*) \leq \L(\pi^{*},d^*) \leq  \L(\pi^{*},d), \qquad \pi \in \mathcal{P}(\Theta),\,\, d:X \rightarrow V.
 \end{equation}
Consequently, an optimal strategy of player II is then the posterior mean $d_{\pi^*}(x)$ of the form \eqref{eqkjhedegdudyg} determined by a  worst-case  measure
and optimal randomized/mixed strategy for player I
\begin{equation}
\label{worstprior}
\pi^*:\in \arg \max_{\pi\in \mathcal{P}(\Theta)} \E_{\theta \sim \pi, x\sim P(\cdot,\theta)}\big[ \|\varphi(\theta)-d_\pi(x)\|^2\big].
\end{equation}
To connect with the Bayesian framework we observe (by changing the order of integration) that the
 Wald's risk \eqref{LalphadefW0lifted} can be written as the average
 \begin{equation}
\label{LalphadefW4}
 \L(\pi,d):=\E_{x \sim X_{\pi} }\bigl[\mathcal{R}_{\pi}(d(x))\bigr]\,
\end{equation}
of the Bayesian decision risk $\mathcal{R}_{\pi}(d(x))$ ($=$\eqref{risk_bayes} for $d=d(x)$) determined by the prior $\pi$  and decision $d(x)$
with respect to the $X$-marginal distribution
\begin{equation}
\label{xmarg0}
 X_{\pi}:=\int_{\Theta}{P(\cdot|\theta)d\pi(\theta)}
\end{equation}
 associated with the prior $\pi$ and the model
$P$.
Therefore, the Wald framework identifies a worst-case prior \eqref{worstprior}, while the prior used in Bayesian decision theory is specified by the practitioner.

 \begin{figure}[htp]
    \centering
    \includegraphics[width= 0.7\textwidth]{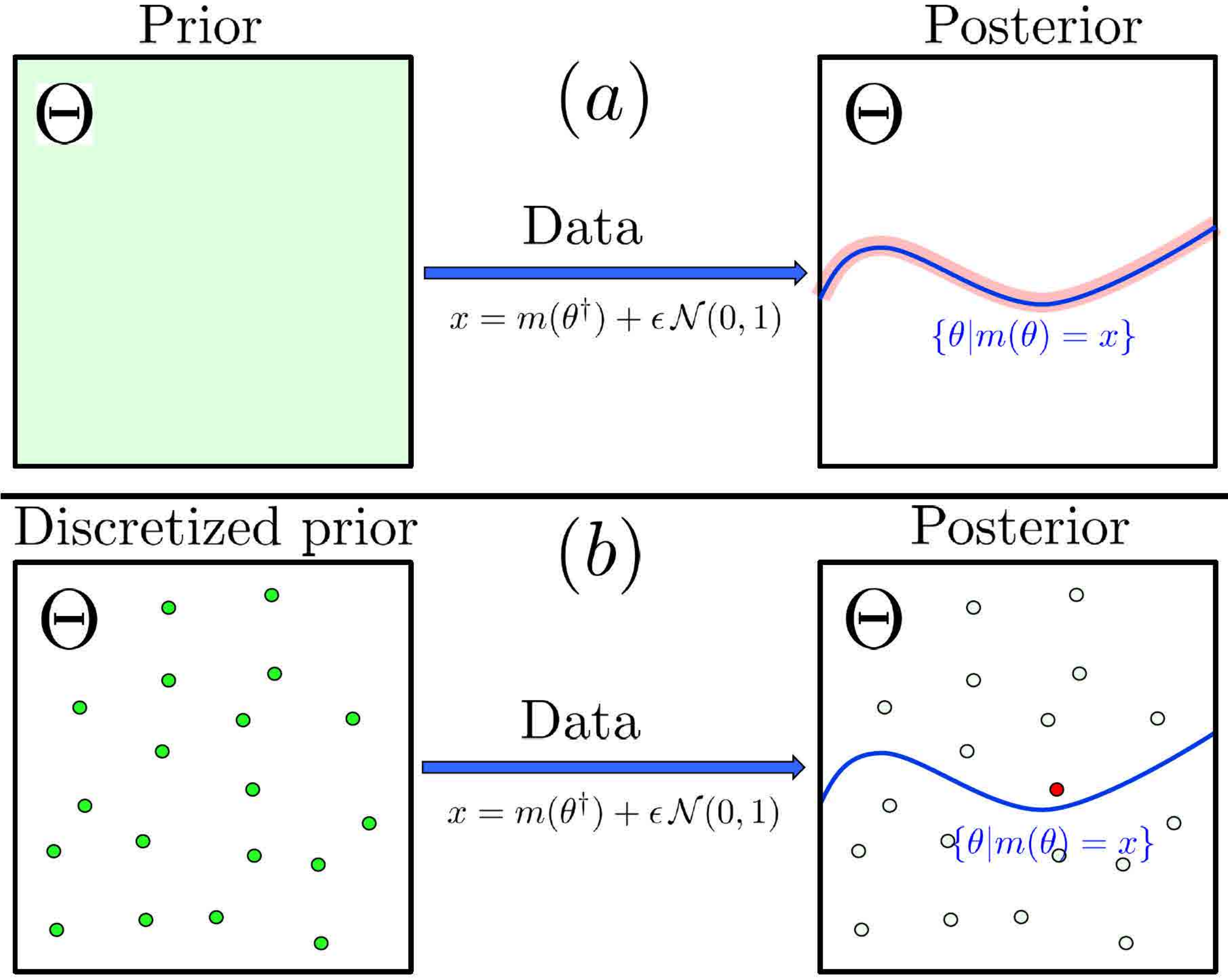}
    \caption{Curse of dimensionality in discretizing the prior. The data is of the form $x=m(\theta)+\epsilon \mathcal{N}(0,1)$ where   $m$ is  deterministic and   $\epsilon \mathcal{N}(0,1)$ is small  noise. (a) For the continuous prior, the posterior concentrates around  $\mathcal{M}:=\{\theta \in \Theta |m(\theta)=x\}$. (b) For the discretized prior, the posterior concentrates on the delta Dirac that is the closest to $\mathcal{M}$. }
    \label{figcurse}
\end{figure}

\subsection{Limitations of the three main approaches to UQ}
All three approaches described in Section \ref{secthreemain} have limitations in terms of accuracy, robustness, and computational complexity. Although the worst-case approach is robust,  it appears unfavorable in terms of accuracy and data assimilation.
The  Bayesian approach,  on the other hand, suffers from the computational complexity of estimating the posterior distribution
and from brittleness \cite{owhadi2015brittleness} with respect
to the choice of prior along with Stark's admonition \cite{Stark}
``your prior can bite you on the posterior.''
Although Kempthorne \cite{kempthorne1987numerical} develops a rigorous numerical procedure with convergence guarantees for solving the equations of Wald's statistical decision theory which
  appears  amenable to computational complexity analysis,
 it suffers from the curse of dimensionality (see Fig.~\ref{figcurse}).
 This can be understood from the fact that the risk associated with
the worst-case measure  in the Wald framework is an average over the observational
 variable $x \in X$ of the conditional risk, conditioned on the observation
 $x$. Consequently, for a discrete approximation of a worst-case measure,
 after an observation is made, there may be insufficient mass near the places where the conditioning will provide a good estimate of the appropriate conditional measure.
 Indeed,
in the proposal \cite{OwhadiScovelMachineWald} to develop Wald's statistical decision theory along the lines of Machine Learning, with its dual focus on performance and computation, it was observed that
\begin{quote}
``Although Wald’s theory of Optimal Statistical Decisions has resulted in many important statistical discoveries, looking through the three
Lehmann symposia of Rojo and P{\'e}rez-Abreu \cite{Rojo:optimality1} in 2004, and Rojo \cite{Rojo:optimality2,Rojo:optimality3} in 2006
and 2009, it is clear that the incorporation of the analysis of the computational algorithm, both in terms of its computational efficiency and its statistical optimality, has
not begun.''
\end{quote}
 Moreover,  one might ask why, after seeing the data, one is choosing a worst-case measure
which optimizes the average \eqref{LalphadefW4} of the Bayesian risk \eqref{postvarrisk_bayes}, instead of choosing it to optimize the value of the risk $\mathcal{R}_{\pi}(d_\pi(x))$ at the value of the observation $x$. It is therefore desirable for an approach to UQ to successfully assimilate the observed data, to avoid requiring having to manually select a prior and to have a data-dependent notion of risk. In Section \ref{sec:UQ4k}, we will propose a framework with all these properties. A comparison of the properties of all the mentioned methods can be found in Table \ref{table:methods}.
\begin{table}[htp]
\centering
\small
\begin{tabular}{cc|c|cl}
                & \begin{tabular}[c]{@{}c@{}}Makes use of\\ the observed data \end{tabular} & \begin{tabular}[c]{@{}c@{}}No need to\\ manually specify prior\end{tabular} & \begin{tabular}[c]{@{}c@{}} Risk depends\\ on the observed data \end{tabular}&  \\ \cline{2-4}
Worst case      & ×                & \checkmark                                                                           & ×                        &  \\
Bayesian        & \checkmark                & ×                                                                           & \checkmark                        &  \\
Decision Theory & \checkmark                & \checkmark                                                                           & ×                        &  \\
UQ4K (Section \ref{sec:UQ4k}) & \checkmark                & \checkmark                                                                           & \checkmark                        &
\end{tabular}
\caption{Comparison of the three previous approaches to uncertainty quantification with our proposed method in Section \ref{sec:UQ4k}}
\label{table:methods}
\end{table}

\section{Uncertainty Quantification of the 4th Kind}
\label{sec:UQ4k}
\subsection{Basic definitions}

In this paper, we introduce a framework which is a hybrid between Wald's statistical decision theory \cite{Wald1950},  Bayesian decision theory \cite[Sec.~4.4]{berger2013statistical}, robust optimization and hypothesis testing. Here we describe its components for simplicity
when the loss function is the $\ell^{2}$ loss. Later in Section \ref{sec_generalloss} we develop the framework for general loss functions.

\subsubsection{Rarity assumption on the data}
In \cite[Pg.~576]{owhadi2015brittleness}  it was demonstrated that one could alleviate the brittleness of Bayesian inference (see  \cite{owhadi2015brittleness2,owhadi2016brittleness})  by restricting  to priors $\pi$ for which the  observed data $x$ is not rare, that is,
\begin{equation}
\label{Palpha1}
 p(x):= \int_{\Theta}{p(x|\theta)d\pi(\theta)} \geq \alpha\,
 \end{equation}
 according to the density of the $X$-marginal determined by $\pi$ and the model $P$,
for some $\alpha >0$.
In the  proposed framework,
we consider playing a game {\em after observing the data $x$}  whose loss function is defined by
the Bayesian decision risk
$\mathcal{R}_{\pi}(d)$ \eqref{risk_bayes}, where player I selects a prior $\pi$ subject to  a
  {\em rarity assumption} ($\pi \in \mathcal{P}_{x}(\alpha)$) and player II selects a decision $d \in V$. The rarity assumption considered here is
\begin{equation}
\label{rarityass0}  \mathcal{P}_{x}(\alpha):=
\Bigl\{ \pi \in \mathcal{P}(\Theta): \supp(\pi) \subset \bigl\{\theta \in \Theta: p(x|\theta)\geq
   \alpha\bigr\}\Bigr\}\, .
 \end{equation}
Since $p(x|\theta)\geq \alpha$ for all $\theta $ in the support of any $\pi \in \mathcal{P}_{x}(\alpha)$
it follows that  such a $\pi$ satisfies \eqref{Palpha1} and therefore is sufficient to prevent Bayesian brittleness.

\subsubsection{The relative likelihood for the rarity assumption}
Observe in \eqref{eq_conditional0} that the map from the prior $\pi$ to posterior $\pi_{x}$ is scale-invariant in the likelihood
$p(x|\cdot)$ and that the effects of
scaling the likelihood  in the rarity assumption
 can be undone by modifying $\alpha$.
 Consequently, we scale the likelihood function
\begin{equation}
\label{def_rellikelihood}
\bar{p}(x|\theta):=\frac{p(x|\theta)}{\sup_{\theta \in \Theta}{p(x|\theta)}}, \quad \theta \in \Theta,
\end{equation}
to its {\em relative likelihood} function
\begin{equation}
\label{def_rellikelihood2}
    \bar{p}(x|\cdot):\Theta \rightarrow (0,1]\,.
\end{equation}
According to Sprott  \cite[Sec.~2.4]{sprott2008statistical},
the relative likelihood  measures the
 plausibility of any parameter value $\theta$ relative to a maximum likely $\theta$ and
summarizes the information about $\theta$  contained in the sample $x$.
See Rossi \cite[p.~267]{rossi2018mathematical} for its large sample connection with the $\X^{2}_{1}$ distribution and several examples of the relationship between likelihood regions and confidence intervals.

For $x \in X$ and $\alpha \in [0,1]$, let \eqref{eqnormrarset}
denote the corresponding {\em likelihood region} and, updating
\eqref{rarityass0},
 redefine the rarity assumption by
\begin{equation}
\label{rarityass}  \mathcal{P}_{x}(\alpha):=
  \mathcal{P}\bigl(\Theta_{x}(\alpha)\bigr)\, .
 \end{equation}
That is, the rarity constraint $\mathcal{P}_{x}(\alpha)$
 constrains priors to have support on the likelihood region $\Theta_{x}(\alpha)$.
We will now define the confidence level of the family $\Theta_{x}(\alpha), x \in X$.

\subsubsection{Significance/confidence level}
\label{sec_significance}
For a given $\alpha$, let the {\em significance} $\beta_\alpha$ at the value $\alpha$ be the maximum (over $\theta \in \Theta$) of the probability that a data $x'\sim P(\cdot|\theta)$ does not satisfy the rarity assumption $\bar{p}(x'|\theta)\geq \alpha$, i.e.,
\begin{eqnarray}
\label{betadef}
    \beta_{\alpha}&:=& 
 \sup_{\theta \in \Theta}
{\int{\mathds{1}_{\{\bar{p}(\cdot|\theta) < \alpha\} }(x')p(x'|\theta)d\nu(x')}}\,,
\end{eqnarray}
where, for fixed $\theta$,
$ \mathds{1}_{\{\bar{p}(\cdot|\theta) < \alpha\}} $ is the indicator function
of the set $\{x' \in X: \bar{p}(x'|\theta) < \alpha\}$.
Observe that, in the setting of hypothesis testing, (1) $\beta_\alpha$ can be interpreted as the p-value associated with the  hypothesis that
the rarity assumption is not satisfied (i.e.~the  hypothesis that $\theta$ does not belongs to the  set \eqref{eqnormrarset}), and (2) $1-\beta_\alpha$ can be interpreted as the confidence level associated with
the rarity assumption (i.e.~the smallest probability that $\theta$  belongs to the  set \eqref{eqnormrarset}).
Therefore, to select  $\alpha \in [0,1]$, we set a {\em significance level} $\beta^{*}$ (e.g.~$\beta^{*}=0.05$) and choose
$\alpha$
to be the largest value such that the significance  at $\alpha$ satisfies $\beta_{\alpha} \leq \beta^{*}$.

\subsubsection{Connection to the Likelihood Ratio Test}
The connection of the likelihood region we are defining and confidence sets can be more explicitly seen via the likelihood ratio test and the inversion of that test to produce a confidence set. Namely, we define a hypothesis test,
\begin{equation} \label{eq:hyp_test}
    H_0: \theta = \theta_0 \quad \text{versus} \quad H_1: \theta \neq \theta_0, \quad \theta \in \Theta.
\end{equation}
The likelihood ratio test statistic is defined by,
\begin{equation} \label{eq:lr_test_stat}
    \Lambda(\theta_0, x) = \frac{p(x | \theta_0)}{\sup_{\theta \in \Theta} p(x|\theta)}.
\end{equation}
Note, $\Lambda(\theta_0, x) = \bar{p}(x|\theta)$, i.e., Equation~\eqref{def_rellikelihood} and Equation~\eqref{eq:lr_test_stat} are equivalent. In the hypothesis testing setting, one would next define a critical value, $c > 0$, such that when $\Lambda(\theta_0, x) \leq c$, the null in test~\eqref{eq:hyp_test} is rejected. Ideally, $c$ is chosen such that the probability of false rejection under the null hypothesis is capped at some probability, $\alpha \in (0, 1)$, defining an $\alpha$-level test. The type 1 error probability control is mathematically characterized as,
\begin{equation} \label{eq:type1_err_control}
     P \left(\Lambda(\theta_0, x) \leq c | \theta_0 \right) \leq \alpha.
\end{equation}
The error control criterion of Equation~\eqref{eq:type1_err_control} can be used to define an acceptance region in the sample space defined as follows,
\begin{equation} \label{eq:acceptance_region}
    A(\theta_0, c) = \{ x \in X : \Lambda(\theta_0, x) \geq c \}.
\end{equation}
The type 1 error control then implies $P(x \in A(\theta_0, c)|\theta_0) \geq 1 - \alpha$ and acts in a similar way to the control exerted by $\beta_\alpha$ in  Equation~\eqref{betadef}. In that equation, $\alpha$ is chosen as the supremum over $\theta \in \Theta$ of $P \left(\Lambda(\theta, x) \leq c | \theta \right)$, thus making $\alpha$ a constant independent of $\theta$. A viable alternative is to consider a curve $\alpha(\theta)$, where
\begin{equation}
         P \left(\Lambda(\theta_0, x) \leq c | \theta_0 \right) \leq \alpha(\theta_0) \quad \forall \; \theta_0,
\end{equation}
and then redefining \eqref{eqnormrarset} accordingly. In this work we consider the fixed $\alpha$ model for simplicity.

In some cases, the false rejection control can be done cleanly. For instance, in the context of a noise model where $x = \theta + \varepsilon, \; \varepsilon \sim N(0, I), \; \theta \in \mathbb{R}^n$, it can be shown that the log-likelihood ratio follows the distribution,
\begin{equation}
    -2 \log \Lambda(\theta_0, x) \sim \chi^2_n,
\end{equation}
i.e., the log-likelihood ratio is distributed as a chi-squared distribution with $n$ degrees of freedom, allowing $c$ to be exactly chosen. In the event the test statistic distribution cannot be exactly known, asymptotic results such as that shown in Theorem~\ref{thm_asympt} can provide similar results.

As discussed in Chapter 9 of Casella/Berger \cite{casella2003statistical}, one can think about inverting an $\alpha$-level hypothesis test such as Test~\eqref{eq:hyp_test} to obtain a $1 - \alpha$ confidence set $C(x) \subset \Theta$, such that $P \left(\theta^* \in C(x)|\theta^* \right) \geq 1 - \alpha$, where $\theta^*$ is the true parameter value. The inverting is performed with the acceptance region of Equation~\eqref{eq:acceptance_region} and the set is defined as follows,
\begin{equation} \label{eq:conf_set}
    C(x) = \{\theta_0 \in \Theta : x \in A(\theta_0, c)\},
\end{equation}
Note the equivalence between Equation~\eqref{eq:conf_set} directly above and Equation~\eqref{eqnormrarset} from the previous section. By the type 1 error control, we have,
\begin{equation}
    P(\theta^* \in C(x)|\theta^*) = P(x \in A(\theta^*, c)|\theta^*) \geq 1 - \alpha,
\end{equation}
implying,
\begin{equation}
    P(\theta^* \notin C(x)|\theta^*) \leq \alpha,
\end{equation}
providing an additional view of the equivalence with Equation~\eqref{betadef}. As such, the \textit{relative likelihood} and \textit{rarity condition} can be seen through the more traditional statistical lens of the \textit{likelihood ratio} and \textit{type 1 error control} in the classical hypothesis testing setting.

\begin{Remark}
\label{rmk_surrogate}
For models where the maximum of the likelihood function
\[M(x'):=\sup_{\theta \in \Theta}{p(x'|\theta)},\quad x' \in X,\] is expensive to compute but for which there exists an efficiently computable upper approximation $M'(x') \geq M(x'), \, x'\in X$ available, the surrogate
\begin{equation}
\label{def_rellikelihooda}
\bar{p}'(x'|\theta):=\frac{p(x'|\theta)}{M'(x')},\quad x' \in X,
\end{equation}
to the relative likelihood may be used in place of \eqref{def_rellikelihood}. If we let $\beta'_{\alpha}$ denote the value
determined in  \eqref{betadef} using the surrogate \eqref{def_rellikelihooda} and $\Theta'_{x}(\alpha)$ denote the corresponding likelihood region, then we have
$\beta_{\alpha}\leq \beta'_{\alpha}\, \, \text{and}\,\, \Theta'_{x}(\alpha) \subset \Theta_{x}(\alpha),\,\, \alpha \in [0,1]$. Consequently, obtaining
$\beta'_{\alpha}\leq \beta^{*}$  for significance level $\beta^{*}$ implies that $\beta_{\alpha}\leq \beta^{*}$.

As an example, for an $N$-dimensional Gaussian model with $p(x'|\theta)=\frac{1}{(\sigma\sqrt{2\pi})^{N}}e^{-\frac{1}{2\sigma^{2}}\|x'-\theta\|^{2}}$
with $\Theta:=[-\tau,\tau]^{N}$,  the elementary upper bound
\[ M(x'):=\sup_{\theta \in \Theta}{p(x|\theta)} \leq \frac{1}{(\sigma\sqrt{2\pi})^{N}}\]
the surrogate relative likelihood defined in \eqref{def_rellikelihooda} becomes
\[ \bar{p}'(x'|\theta):= e^{-\frac{1}{2\sigma^{2}}\|x'-\theta\|^{2}}.\]

\end{Remark}

\subsubsection{Posterior game and risk}

After observing $x \in X$,
we now consider playing a game
using the loss
\begin{equation}
\label{Lalphadefa}
 \L(\pi,d):=\E_{\theta \sim \pi_{x}}\bigl[\|\varphi(\theta)-d\|^2\bigr], \quad    \pi \in
  \mathcal{P}_{x}(\alpha), d \in V  ,
\end{equation}
where $\pi_{x}$ is the posterior \eqref{eq_conditional0}. In \eqref{Lalphadefa}, we think about the maximizing player as choosing the prior $\pi$ and then the loss function depends on the posterior $\pi_x$. Since the likelihood $p(x|\cdot)$ is positive and the data $x$ is fixed, we have $\supp(\pi_{x})=\supp(\pi)$ and the map \eqref{eq_conditional0} mapping the prior $\pi$ to the posterior $\pi_x$ is bijective. Therefore one can equivalently consider the choice of the maximizing player to be directly maximizing the posterior $\pi_x$ instead of the prior $\pi$ that is later mapped to the posterior using the data. The optimal choices of these two games can then be mapped by \eqref{eq_conditional0} and its inverse.
Using the invariance of  posterior \eqref{eq_conditional0} under the scaling of the likelihood function $p(x|\cdot)$ we write the posterior in terms of the relative likelihood  \eqref{def_rellikelihood} as
\begin{equation}
\label{eq_conditional}
  \pi_{x}:=\frac{\bar{p}(x|\cdot)\pi}{\int_{\Theta}{\bar{p}(x|\theta)d\pi(\theta)}}\, .
\end{equation}
Therefore for simplicity one directly considers a game using the loss
\begin{equation}
\label{Lalphadefa2}
 \L(\pi,d):=\E_{\theta \sim \pi}\bigl[\|\varphi(\theta)-d\|^2\bigr], \quad    \pi \in
  \mathcal{P}_{x}(\alpha), d \in V  .
\end{equation}
Recall that a pair $(\pi^{\alpha},d^\alpha) \in \mathcal{P}_{x}(\alpha) \times V$ is a saddle point
of the game \eqref{Lalphadefa2} if
\begin{equation*}
    \label{eq_saddle}
\L(\pi, d^{\alpha}) \leq \L(\pi^{\alpha},d^{\alpha}) \leq \L(\pi^{\alpha},d),\qquad  \pi  \in \mathcal{P}_{x}(\alpha), \,\,  d \in V .
\end{equation*}
We then have the following theorem.

\begin{Theorem}
\label{thm_saddles}
Consider $x \in X$, $\alpha \in [0,1]$, and suppose that
the relative likelihood $p(x|\cdot)$ and the quantity of interest $\varphi:\Theta \rightarrow V$ are continuous. The loss function $\L$ for the game \eqref{eqgame2} ($=$\eqref{Lalphadef}) has saddle points
and
a pair $(\pi^{\alpha},d^\alpha) \in \mathcal{P}_{x}(\alpha) \times V$ is a saddle point
for $\L$  if and only if
 \begin{equation}
  \label{dopt}
 d^\alpha:=\E_{\pi^\alpha}{[\varphi]}\,
 \end{equation}
and
\begin{equation}
\label{piopt}
\pi^{\alpha} \in \arg \max_{\pi \in  \mathcal{P}_{x}(\alpha)}\E_{\pi}{\bigl[\|\varphi-\E_{\pi}{[\varphi]}\|^{2}\bigr]}\, .
\end{equation}
Furthermore
 the associated risk  \textup{(}the value of the  two person game  \eqref{eqgame2} ($=$\eqref{Lalphadef}) \textup{)}
\begin{equation}
\label{Rstar2}
\mathcal{R}(d^\alpha):=\L(\pi^{\alpha},d^\alpha)=\E_{\pi^{\alpha}}{\bigl[\|\varphi-\E_{\pi^{\alpha}}{[\varphi]}\|^{2}\bigr]}
\end{equation}
is the same for all saddle points of $\L$.
Moreover, the second component $d^\alpha$ of the set of saddle points is unique and the set
$\mathcal{O}_{x}(\alpha)\subset \mathcal{P}_{x}(\alpha)$
of first components of saddle points is convex,
 providing  a  convex ridge $\mathcal{O}_{x}(\alpha)\times \{d^{\alpha}\}$
  of saddle points.
\end{Theorem}

\subsection{Duality with the minimum enclosing ball}
Although the Lagrangian duality between the maximum variance problem and the minimum enclosing ball problem on finite sets is known, see Yildirim
\cite{yildirim2008two}, we now analyze the infinite case.
Utilizing the recent  generalization of the one-dimensional result
 of Popoviciu \cite{Popoviciu} regarding the relationship between variance maximization   and the   minimum enclosing ball by Lim and McCann \cite[Thm.~1]{lim2021geometrical},
 the following theorem demonstrates that essentially the maximum variance problem  \eqref{piopt} determining a worst-case measure is the Lagrangian dual of the minimum enclosing ball problem  on  the image
$\varphi(\Theta_{x}(\alpha))$. Let
$\varphi_{*}:\mathcal{P}\bigl(
   \Theta_{x}(\alpha)\bigr) \rightarrow \mathcal{P}(\varphi(\Theta_{x}(\alpha))$ denote the pushforward map (change of variables) defined by
$(\varphi_{*}\pi)(A):=\pi(\varphi^{-1}(A))$ for every  Borel set $A$,
mapping probability measures on $\Theta_{x}(\alpha)$ to  probability measures on $\varphi\bigl(\Theta_{x}(\alpha)\bigr)$.

\begin{Theorem}\label{thm_duala}
For $x\in X$, $\alpha \in [0,1]$, suppose the relative likelihood $\bar{p}(x|\cdot)$ and the quantity of interest $\varphi:\Theta \rightarrow V$ are continuous.   Consider a saddle point $(\pi^\alpha,d^\alpha)$ of the game \eqref{eqgame2} ($=$\eqref{Lalphadef}).
The optimal decision $d^{\alpha}$  and its associated risk $\mathcal{R}(d^\alpha)=$\eqref{Rstar0} are
 equal to the center and squared radius, respectively,  of the minimum enclosing ball of
 $\varphi(\Theta_{x}(\alpha))$, i.e.~the  minimizer $z^{*}$ and the value $R^2$  of the minimum enclosing ball optimization problem
\begin{equation}
\label{opt_minball}
\begin{cases}
\text{Minimize }  r^{2}
 \\
\text{Subject to }   r \in \R,\, z \in \varphi(\Theta_{x}(\alpha)),\\
\|x-z\|^{2} \leq r^{2}, \quad x \in \varphi(\Theta_{x}(\alpha)).
\end{cases}
\end{equation}
Moreover, the  variance maximization problem  on $\mathcal{P}_{x}(\alpha)$ \eqref{piopt}
  pushes forward  to the variance maximization problem on  the image of the likelihood region
$\mathcal{P}(\varphi(\Theta_{x}(\alpha)))$ under $\varphi$  giving the identity
\[\E_{\pi}{\bigl[\|\varphi-\E_{\pi}{[\varphi]}\|^{2}\bigr]}
=
\E_{\varphi_{*}\pi}{\bigl[\|v-\E_{\pi'}{[v]}\|^{2}\bigr]} , \quad \pi \in \mathcal{P}_{x}(\alpha),
\] and
the latter is the Lagrangian dual to the minimum enclosing ball problem \eqref{opt_minball} on the image
$\varphi(\Theta_{x}(\alpha))$. Finally, let $B$, with center $z^{*}$, denote the minimum enclosing ball
of $\varphi(\Theta_{x}(\alpha))$. Then
 a measure $\pi^{\alpha} \in \mathcal{P}_{x}(\alpha)$ is optimal
for the variance maximization problem   \eqref{piopt} if and only if
\[ \varphi_{*}\pi^{\alpha}\bigl( \varphi(\Theta_{x}(\alpha)) \cap \partial B\bigr) =1 \]
and \[ z^{*}=\int_{V}{v d(\varphi_{*}\pi^{\alpha})(v)}, \]
that is, all the mass of $\varphi_{*}\pi^{\alpha}$ lives on the intersection
 $\varphi(\Theta_{x}(\alpha)) \cap \partial B$  of the image $ \varphi(\Theta_{x}(\alpha))$ of the
likelihood region and  the boundary  $\partial B$ of its minimum enclosing ball and
the center of mass of the measure $\varphi_{*}\pi^{\alpha}$ is the center $z^{*}$ of  $B$.
\end{Theorem}
\begin{Remark}
Note that once $\alpha$, and therefore $\Theta_{x}(\alpha)$, is determined
that the computation of the risk and the minmax estimator is determined by
the minimum enclosing ball about $\varphi(\Theta_{x}(\alpha))$, which is also determined by the
 worst-case optimization problem \eqref{worstcase} for $\Theta:=\Theta_{x}(\alpha)$.
\end{Remark}

Theorem  \ref{thm_duala} introduces the possibility of primal-dual algorithms. In particular, the availability of rigorous stopping criteria for
the maximum variance problem \eqref{piopt}. To that end, for
 a feasible measure $\pi \in \mathcal{P}_{x}(\alpha)$, let $\Var(\pi):=\E_{\pi}{\bigl[\|\varphi-\E_{\pi}{[\varphi]}\|^{2}\bigr]}$ denote its variance and denote by
 $\Var^{*}:=\sup_{\pi \in \mathcal{P}_{x}(\alpha)}{\Var(\pi)}=$\eqref{Rstar2} the optimal variance. Let $(r,z)$ be a feasible for the minimum enclosing ball problem \eqref{opt_minball}. Then the inequality $\Var^{*}=R^{2} \leq r^{2}$  implies the rigorous bound
\begin{equation}
\Var^{*}-\Var(\pi) \leq r^{2}- \Var(\pi)
\end{equation}
 quantifying  the suboptimality of the measure $\pi$ in terms of known quantities $r$ and $\Var(\pi)$.

\subsection{Finite-dimensional reduction}

Let $\Delta^{m}(\Theta)$ denote the set of convex sums of $m$ Dirac measures located in $\Theta$ and,
 let $\mathcal{P}^{m}_{x}(\alpha) \subset \mathcal{P}_{x}(\alpha)$ defined by
\begin{equation}
\label{PalphaHm}
 \mathcal{P}^{m}_{x}(\alpha):=
 \Delta^{m}(\Theta) \cap \mathcal{P}_{x}(\alpha)
\,
 \end{equation}
denote the  finite-dimensional subset of the  rarity assumption set $\mathcal{P}_{x}(\alpha)$
 consisting of
the convex combinations of $m$ Dirac measures supported in $\Theta_{x}(\alpha)$.

\begin{Theorem}\label{pioptreduced}
Let $\alpha \in [0,1]$ and $x \in X$, and suppose that the likelihood function $p(x|\cdot)$ and quantity of interest
$\varphi:\Theta \rightarrow V$ are continuous. Then for any $m\geq dim(V)+1$, the variance maximization problem \eqref{piopt} has the
finite-dimensional reduction
\begin{equation}
\label{eq_pioptreduced}
\max_{\pi \in  \mathcal{P}_{x}(\alpha)}\E_{\pi}{\bigl[\|\varphi-\E_{\pi}{[\varphi]}\|^{2}\bigr]}=
\max_{\pi \in  \mathcal{P}^{m}_{x}(\alpha)}\E_{\pi}{\bigl[\|\varphi-\E_{\pi}{[\varphi]}\|^{2}\bigr]}\, .
\end{equation}
Therefore one can compute a saddle point $(d^\alpha,\pi^\alpha)$ of the game \eqref{eqgame2} ($=$\eqref{Lalphadef}) as
\begin{equation}\label{eqjhegvdievd6}
\pi^\alpha=\sum_{i=1}^{m}{w_{i}\updelta_{\theta_{i}}}\text{ and } d^\alpha=\sum_{i=1}^{m}{w_{i}\varphi(\theta_{i})}
\end{equation}
where $w_{i} \geq 0, \theta_{i} \in \Theta, \, i=1,\ldots, m$ maximize
\begin{equation}
\label{opt_priora}
\begin{cases}
\text{Maximize }  \sum_{i=1}^{m}{w_{i}\bigl\|\varphi(\theta_{i})\bigr\|^{2}}-\bigl\|\sum_{i=1}^{m}{w_{i}\varphi(\theta_{i})}\bigr\|^{2}
 \\
\text{Subject to }  \,\, w_{i}\geq 0, \theta_{i}\in \Theta, i=1,\ldots, m,\,
\sum_{i=1}^{m}{w_{i}}=1\\
 \bar{p}(x|\theta_{i})\geq \alpha, \quad i=1,\ldots, m\, .
\end{cases}
\end{equation}
\end{Theorem}

As a consequence of Theorems  \ref{thm_duala} and \ref{pioptreduced},
a measure with finite support  $\mu:=\sum{w_{i}\updelta_{z_{i}}}$ on $V$ is the pushforward under $\varphi:\Theta \rightarrow V$ of an optimal measure $\pi^\alpha$ for  the maximum variance problem \eqref{piopt}
if and only if, as illustrated in Figure \ref{figminball},  it is supported on
 the intersection of $\varphi(\Theta_{x}(\alpha))$ and the boundary $\partial B$ of the minimum enclosing ball of $\varphi(\Theta_{x}(\alpha)) $
 and the center $z^{*}$ of $B$ is the center of mass $z^{*}=\sum{w_{i}z_{i}}$
 of the measure $\mu$.

 \subsection{Relaxing MLE with an  accuracy/robustness tradeoff}

For fixed $x \in X$,
assume that the model $P$ is such that  the maximum likelihood estimate (MLE)
\begin{equation}
\theta^*:=\arg \max_{\theta \in \Theta}p(x|\theta)\,
\end{equation}
of $\theta^\dagger$ exists and is unique.

Observe that for $\alpha$ near one (1) the support of $\pi^\alpha$  and $d^\alpha$ concentrate around the MLE $\theta^*$
and  $\varphi(\theta^*)$,  (2) the risk $\mathcal{R}(d^\alpha)=$\eqref{Rstar0} concentrates around zero, and (3) the confidence $1-\beta_\alpha$ associated with the rarity assumption $\theta^\dagger\in \Theta_{x}(\alpha)$ is the smallest.
In that limit, our estimator inherits the accuracy and lack of robustness of the MLE approach to estimating the quantity of interest.

Conversely for $\alpha$ near zero, since by \eqref{eqnormrarset} $\Theta_{x}(\alpha) \approx \Theta$,  (1) the support of the pushforward of $\pi^\alpha$  by $\varphi$ concentrates on the boundary of $\varphi(\Theta)$ and   $d^\alpha$ concentrate around the center of the minimum enclosing ball of $\varphi(\Theta)$,  (2) the risk $\mathcal{R}(d^\alpha)=$\eqref{Rstar2} is the highest and concentrates around the worst-case risk \eqref{worstcase}, and (3) the confidence $1-\beta_\alpha$ associated with the rarity assumption $\theta^\dagger\in \Theta_{x}(\alpha)$ is the highest.
In that limit, our estimator inherits the robustness and lack of accuracy of the worst-case approach to estimating the quantity of interest.

 For $\alpha$ between $0$ and $1$, the proposed game-theoretic approach induces a minmax optimal tradeoff between the accuracy of MLE and the robustness of the worst case.

\section{Computational framework} \label{sec:compfram}
The introduction developed this framework
in the context of a model $P$ with density $p$  in terms of a single sample $x$.
In Section  \ref{sec_coin},
the single sample case was extended to $N$ i.i.d.~samples by defining
the multisample $\mathcal{D}:=(x_{1},\ldots,x_{N})$ and defining the product model density
$p(\mathcal{D}|\theta):=\Pi_{i=1}^{N}{p(x_{i}|\theta)}$.
 Extensions incorporating correlations in the samples,
 such as Markov or other stochastic processes can easily be developed.
Here we continue this development for the general model of the introduction for the $\ell^{2}$ loss and also develop more fully a Gaussian noise model.
Later, in Sections \ref{sec_gradient} and \ref{sec_LV}
 these models will be tested on estimating a quadratic function and  a Lotka-Volterra predator-prey model based on noisy observations. In Section
\ref{sec_generalloss} the framework will be generalized to more general loss functions and rarity assumptions,  which much of the current section generalizes to.

Let the   possible  states  of nature be a compact subset
 $ \Theta \subset \mathbb{R}^{k}$,
the decision space be $V:=\R^{n}$ and  the elements of  the  $N$-fold
multisample
\[\mathcal{D}:=(\boldsymbol{x}_{1},\ldots,\boldsymbol{x}_{N})\]
  lie in $X$,
that is, $\mathcal{D}$ lies in  the multisample space $X^{N}$.
Let the $n$ components of the quantity of interest $\varphi:\Theta \rightarrow \R^{n}$ be indicated by
$\varphi_{t}:\Theta \rightarrow \R, \, t =1,\ldots n$. Here, using the i.i.d.~ product model
$P(\cdot|\theta)$ with density
$p(\mathcal{D}|\theta):=\Pi_{i=1}^{N}{p(x_{i}|\theta)}$,  the definition of
$\beta_{\alpha}$ in \eqref{betadef} becomes
\begin{eqnarray}
\label{betadef2}
    \beta_{\alpha}&:=& \sup_{\theta \in \Theta}
P\Big(\bigl\{\mathcal{D}' \in X^{N}: \theta \notin \Theta_{\mathcal{D}'}(\alpha)\bigr\}\Big|\theta\Big)\nonumber\\
&=& \sup_{\theta \in \Theta}
{\int{\mathds{1}_{\{\bar{p}(\cdot|\theta) < \alpha\} }(\mathcal{D}')p(\mathcal{D}'|\theta)d\nu^{N}(\mathcal{D}')}}\,,
\end{eqnarray}
where, for fixed $\theta$,
$ \mathds{1}_{\{\bar{p}(\cdot|\theta) < \alpha\}} $ is the indicator function
of the set $\{\mathcal{D}' \in X^{N}: \bar{p}(\mathcal{D}'|\theta) < \alpha\}$.
We use boldface, such as $\boldsymbol{x}_{i}$ or $\boldsymbol{\theta}$ to emphasize the vector nature of variables and functions in the computational framework.

In this notation, the finite dimensional reduction guaranteed by Theorem \ref{pioptreduced}
in \eqref{eqjhegvdievd6} and \eqref{opt_priora}
of the
  optimization problem \eqref{piopt} defining a worst-case measure of the form
$\pi^{\alpha}:=\sum_{i=1}^{m}{w_{i}\updelta_{\boldsymbol{\theta_{i}}}}$
 takes the form
\begin{equation}
\begin{array}{ll}
\underset{\left\{ \left(w_{i},\boldsymbol{\theta_{i}}_{i}\right)\right\} _{i=1}^{m}}{\textrm{maximize}} & \sum_{t=1}^{n}\left(\sum_{i=1}^{m}\varphi_{t}^{2}\left(\boldsymbol{\theta}_{i}\right)w_{i}-\sum_{i,j=1}^{m}w_{i}\varphi_{t}\left(\boldsymbol{\theta}_{i}\right)\varphi_{t}\left(\boldsymbol{\theta}_{j}\right)w_{j}\right)\\
s.t. & \boldsymbol{\theta}_{i} \in \Theta,\,  w_{i}\geqslant0,\quad i=1,\ldots,m\\
 & \sum_{i=1}^{m}w_{i}=1,\\
 & \bar{p}\left(\mathcal{D}|\boldsymbol{\theta}_{i}\right)\geqslant\alpha,i=1,\ldots,m,
\end{array}\label{eq: Posterior-based Opt}
\end{equation}
where the  component of the objective function
 \[var\left(\varphi_{t}\right):=
\sum_{i=1}^{m}\varphi_{t}^{2}\left(\boldsymbol{\theta}_{i}\right)w_{i}-\sum_{i,j=1}^{m}w_{i}\varphi_{t}\left(\boldsymbol{\theta}_{i}\right)\varphi_{t}\left(\boldsymbol{\theta}_{j}\right)w_{j}\]
 is the variance
of the random variable $\varphi_{t}:\Theta \rightarrow \R$ under the measure
$\pi:=\sum_{i=1}^{m}{w_{i}\updelta_{\boldsymbol{\theta}_{i}}}$.

\subsection{Algorithm for solving the game}
We are now prepared to develop an algorithm for player II (the decision maker) to play the game
\eqref{Lalphadefa2}, using the saddle point Theorem \ref{thm_saddles} and the finite dimensional reduction Theorem \ref{pioptreduced} after selecting the rarity parameter $\alpha$ quantifying the rarity  assumption \eqref{eqnormrarset} in terms of the relative likelihood \eqref{def_rellikelihood}
or a surrogate as described in Remark \ref{rmk_surrogate},  to be the largest $\alpha$ such that
the significance
$\beta_{\alpha}$ \eqref{betadef2} at $\alpha$ satisfies $\beta_{\alpha} \leq \beta^{*}$, the significance level.

At a high level the algorithm for computing a worst-case measure, its resulting risk (variance) and optimal estimator is as follows:
\begin{enumerate}
\item Observe a multisample $\mathcal{D}$
\item Find the largest $\alpha$ such that
 $\beta_{\alpha}$ defined in \eqref{betadef2} satisfies $\beta_{\alpha} \leq \beta^{*}$
\item Solve \eqref{eq: Posterior-based Opt}
determining a worst-case measure
$\pi^{\alpha}:=\sum_{i=1}^{m}{w_{i}\updelta_{\boldsymbol{\theta}_{i}}}$.
\item output the Risk as the value of \eqref{eq: Posterior-based Opt}
\item output optimal decision $d^{\alpha}:=\sum_{i=1}^{m}{w_{i}\varphi(\boldsymbol{\theta}_{i})}$
\end{enumerate}

To solve \eqref{eq: Posterior-based Opt} in Step 3 we apply the duality
 of the variance maximization problem with the minimum enclosing ball problem,
 Theorem \ref{thm_duala}, to obtain the following complete algorithm.
 It uses Algorithm \ref{Alg: Miniball} for computing the minimum enclosing ball about
the (generally)  infinite set
$\varphi(\Theta_{x}(\alpha))$, which in turn uses a minimum enclosing ball algorithm {\em Miniball}
applied to sets of size at most $dim(V)+2$, see e.g.~Welzl \cite{welzl1991smallest},
 Yildirim \cite{yildirim2008two}  and Gartner \cite{gartner1999fast}. Here we use that of
 Welzl \cite{welzl1991smallest}.
See Section \ref{sec_miniball} for a discussion and a proof in Theorem \ref{thm_miniball}
of the  convergence of Algorithm \ref{Alg: Miniball}. Theorem \ref{thm_miniball}
also establishes a convergence proof when the  distance maximization Step \ref{step8a} in
  Algorithm \ref{Alg: ODAD General} is performed approximately.
Note that the likelihood region $\Theta_{\mathcal{D}}(\alpha)$ is defined by
\[ \Theta_{\mathcal{D}}(\alpha):=\bigl\{\boldsymbol{\theta} \in \Theta:  p(\mathcal{D}|\boldsymbol{\theta}) \geq \alpha \,  p(\mathcal{D}|\boldsymbol{\theta}^{*})\bigr\} \]
where
\[ \boldsymbol{\theta}^{*}\in\underset{\boldsymbol{\theta}}{\arg\max}\,\,{ p\bigl(\mathcal{D}|\boldsymbol{\theta}\bigr)}\]
is a MLE.

\begin{algorithm}

\begin{enumerate}
\item Inputs:
\begin{enumerate}
\item Multisample $\mathcal{D}:=(\boldsymbol{x}_{1},\ldots,\boldsymbol{x}_{N})$
\item $\varepsilon_{0}$
\item Significance level $\beta^{*}$
\end{enumerate}
\item \label{step2} Find MLE $\boldsymbol{\theta}^{*}$ by
 $\boldsymbol{\theta}^{*}\in\underset{\boldsymbol{\theta}}{\arg\max}\,\,{ p\bigl(\mathcal{D}|\boldsymbol{\theta}\bigr)}$
\item  \label{step3} Find the largest $\alpha$ such that
 $\beta_{\alpha}$ defined in \eqref{betadef2} satisfies $\beta_{\alpha} \leq \beta^{*}$
\item $\boldsymbol{c}\leftarrow\boldsymbol{\varphi}\left(\boldsymbol{\theta}^{*}\right)$
\item $S\leftarrow \left\{ \boldsymbol{\varphi}\left(\boldsymbol{\theta}^{*}\right)\right\} $
\item $\rho_{0}\leftarrow0$
\item $e\leftarrow2\varepsilon_{0}$
\item while $e\geqslant\varepsilon_{0}$
\begin{enumerate}
\item \label{step8a} $\begin{array}{lll}
\bar{\boldsymbol{\theta}} & \in & \underset{\boldsymbol{\theta}}{\arg\max}\left\Vert \boldsymbol{\varphi}\left(\boldsymbol{\theta}\right)-\boldsymbol{c}\right\Vert^{2}\\
 & s.t. &  p(\mathcal{D}|\boldsymbol{\theta}) \geq \alpha \,  p(\mathcal{D}|\boldsymbol{\theta}^{*})
\end{array}$
\item if $\left\Vert \boldsymbol{\varphi}\left(\bar{\boldsymbol{\theta}}\right)-\boldsymbol{c}\right\Vert\geqslant\rho_{0}$
\begin{enumerate}
\item $S\leftarrow S\cup\left\{ \boldsymbol{\varphi}\left(\bar{\boldsymbol{\theta}}\right)\right\} $
\end{enumerate}
\item $\boldsymbol{c},\rho\leftarrow Miniball\left(S\right)$
\item $e=\left|\rho-\rho_{0}\right|$
\item $\rho_{0}\leftarrow\rho$
\item if $\left|S\right|>n+1$
\begin{enumerate}
\item find subset $S'\subset S$  of size $  n+1$ such that $ Miniball\left(S'\right)= Miniball\left(S\right)$
\item $ S \leftarrow S'$
\end{enumerate}
\end{enumerate}
\item Find $\left\{ w_{i}\right\} _{i=1}^{n+1}$ from $\underset{\left\{ w_{i}\right\} _{i=1}^{n+1}}{\textrm{maximize}}\sum_{t=1}^{n}\left(\sum_{i=1}^{n+1}\varphi_{t}^{2}\left(\boldsymbol{\theta}_{i}\right)w_{i}-\sum_{i,j=1}^{n+1}w_{i}\varphi_{t}\left(\boldsymbol{\theta}_{i}\right)\varphi_{t}\left(\boldsymbol{\theta}_{j}\right)w_{j}\right)$
\end{enumerate}
\caption{UQ4K algorithm}
\label{Alg: ODAD General}
\end{algorithm}

\subsubsection{Large sample simplifications}
\label{sec_hamed}
Here we demonstrate that when the number of samples $N$ is large, under classic regularity assumptions, the significance $\beta_{\alpha}$ is approximated by the value of a chi-squared distribution, substantially simplifying the  determination of $\alpha$  in Step \ref{step3} of Algorithm \ref{Alg: ODAD General}.

Let $\Theta\subseteq \mathbb{R}^k$ and let data be generated by the model at the value $\theta\in \Theta$. Under standard regularity conditions, check Casella and Berger \cite[Sec.~10.6.2 \& Thm.~10.1.12]{casella2003statistical}, the maximum likelihood estimator (MLE), $\hat{\theta}_N$, is asymptotically efficient for $\theta$. That is as the sample size $N\to \infty$
\begin{equation}
\sqrt{N}(\hat{\theta}_N-\theta)\xrightarrow{\text{d}} N(0,I(\theta)^{-1})\, ,
\end{equation}
where $I(\theta)$ is the Fisher information matrix. Therefore, standard arguments, see Casella and Berger \cite[Thm.~10.3.1]{casella2003statistical}, for the asymptotic distribution of the likelihood ratio test result in the following approximation of $\beta_\alpha$.

\begin{Theorem}\label{thm_asympt}
Let $\Theta\subseteq\mathbb{R}^k$ and assume that the model density $p$ satisfies the regularity conditions of Casella and Berger \cite[Section.~10.6.2]{casella2003statistical}. Then
\begin{equation}
    \beta_\alpha \to 1-\chi^2_k\big(2\ln \frac{1}{\alpha}\big)\,
\end{equation}
as  $N\to \infty$,
where $\chi^2_k$ is the chi-square distribution with $k$ degrees of freedom.
\end{Theorem}

Consequently, under these conditions  Step \ref{step3}  of Algorithm \ref{Alg: ODAD General} can take the simple form
\begin{enumerate}
 \item[] (Step \ref{step3}): Solve for $\alpha$ satisfying
$\beta_{\alpha}:= 1-\chi^{2}_{k}\bigl( 2 \ln \frac{1}{\alpha}\bigr) =\beta^{*}$
\end{enumerate}

\subsection{Algorithm \ref{Alg: ODAD General} for a Gaussian noise  model}
\label{sec_gaussian}
Consider  a Gaussian noise model where, $X=\R^{r}$ and for $\boldsymbol{\theta} \in \Theta$, the
components of the multisample
$\mathcal{D}:=(\boldsymbol{x}_{1},\ldots,\boldsymbol{x}_{N}) \in \R^{rN}$ are
 i.i.d.~samples
from  the Gaussian distribution
$\mathcal{N}(\boldsymbol{m}\left(\boldsymbol{\theta}\right), \sigma^{2}I_{r})$,
with mean $\boldsymbol{m}\left(\boldsymbol{\theta}\right)$ and covariance
$ \sigma^{2}I_{r}$,
where
$\boldsymbol{m}:\Theta \rightarrow\mathbb{R}^{r}$ is a \emph{measurement function}, $\sigma >0$ and
$I_{r}$ is the $r$-dimensional identity matrix. The  measurement function $\boldsymbol{m}$ is a function  such that the its value $\boldsymbol{m}(\theta)$ can be computed when the model parameter $\boldsymbol{\theta}$ is known.
Therefore the i.i.d.~multisample $\mathcal{D}:=(\boldsymbol{x}_{1},\ldots,\boldsymbol{x}_{N})$ is drawn from
$\bigl(\mathcal{N}(\boldsymbol{m}\left(\boldsymbol{\theta}\right), \sigma^{2}I_{r})\bigr)^{N}$ and so has the
probability density
\begin{equation}
p\left(\mathcal{D}|\boldsymbol{\theta}\right)=\frac{1}{\left(\sigma\sqrt{2\pi}\right)^{rN}}\exp\left(-\frac{1}{2\sigma^{2}}\sum_{j=1}^{N}\left\Vert \boldsymbol{x}_{j}-\boldsymbol{m}\left(\boldsymbol{\theta}\right)\right\Vert^{2}\right)\label{eq: Likelihood0}
\end{equation}
with respect to the Lebesgue measure $\nu$ on $X:=\mathbb{R}^{rN}$,
and defining $\left(\sigma\sqrt{2\pi}\right)^{rN}$ times the maximum likelihood
\begin{equation}
    \label{def_maxlike}
    M(\mathcal{D}):=\exp\left(-\frac{1}{2\sigma^{2}}\inf_{\boldsymbol{\theta} \in \Theta}{\Bigl(\sum_{j=1}^{N}\left\Vert \boldsymbol{x}_{j}-\boldsymbol{m}\left(\boldsymbol{\theta}\right)\right\Vert^{2}\Bigr)}\right)\,,
\end{equation}
the relative likelihood \eqref{def_rellikelihood} is
\begin{equation}
\bar{p}\left(\mathcal{D}|\boldsymbol{\theta}\right)=\frac{\exp\left(-\frac{1}{2\sigma^{2}}\sum_{j=1}^{N}\left\Vert \boldsymbol{x}_{j}-\boldsymbol{m}\left(\boldsymbol{\theta}\right)\right\Vert^{2}\right)}{M(\mathcal{D})}\, .
\label{eq: Likelihood}
\end{equation}
Taking the logarithm of the constraint $\bar{p}\left(\mathcal{D}|\boldsymbol{\theta}_{i}\right)\geqslant\alpha$
 defining the likelihood region  $\Theta_{\mathcal{D}}(\alpha)$, using
\eqref{eq: Likelihood} we obtain
\begin{equation}
\label{eq: Rarity for Gaussian}
\Theta_{\mathcal{D}}(\alpha) = \Bigl\{ \boldsymbol{\theta} \in \Theta : \sum_{j=1}^{N}\bigl\|\boldsymbol{x}_{j}-\boldsymbol{m}\left( \boldsymbol{\theta}\right)\bigr\| ^{2}\leqslant M_{\alpha} \Bigr\}
\end{equation}
in terms of
\begin{equation}
\label{Mdef}
  M_{\alpha}:=\inf_{\boldsymbol{\theta} \in \Theta}{\sum_{j=1}^{N}\bigl\|
 \boldsymbol{x}_{j}-\boldsymbol{m}\left( \boldsymbol{\theta}\right)\bigr \|^{2}} +2\sigma^{2} \ln \frac{1}{\alpha}\, .
\end{equation}
Consequently, for the Gaussian case, the worst-case measure optimization problem
\eqref{eq: Posterior-based Opt} becomes
\begin{equation}
\begin{array}{ll}
\underset{\left\{ \left(w_{i},\boldsymbol{\theta}_{i}\right)\right\} _{i=1}^{m}}{\textrm{maximize}} & \sum_{t=1}^{n}\left(\sum_{i=1}^{m}\varphi_{t}^{2}\left(\boldsymbol{\theta}_{i}\right)w_{i}-\sum_{i,j=1}^{m}w_{i}\varphi_{t}\left(\boldsymbol{\theta}_{i}\right)\varphi_{t}\left(\boldsymbol{\theta}_{j}\right)w_{j}\right)\\
s.t. & \boldsymbol{\theta}_{i} \in \Theta,\,  w_{i}\geqslant0,\quad i=1,\ldots,m\\
 & \sum_{i=1}^{m}w_{i}=1,\\
 & \sum_{j=1}^{N}\left\Vert \boldsymbol{x}_{j}-\boldsymbol{m}\left(\boldsymbol{\theta}_{i}\right)\right\Vert^{2}\leqslant M_{\alpha},\quad i=1,\ldots,m.
\end{array}\label{eq: Posterior-based Opt Gaussian}
\end{equation}
Consequently, in the Gaussian noise case,
   Algorithm \ref{Alg: ODAD General} appears with these modifications:
\begin{enumerate}
\item (Step \ref{step2}): Find MLE $\boldsymbol{\theta}^{*}$ by
 $\boldsymbol{\theta}^{*}\in \underset{\boldsymbol{\theta}}{\arg\min}\sum_{j=1}^{N}\left\Vert \boldsymbol{x}_{j}-\boldsymbol{m}\left(\boldsymbol{\theta}\right)\right\Vert^{2}$
\item (Step \ref{step8a}): Solve
\begin{equation}
\begin{array}{lll}
\bar{\boldsymbol{\theta}} & \in & \underset{\boldsymbol{\theta}}{\arg\max}\left\Vert \boldsymbol{\varphi}\left(\boldsymbol{\theta}\right)-\boldsymbol{c}\right\Vert^{2}\\
 & s.t. & \sum_{j=1}^{N}\left\Vert \boldsymbol{x}_{j}-\boldsymbol{m}\left(\boldsymbol{\theta}\right)\right\Vert^{2}\leqslant M_{\alpha}.
\end{array}\label{eq: ODAD2}
\end{equation}
\end{enumerate}

\subsubsection{Farthest point optimization in the Gaussian model }
\label{sec_merit}
In Step \ref{step8a} of  Algorithm \ref{Alg: ODAD General}  we
seek the farthest point $\bar{\boldsymbol{\theta}}$ from a center
$\boldsymbol{c}$:
\begin{equation}
\begin{array}{lll}
\bar{\boldsymbol{\theta}} & \in & \underset{\boldsymbol{\theta}}{\arg\max}\left\Vert \boldsymbol{\varphi}\left(\boldsymbol{\theta}\right)-\boldsymbol{c}\right\Vert^{2}\\
 & s.t. & \sum_{j=1}^{N}\left\Vert \boldsymbol{x}_{j}-\boldsymbol{m}\left(\boldsymbol{\theta}\right)\right\Vert^{2}\leqslant M_{\alpha}.
\end{array}\label{eq: Farthest point}
\end{equation}
To solve this optimization, we use the merit function technique \cite{nocedal2006numerical}
as follows:
\begin{equation}
\underset{\boldsymbol{\theta}}{\textrm{minimize}}-\left\Vert \boldsymbol{\varphi}\left(\boldsymbol{\theta}\right)-\boldsymbol{c}\right\Vert^{2}+\mu\max\left\{ 0,\sum_{j=1}^{N}\left\Vert \boldsymbol{x}_{j}-\boldsymbol{m}\left(\boldsymbol{\theta}\right)\right\Vert^{2}-M_{\alpha}\right\} .\label{eq: Merit Function}
\end{equation}
In implementation, one should start with a small value of $\mu$ and
increase it to find the optimum \cite{nocedal2006numerical}. The
first term in (\ref{eq: Merit Function}) intends to increase the
distance from the center $\boldsymbol{c}$ and the second term keeps
the solution feasible. Any algorithm picked to solve (\ref{eq: Merit Function})
must be able to slide near the feasibility region $\sum_{j=1}^{N}\left\Vert \boldsymbol{x}_{j}-\boldsymbol{m}\left(\boldsymbol{\theta}\right)\right\Vert^{2}\leqslant M_{\alpha}$
to guarantee a better performance. Suggestions of such algorithms
are the gradient descent \cite{kingma2014adam}, if the gradients
are available, and differential evolution \cite{storn1997differential},
if gradients are not available.

\subsubsection{Surrogate relative likelihoods}
Although the computation of the maximum likelihood $\boldsymbol{\theta}^{*}$ in Step \ref{step2}
of Algorithm \ref{Alg: ODAD General} is only done once -for the observed data $\mathcal{D}$,
 the computation of $\beta_{\alpha}$ in Step \ref{step3} requires it to be computed for all data
$\mathcal{D}'$ generated by the statistical model.   Simplification of this computation  can be obtained by
large sample $N$ approximations, see Section \ref{sec_hamed}, or the utilization of  a surrogate relative likelihood as discussed in Remark \ref{rmk_surrogate}, which we now address.

 Let the generic multisample be  $\mathcal{D}':=(x'_{1},\ldots,x'_{N})$
 in the computation of $\beta_{\alpha}$ in \eqref{betadef2}, and
consider the upper bound on the maximum likelihood of the Gaussian noise model \eqref{eq: Likelihood0}
\begin{eqnarray*}
 \sup_{\theta \in \Theta}{ p\left(\mathcal{D}'|\boldsymbol{\theta}\right)}& =&
  \frac{1}{\left(\sigma\sqrt{2\pi}\right)^{rN}} \sup_{\theta \in \Theta}{\exp\left(-\frac{1}{2\sigma^{2}}\sum_{j=1}^{N}\bigl\| \boldsymbol{x}'_{j}-\boldsymbol{m}\left(\boldsymbol{\theta}\right)\bigr\|^{2}\right)}\\
& \leq&
  \frac{1}{\left(\sigma\sqrt{2\pi}\right)^{rN}} \sup_{\boldsymbol{m} \in \R^{r}}{\exp\left(-\frac{1}{2\sigma^{2}}\sum_{j=1}^{N}\bigl\| \boldsymbol{x}'_{j}-\boldsymbol{m}\bigr\|^{2}\right)}\\
& =&
  \frac{1}{\left(\sigma\sqrt{2\pi}\right)^{rN}} \exp\left(-\frac{1}{2\sigma^{2}}\sum_{j=1}^{N}\bigl\| \boldsymbol{x}'_{j}-\frac{1}{N}\sum_{k=1}^{N}{
 \boldsymbol{x}'_{k}}\bigr\|^{2}\right),
\end{eqnarray*}
so that the  resulting
 surrogate relative
likelihood (using the same symbol as the relative likelihood)   discussed in Remark \ref{rmk_surrogate} becomes
\begin{equation}
\bar{p}\left(\mathcal{D}'|\boldsymbol{\theta}\right)=\frac{
\exp\left(-\frac{1}{2\sigma^{2}}\sum_{j=1}^{N}\bigl\| \boldsymbol{x}'_{j}-\boldsymbol{m}\left(\boldsymbol{\theta}\right)\bigr\|^{2}\right)}
{\exp\left(-\frac{1}{2\sigma^{2}}\sum_{j=1}^{N}\bigl\| \boldsymbol{x}'_{j}-\frac{1}{N}\sum_{k=1}^{N}{
 \boldsymbol{x}'_{k}}\bigr\|^{2}\right)} ,
\label{eq: Likelihood3}
\end{equation}
 and therefore the condition
$\bar{p}\left(\cdot|\boldsymbol{\theta}\right) < \alpha$ in
the computation of the surrogate significance $\beta'_{\alpha} \geq \beta_{\alpha} $  defined in \eqref{betadef2} in terms of the surrogate relative likelihood \eqref{eq: Likelihood3} in Step \ref{step3} appears as
\begin{equation}
\label{eijuuurtuur}
 \sum_{j=1}^{N}\bigl\|\boldsymbol{x}'_{j}-\boldsymbol{m}\left( \boldsymbol{\theta}\right)\bigr\| ^{2}-  \sum_{j=1}^{N}\bigl\| \boldsymbol{x}'_{j}-\frac{1}{N}\sum_{k=1}^{N}{
 \boldsymbol{x}'_{k}}\bigr\|^{2}> 2\sigma^{2} \ln \frac{1}{\alpha}   .
\end{equation}
Rewriting in terms of the $N(0,\sigma^{2}I_{r})$ Gaussian random variables
\[
  \epsilon_{i}:=\boldsymbol{x}'_{j}-\boldsymbol{m}( \boldsymbol{\theta}), i=1,\ldots, N \]
 we obtain
\begin{eqnarray*}
 \sum_{j=1}^{N}\bigl\|\boldsymbol{x}'_{j}-\boldsymbol{m}\left( \boldsymbol{\theta}\right)\bigr\| ^{2}-  \sum_{j=1}^{N}\bigl\| \boldsymbol{x}'_{j}-\frac{1}{N}\sum_{k=1}^{N}{
 \boldsymbol{x}'_{k}}\bigr\|^{2}
 &=&
\sum_{j=1}^{N}\bigl\|\epsilon_{j}\bigr\| ^{2}-  \sum_{j=1}^{N}\bigl\| \epsilon_{j}-
\frac{1}{N}\sum_{k=1}^{N}{
\epsilon_{k}}\bigr\|^{2}\\
 &=&2
\sum_{j=1}^{N}\bigl\langle \epsilon_{j},\frac{1}{N}\sum_{k=1}^{N}{
\epsilon_{k}} \rangle  - \bigl\|
\frac{1}{N}\sum_{k=1}^{N}{
\epsilon_{k}}\bigr\|^{2}\\
&=&(2N-1)
\Bigl\|
\frac{1}{N}\sum_{k=1}^{N}{
\epsilon_{k}}\Bigr\|^{2} ,
\end{eqnarray*}
that is
\begin{equation}
\label{eooooieiiiei}
 \sum_{j=1}^{N}\bigl\|\boldsymbol{x}'_{j}-\boldsymbol{m}\left( \boldsymbol{\theta}\right)\bigr\| ^{2}-  \sum_{j=1}^{N}\bigl\| \boldsymbol{x}'_{j}-\frac{1}{N}\sum_{k=1}^{N}{
 \boldsymbol{x}'_{k}}\bigr\|^{2}=(2N-1)\|v\|^{2}
\end{equation}
where
\[ v:=\frac{1}{N}\sum_{k=1}^{N}{
\epsilon_{k}}\]
is Gaussian with mean zero and, since the $\epsilon_{k}$ are i.i.d,  have  covariance  $\sigma^{2} I_{r}$, that is
$v \in N(0,\frac{\sigma^{2}}{N}I_{r})$.
 Since Schott
\cite[Thm.~9.9]{schott2016matrix} implies that
$\frac{N}{\sigma^{2}}\|v\|^{2} $ is distributed as $\chi^{2}_{r}$, it follows from
  \eqref{eijuuurtuur}, \eqref{eooooieiiiei} and the definition
of the surrogate significance $\beta'_{\alpha}$ \eqref{betadef2}  that
\begin{equation}
\label{beta_chi}
\beta'_{\alpha}=1-\chi^{2}_{r}\bigl(\frac{2N}{2N-1}\ln \frac{1}{\alpha}\bigr) \geq \beta_{\alpha}.
\end{equation}

Consequently, removing the prime indicating the surrogate significance $\beta'_{\alpha}$, denoting it as $\beta_{\alpha}$,  the modifications \eqref{eq: ODAD2} to
   Algorithm \ref{Alg: ODAD General} are augmented  to:
\begin{enumerate}
\item (Step \ref{step2}): Find MLE $\boldsymbol{\theta}^{*}$ by
 $\boldsymbol{\theta}^{*}\in \underset{\boldsymbol{\theta}}{\arg\min}\sum_{j=1}^{N}\left\Vert \boldsymbol{x}_{j}-\boldsymbol{m}\left(\boldsymbol{\theta}\right)\right\Vert^{2}$
\item (Step \ref{step3}): Solve for $\alpha$ satisfying
$\beta_{\alpha}:= 1-\chi^{2}_{r}\bigl(\frac{2N}{2N-1}\ln \frac{1}{\alpha}\bigr) =\beta^{*}$
\item (Step \ref{step8a}): Solve
\begin{equation}
\begin{array}{lll}
\bar{\boldsymbol{\theta}} & \in & \underset{\boldsymbol{\theta}}{\arg\max}\left\Vert \boldsymbol{\varphi}\left(\boldsymbol{\theta}\right)-\boldsymbol{c}\right\Vert^{2}\\
 & s.t. & \sum_{j=1}^{N}\left\Vert \boldsymbol{x}_{j}-\boldsymbol{m}\left(\boldsymbol{\theta}\right)\right\Vert^{2}\leqslant M_{\alpha}.
\end{array}\label{eq: ODAD3}
\end{equation}
\end{enumerate}

\subsection{Stochastic processes}

We now consider the case where $\Theta\subseteq\mathbb{R}^k$ and the data is the  multisample $\mathcal{D}:=(x_{1},\ldots,x_{N})$  with $x_i=(y_i,t_i)\in \mathcal{Y}\times \mathcal{T}$ where  $y_i$ corresponds to the observation of a stochastic process at time $t_i$. Letting $\theta$ parameterize  the distribution of the stochastic process and assuming  the $y_i$ to be independent given $\theta$ and the $t_i$, the model density takes the form
(A) $p(\mathcal{D}|\theta):=\Pi_{i=1}^{N}{p(x_{i}|\theta, t_i) q(t_i)}$
if the $t_i$ are assumed to be i.i.d. with distribution $Q$ (and density $q$ with respect to some given base measure on $\mathcal{T}$), (B) and $p\big((x_1,\ldots,x_N)\big|\theta, (t_1,\ldots,t_N)\big):=\Pi_{i=1}^{N}{p(x_{i}|\theta, t_i)}$ if the $t_i$ are assumed to be arbitrary.
Observe that  the model densities for cases (A) and (B) are proportional and, as a consequence, given the $t_i$ (arbitrary or sampled), they share the same likelihood region
$ \Theta_{\mathcal{D}}(\alpha) = \bigl\{\theta \in \Theta: \Pi_{i=1}^N p(y_i|\theta, t_i)\geq \alpha \sup_{\theta'} \Pi_{i=1}^N p(y_i|\theta', t_i)\bigr\}$.
Let $Q_N:=\frac{\updelta_{t_1}+\cdots+\updelta_{t_N}}{N}\in \mathcal{P}(\mathcal{T})$ be the empirical probability distribution defined by $(t_1,\ldots,t_n)$. And assume $\mathcal{T}$ to be a compact subset of a finite dimensional Euclidean space.
 The following theorem indicates the result of Theorem \ref{thm_asympt} remains valid in case (B) if $Q_N\to Q$ (e.g. when $\mathcal{T}=[0,1]$ and $Q$ is the uniform distribution and the $t_i=i/N$).

\begin{Theorem}\label{thmasymcaseB}
Assume that  the model density $(y,t)\rightarrow p(y|\theta,t) q(t)$  satisfies the regularity conditions of Casella and Berger \cite[Section.~10.6.2]{casella2003statistical}, and that  $Q_N\to Q$ (in the sense of weak convergence)  as $N\to \infty$. Then in both cases (A) and (B) the limit
$\beta_\alpha \to 1-\chi^2_k\big(2\ln \frac{1}{\alpha}\big)$ holds true as $N\to \infty$.
\end{Theorem}

\section{Minimum enclosing ball algorithm}
\label{sec_miniball}
Let $K \subset \R^{n}$ be a compact subset and let
$B\supset K$, with center $z$ and radius $R$, be the smallest closed ball containing $K$. Together
 Theorem
\ref{thm_duala} and Theorem \ref{thm_red}  demonstrate that the minimum enclosing ball exists and is unique.
The problem of  computing the minimum enclosing ball has received a considerable amount of attention, beginning with Sylvester \cite{sylvester1857question} in 1857.
Probably the most cited method is that of Welzl \cite{welzl1991smallest}, which,
 by \cite[Thm.~2]{welzl1991smallest},  achieves the solution in expected $\mathcal{O}((n+1)(n+1)!|K|)$ time,
where $|K|$ is the cardinality of the set $K$.   Yildirim \cite{yildirim2008two}
provides two algorithms which  converge to an $\epsilon$-approximate minimum enclosing ball in $\mathcal{O}(\frac{|K|n}{\epsilon})$ computations and
provides a historical review of the literature along with extensive references.

Although Yildirim does address the infinite $K$ situation, we provide a new  algorithm, Algorithm
 \ref{Alg: Miniball}, based on
 that of B\u{a}doiu,  Har-Peled and Indyk \cite[p.~251]{badoiu2002approximate}, to  approximately compute the minimum enclosing ball $B$ containing a (possibly infinite)  compact set
 $K$ in $\R^{n}$, using
the approximate computation of maximal distances from the set $K$ to fixed points in $\R^{n}$.
To  that end, let  $\MINIBALL$ denote an existing algorithm for computing the minimum enclosing ball
for sets of size $ \leq n+2$. As we will demonstrate,
the FOR loop  in Algorithm \ref{Alg: Miniball} always gets broken at Step \ref{forloopbreak} for some $x$ since, by Caratheodory's theorem,
 see e.g.~Rockafellar \cite{rockafellar1970convex,},
a minimum enclosing ball in $n$ dimensions is always determined by $n+1$ points.

For $\delta \geq 0$, and a function $f:X \rightarrow \R$, let $\arg \max^{\delta}{f}$ denote a
 $\delta$-approximate maximizer in the following sense;   $x^{*} \in \arg \max^{\delta}{f}$ if
\[  f(x^{*}) \geq \frac{1}{1+\delta}\sup_{x \in X}{f(x)}.\]
For, $\epsilon >0$, the $\epsilon$-enlargement $B^{(1+\epsilon)}(x,r)$ of a closed ball $B(x,r)$ with center $x$ and radius $r$
is the closed ball $B(x,(1+\epsilon)r)$. In the following algorithm, $\delta$ is a parameter quantifying
the degree of optimality of distance maximizations and $\epsilon$ is a parameter specifying
the accuracy required of the produced estimate to the minimum enclosing ball.
\begin{algorithm}[!ht]
\begin{algorithmic}[1]
\STATE\label{step3glocs} Inputs: $\epsilon \in [0,1), \delta \geq 0$ , $ K$ and $ \MINIBALL$ for sets of size $\leq n+2$
\STATE   $(x_{\alpha},x_{\beta})  \leftarrow  \arg\max^{\delta}_{x,x'\in K}{\|x-x'\|}$
\STATE $A_{0} \leftarrow \{x_{\alpha},x_{\beta}\}$
\STATE   $0 \leftarrow k$ (\text{iteration counter})
\REPEAT
\STATE  \label{newball} $B_{k}\leftarrow \MINIBALL(A_{k})$
\WHILE {$|A_{k}| > n+1$}
\FOR  {$x \in A_{k}$}
\STATE {$ B^{x}_{k}\leftarrow \MINIBALL(A_{k}\setminus \{x\})$}
\STATE \label{forloopbreak} {if $B^{x}_{k}=B_{k}$
then $ A_{k} \leftarrow A_{k}\setminus \{x\}$ and break loop}
\ENDFOR
\ENDWHILE
\STATE $z_{k} \leftarrow \Center(B_{k})$
\STATE \label{distantpoint} $x_{k+1} \leftarrow  \arg \max^{\delta}_{x\in K}{\|x-z_{k}\|}$
\STATE  \label{newpoint} $A_{k+1} \leftarrow A_{k}\cup \{x_{k+1}\}$
\STATE    $k \leftarrow k+1$
\UNTIL  \label{untils} {$x_{k} \in B^{(1+\epsilon)}_{k-1}$}
\RETURN \label{returns} {$B^{(1+\epsilon)(1+\delta)}_{k-1},A_{k-1}$}
\end{algorithmic}
\caption{Miniball algorithm}\label{Alg: Miniball}

\end{algorithm}
The following theorem demonstrates that Algorithm \ref{Alg: Miniball} produces an approximation with guaranteed accuracy to the minimum enclosing ball in a quantified finite number of steps.
\begin{Theorem}
\label{thm_miniball}
For a  compact subset $K\subset \R^{n}$, let $R$ denote the radius of the minimum enclosing ball of $K$.
Then, for $\epsilon \in [0,1), \delta \geq 0$, Algorithm \ref{Alg: Miniball} converges to a ball $B^{*}$  satisfying
\[ B^{*} \supset K\] and
\[ R\bigl(B^{*}\bigr) \leq  (1+\epsilon)(1+\delta)R  \]
in at most $\frac{16}{\epsilon^{2}}(1+2\delta)$ steps of the REPEAT loop. Moreover, the size of the working set $A_{k}$ is bounded by
  \[|A_{k}|\leq \min {\Bigl(2+\frac{16}{\epsilon^{2}}(1+2\delta), n+2\Bigr)}\]
 for all $k$.
\end{Theorem}

\section{Examples}\label{sec:example}
\subsection{Gaussian Mean Estimation}
Consider the problem of estimating the mean $\theta^\dagger$ of a Gaussian distribution $\mathcal{N}(\theta^\dagger, \sigma^2)$
with known variance $\sigma^2>0$ from the observation  of one sample $x$ from that distribution and from the information that
$\theta^\dagger\in [-\tau,\tau]$ for some given $\tau>0$.
Note that this problem can be formulated in the setting of  Problem \ref{pb1} by letting  (1) $P(\cdot|\theta)$ be the Gaussian distribution on $X:=\R$ with mean $\theta$ and variance $\sigma^2$,
  (2) $\Theta:=[-\tau,\tau]$ and $V:=\R$  and (3) $\varphi:\Theta \rightarrow V$ be the identity map $\varphi(\theta)=\theta$.
The relative likelihood \label{def_rellikelihood} is 
\begin{equation}
\label{eq_surr}
  \bar{p}(x|\theta)= \frac{e^{-\frac{1}{2\sigma^{2}}|x-\theta|^{2}}}{\sup_{\theta \in \Theta} e^{-\frac{1}{2\sigma^{2}}|x-\theta|^{2}}}
\end{equation}
with the supremum in the denominator achieved at the closest $\theta \in \Theta$ to $x$ ($x$ itself if $x \in [-\tau, \tau]$). This defines the likelihood region
$\Theta_{x}(\alpha):=\{\theta \in \Theta:
 \bar{p}(x|\theta) \geq  \alpha\}$. A simple calculation yields, for the case $x \in [-\tau, \tau]$
\begin{equation}
\Theta_x(\alpha) = \Bigl[\max\bigl(-\tau, x-\sqrt{2\sigma^2\ln(1/\alpha)}\bigr)\,,\, \min\bigl(\tau, x+\sqrt{2\sigma^2\ln(1/\alpha)}\bigr)\Bigr].
\end{equation}

Using Theorem \ref{pioptreduced} with $m=dim(V)+1=2$, for $\alpha \in [0,1]$, one can compute a saddle point $(\pi^\alpha,d^\alpha)$ of the game \eqref{Lalphadefa2}
as
\begin{equation}
\pi^\alpha=w \updelta_{\theta_{1}}+(1-w) \updelta_{\theta_{2}}\text{ and } d^\alpha=w\theta_{1}+(1-w) \theta_{2}
\end{equation}
where $w,  \theta_{1}, \theta_2$  maximize the variance
\begin{equation}\label{Opt:Gaussian}
\begin{cases}
    \text{Maximize}&\quad w \theta_1^2 + (1 - w) \theta_2^2 - \left( w \theta_1 + (1 - w) \theta_2 \right)^2 \\
    \text{over}& \quad 0 \leq w \leq 1,\quad \theta_1, \theta_2 \in [-\tau, \tau] \\
    \text{subject to}& \quad  \frac{(x - \theta_i)^2}{2 \sigma^2} \leq \ln \frac{1}{ \alpha}, \qquad i = 1, 2,
\end{cases}
\end{equation}
where the last two constraints are equivalent to the
rarity assumption $ \theta_{i}\in \Theta_{x}(\alpha)$.

Hence for $\alpha$ near $0$, $\Theta_{x}(\alpha) =\Theta = [-\tau, \tau]$, and by Theorem \ref{thm_duala},   the variance is maximized by placing each Dirac on each boundary point of the  region $\Theta$, each receiving half of the total probability mass, that is by $\theta_1=-\tau$, $\theta_2=\tau$ and $w=1/2$, in which case $\Var{\pi^\alpha}=\tau^2$ and $d^\alpha=0$.
For $\alpha=1$, the rarity constraint implies $\theta_1=\theta_2=x$ when $x \in [-\tau,\tau]$, leading to the
 MLE $d^\alpha=x$ with $\Var{\pi^\alpha}=0$.
Note that from \eqref{betadef} we have
\begin{align*}
    \beta_{\alpha} &= \sup_{\theta \in [-\tau,\tau]} \mathbb{P}_{x' \sim \mathcal{N}(\theta,\sigma^2)} \big[ \bar{p}(x' | \theta) <  \alpha \big]
\end{align*}
which can be computed analytically for this example using \ref{eq_surr} and separating into the three cases $x < -\tau$, $x \in [-\tau, \tau]$ and $x > \tau$. We illustrate in Figure \ref{fig:4plotsnormalexample} the different results of solving the optimization problem \eqref{Opt:Gaussian} in the case $\sigma^2 = 1$, $x = 1.5$ and $\tau = 3$. We plot the $\alpha - \beta$ curve (top left), the likelihood of the model in $[-\tau, \tau]$ and the $\alpha-$level sets (top right), and the evolutions of the risk with $\beta$ (bottom left), and the optimal decision with $\beta$ (bottom right). Since, by Theorem \ref{thm_duala}, the optimal decision is the midpoint of the interval with extremes in either the $\alpha-$level sets or $\pm \tau$, we observe that for low $\beta$, our optimal decision does not coincide with the MLE.
\begin{figure}
    \centering
    \includegraphics[width = 0.8\textwidth]{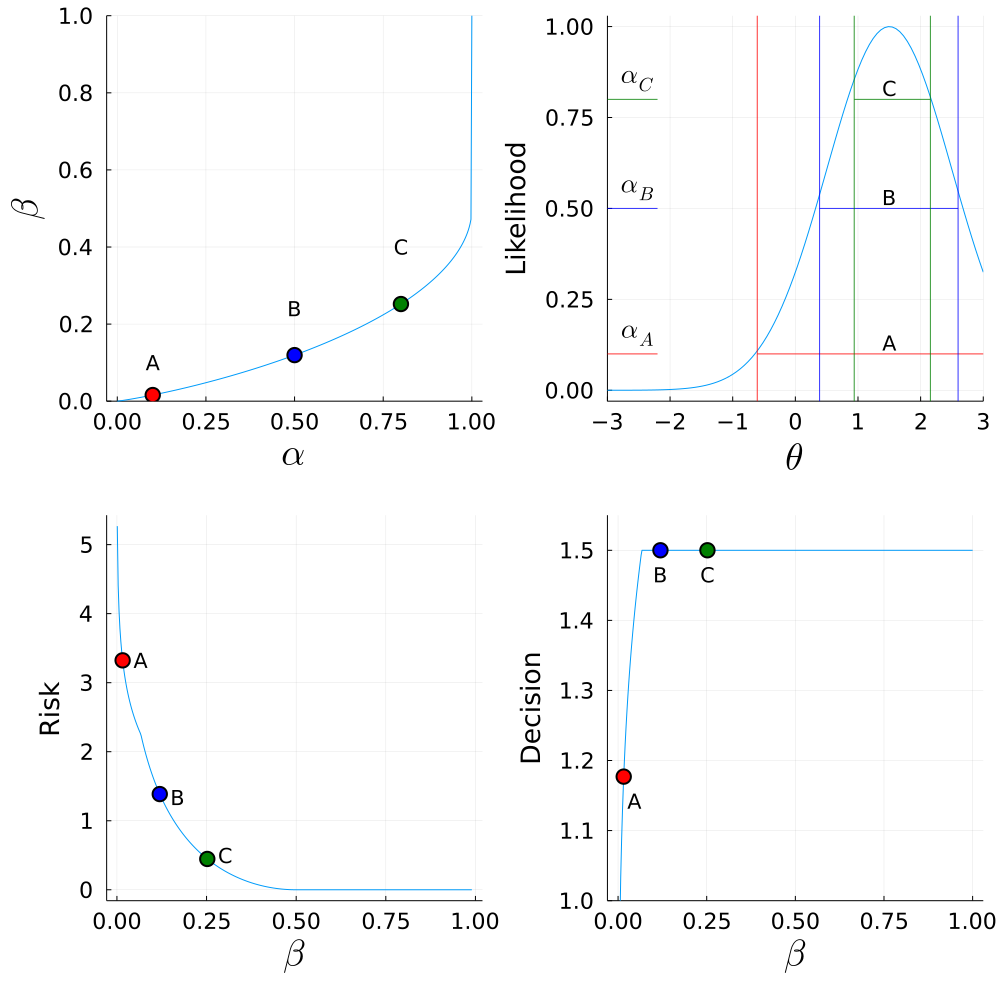}
    \caption{$\alpha - \beta$ relation, likelihood level sets, risk value and decision for different choices of $\alpha$ (and consequently $\beta$) for the normal mean estimation problem with $\tau = 3$ and observed value $x = 1.5$. Three different values in the $\alpha-\beta$ curve are highlighted across the plots}
    \label{fig:4plotsnormalexample}
\end{figure}

\subsection{Estimation of a quadratic function}
\label{sec_gradient}

The measurement function $\boldsymbol{m}$ of Section \ref{sec_LV}, being defined as the solution of the Lotka-Volterra predator-prey model as a function of its parameters $\boldsymbol{\theta}\in \Theta$,  does not appear simple to differentiate and therefore \texttt{SciPy}'s version
of Storn and Price's \cite{storn1997differential} Differential Evolution optimizer  \cite{2020SciPy-NMeth} was used to perform the farthest point optimization problem in Step \ref{step8a}
in Algorithm \ref{Alg: ODAD General}. In this section, we test this framework on a problem
which does not possess this complication; estimating the parameters  $\boldsymbol{\theta}:=
(\theta_0,\theta_1,\theta_2)$ of a  quadratic function
\[  \boldsymbol{m}(t;\boldsymbol{\theta}):=\theta_0 + \theta_1 t + \theta_2 t^2 \]
 on a uniform grid $T$ of the interval $(0,5)$
consisting of $100$ points,
using noisy observational data.
In this case, we can use automatic differentiation  in the merit function technique
of Section \ref{sec_merit}
 to perform the farthest point optimization problem in Step \ref{step8a} using gradient descent methods  via automatic differentiating modules  available in packages like autograd, or
 computing the gradient and applying a gradient descent method.

We proceed as in Section \ref{sec_LV}
 with $\Theta:=[-30,30]^{3}$
 and   assume that, given
$\boldsymbol{\theta} \in \Theta$, a single sample path
$\mathcal{D}:= \{\boldsymbol{x}\}:=\bigl(\boldsymbol{x}^{(t)}\bigr)_{t=1}^{100}$
is generated on the grid $T$
to the stochastic process
\begin{equation} \label{quad_stat_model}
    \boldsymbol{x}^{(t)} =  \boldsymbol{m}(t;\boldsymbol{\theta}) + \epsilon_t, \quad \epsilon_t \sim \mathcal{N}(\mathbf{0}, \sigma^2 ), \quad t \in T,
\end{equation}
with $\sigma^{2}:=10$. Consequently $X:=\R^{100}$.
We let the decision space be $V=\R^{3}$ and the quantity of interest $\varphi:\Theta \rightarrow \R^{3}$ be the identity function.

For the experiment, we generate a single ($N:=1$) full sample path $\mathcal{D}:=\{\boldsymbol{x}\}:=\bigl(\boldsymbol{x}^{(t)}\bigr)_{t=1}^{100}$
 according to \eqref{quad_stat_model} at the unknown values
$\boldsymbol{\theta}^{*}:=(\theta^{*}_{1},\theta^{*}_{2}, \theta^{*}_{3})=(1,.5,1)$.
 As discussed in Section \ref{sec_merit},  we tune $\mu$ in the merit
 function \eqref{eq: Merit Function}, and set gradient descent with adaptive moment estimation optimizer \cite{pytorch} with parameters (e.g. learning rate = 0.001, max epochs=50,000) to achieve full convergence.
We observe the convergence plots for the increment of $\theta$s from the ball center as diagnostic.

Figure \ref{fig:quad_beta_alpha} shows the $\beta$ vs $\alpha$ relationship
defined by  the surrogate significance \eqref{beta_chi} derived from the
surrogate likelihood method
 with $N:=1$ and $r:=100$,
 the risk as function of $\alpha $,  and the likelihood regions, their minimum enclosing balls and the optimal decisions (centers of the balls), for $\alpha \in [0,1]$. As can be seen the optimal decisions, being the centers of the minimum enclosing balls, do not move and only the size of the minimum enclosing balls change, resulting
in various risk values associated with the same optimal estimates.

Finally, Figure \ref{fig:quad_beta_alpha} shows the results for the maximum likelihood solution and  the two supporting points of the minimum enclosing balls for the case that $\beta_{\alpha}=\beta^{*} =0.05 $. From this experiment, we obtained the optimal decision
$d^{*}=(0.24,1.22,0.89)$ along with the two support points $S=\bigl\{(-2.27 ,  3.52,  0.48), (2.75, -1.10,  1.31)\bigr\}$ of the minimum enclosing ball. For the sake of comparison, we also performed the same experimentation with
\texttt{SciPy}'s version of Storn and Price's \cite{storn1997differential} Differential Evolution optimizer  \cite{2020SciPy-NMeth}
 at  the  default  settings,  to
perform the farthest point optimization problem in Step \ref{step8a}
in Algorithm \ref{Alg: ODAD General}, using
 the merit function \eqref{eq: Merit Function} of   Section \ref{sec_merit},
and obtained similar results.

\begin{figure}[htp]
    \centering
    \includegraphics[scale=0.5]{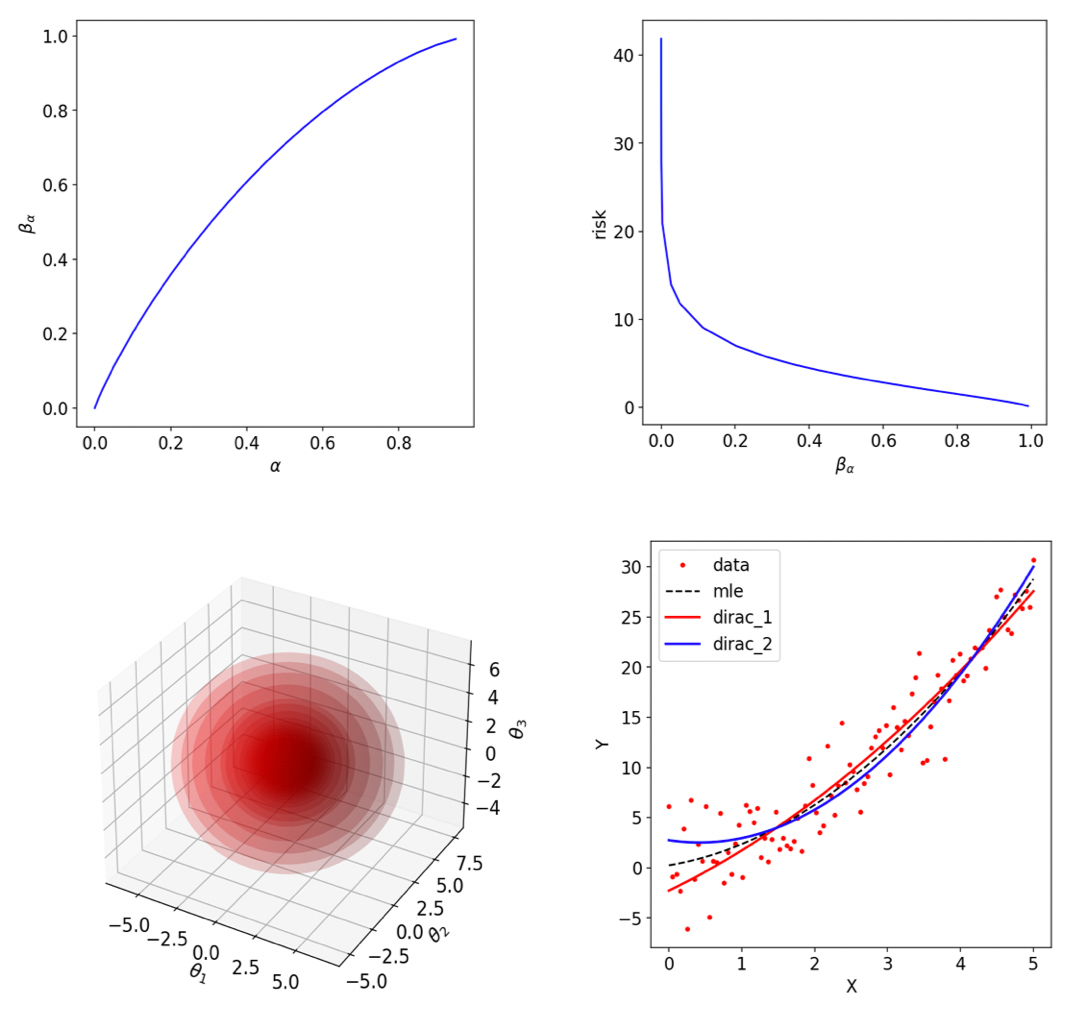}

    \caption{Quadratic model results:  (left-top)  $\beta_\alpha$ vs $\alpha$,
  (right-top)  risk vs  $\beta_\alpha$, (left-bottom)
  the supporting surfaces with the decision points at the center of each surface and (right-bottom) maximum-likelihood solution and supporting points of the minimum enclosing ball for $\beta_{\alpha}=\beta^{*} =.05 $ }
    \label{fig:quad_beta_alpha}
\end{figure}

\subsection{Estimation of a Lotka-Volterra predator-prey model}

\label{sec_LV}
Here we implement Algorithm \ref{Alg: ODAD General} for the Gaussian noise model of
 Section \ref{sec_gaussian},
where the measurement function $ (\theta_{1},\theta_{2}) \mapsto
 \boldsymbol{m}(\theta_{1},\theta_{2})$
is defined as the solution  map of the
 Lotka-Volterra \cite{lotka}  predator-prey model
\begin{align}
    \frac{dx}{dt} &= \theta_1 x - \eta xy \\
    \frac{dy}{dt} &= \xi xy - \theta_2 y,
\end{align}
evaluated on the uniform time grid
 $T := \{ t_i \}_{i=1}^{200}$ such that $t_0=0$ and $t_{200} = 20$,
with fixed and known  parameters $\eta, \xi$ and  initial data $x_{0},y_{0}$,
 describing the evolution of  a prey population with variable $x$ and a predator population with variable $y$. As such, denoting $\boldsymbol{\theta}:=(\theta_{1},\theta_{2}) \in \Theta$,
 we denote the solution map
 \[ \boldsymbol{\theta} \mapsto
\boldsymbol{m}(t;\boldsymbol{\theta}), t \in T, \]
by
 $\boldsymbol{m}:\Theta \rightarrow  (\R^{2})^{T} $.
For the
 probabilistic model, we let $\Theta:=[-5,5]^{2}$ and   assume  the Gaussian model \eqref{eq: Likelihood} with $N=1$, where
 the data $\mathcal{D}$ consists of a single sample path
$\mathcal{D}:=\{\boldsymbol{x}\}:=\bigl(\boldsymbol{x}^{(t)}\bigr)_{t=1}^{200}$
of  the $T$-indexed stochastic process
\begin{equation} \label{lot_vol_stat_model}
    \boldsymbol{x}^{(t)} = \begin{bmatrix} x_t \\ y_t \end{bmatrix} = \boldsymbol{m}(t;\boldsymbol{\theta}) + \epsilon_t, \quad \epsilon_t \sim \mathcal{N}(\mathbf{0}, \sigma^2 \mathbf{I}), \quad t \in T,
\end{equation}
where $\mathbf{I} \in \mathbb{R}^{2 \times 2}$ is the identity matrix and $\sigma = 5$.
Note that in the notation of \eqref{eq: Likelihood} we have
\[ \| \boldsymbol{x}-\boldsymbol{m}\left(\boldsymbol{\theta}\right)\|^{2}
=\sum_{t\in T}{\| \boldsymbol{x}^{(t)}-\boldsymbol{m}\left(t;\boldsymbol{\theta}\right)\|^{2}}.\]

Let the decision space be $V:=\R^{2}$ and
let the quantity of interest $\varphi:\Theta \rightarrow \R^{2}$ be the identity.
For the experiment, we generate one sample path $\mathcal{D}:=\{\boldsymbol{x}\}:=\bigl(\boldsymbol{x}^{(t)}\bigr)_{t=1}^{200}$
 according to \eqref{lot_vol_stat_model} at the unknown values
$\boldsymbol{\theta}^{*}:=(\theta^{*}_{1},\theta^{*}_{2})=(0.55,0.8)$, with
  $x_0 = 30$, $y_0 = 10$,  $\eta = 0.025$ and $\xi = 0.02$ known.
  We consider the evolution $ t \mapsto \boldsymbol{m}(t;\boldsymbol{\theta}^{*}),\, t \in T,$
  the {\em true predator-prey} values and the sample path $\bigl(\boldsymbol{x}^{(t)}\bigr)_{t=1}^{200}$
  as noisy observations of it.
   The resulting time series $\mathcal{D}$  is shown in the top image of Figure \ref{fig:lot_vol_uq_plot}.

For $\alpha \in [0,1]$,
 by taking  the logarithm of the defining relation \eqref{eqnormrarset} of the
 likelihood region $\Theta_{\mathcal{D}}(\alpha)$, we obtain the representation  \eqref{eq: Rarity for Gaussian},
\[\Theta_{\mathcal{D}}(\alpha) = \Bigl\{ \boldsymbol{\theta} \in \Theta : \sum_{t \in T}\bigl\|\boldsymbol{x}^{\left(t\right)}-\boldsymbol{m}\left(t; \boldsymbol{\theta}\right)\bigr\| ^{2}\leqslant M_{\alpha} \Bigr\}  \]
in terms of
\[  M_{\alpha}:=\inf_{\boldsymbol{\theta} \in \Theta}{\sum_{t\in T}\bigl\| \boldsymbol{x}^{\left(t\right)}-\boldsymbol{m}\left(t; \boldsymbol{\theta}\right)\bigr \|^{2}} +2\sigma^{2} \ln \frac{1}{\alpha}\, , \]
for the likelihood region $\Theta_{\mathcal{D}}(\alpha)$ in terms of the data
 $\mathcal{D}$.

To  determine $\alpha$  at significance level $\beta^{*}:=.05$, we approximate
the significance $\beta_{\alpha}$  defined in
 \eqref{betadef}
using  the chi-squared approximation \eqref{beta_chi}
and then select $\alpha$ to be the value such that this approximation yields  $\beta_{\alpha}=\beta^{*}=.05$. The validity of this approximation for this example is additionally demonstrated in the right image of Figure \ref{fig:chisq_mc_comparison}, which shows via Monte Carlo simulation that the $1-\beta_\alpha$ versus $\alpha$ curve is well characterized by the $\chi^2_2$ distribution.

Having selected $\alpha$,   to  implement Algorithm \ref{Alg: ODAD General}, we need to select an optimizer for
Step \ref{step8a}. Instead of computing the Jacobian of the
  solution map $\boldsymbol{m}$, here we utilize the gradient-free method
of \texttt{SciPy}'s version of Storn and Price's \cite{storn1997differential} Differential Evolution optimizer  \cite{2020SciPy-NMeth}
 at  the  default  settings.
 Given the data generating value $\boldsymbol{\theta}^{*}=  (0.55,0.8)$,
 the primary feasible region $\Theta:=[-5,5]^2$ is sufficiently non-suggestive of the
 the data generating value $\boldsymbol{\theta}^{*}$.
Finally, since $dim(V)=2$,  Algorithm  \ref{Alg: ODAD General} produces a set  $S$
 of at most three  boundary points of $\Theta_{\mathcal{D}}(\alpha)$,
 the minimum enclosing ball $B$  of  $\Theta_{\mathcal{D}}(\alpha)$, its center as the
optimal estimate of $\boldsymbol{\theta}^{*}$ and the weights
of the set $S$ corresponding to a worst-case measure, optimal for the variance maximization problem
\eqref{piopt}. The results  are displayed  in Figure \ref{fig:lotka_volterra_min_enclosing_result}.
\begin{figure}[htp]
    \centering
    \includegraphics[width= 0.8 \textwidth]{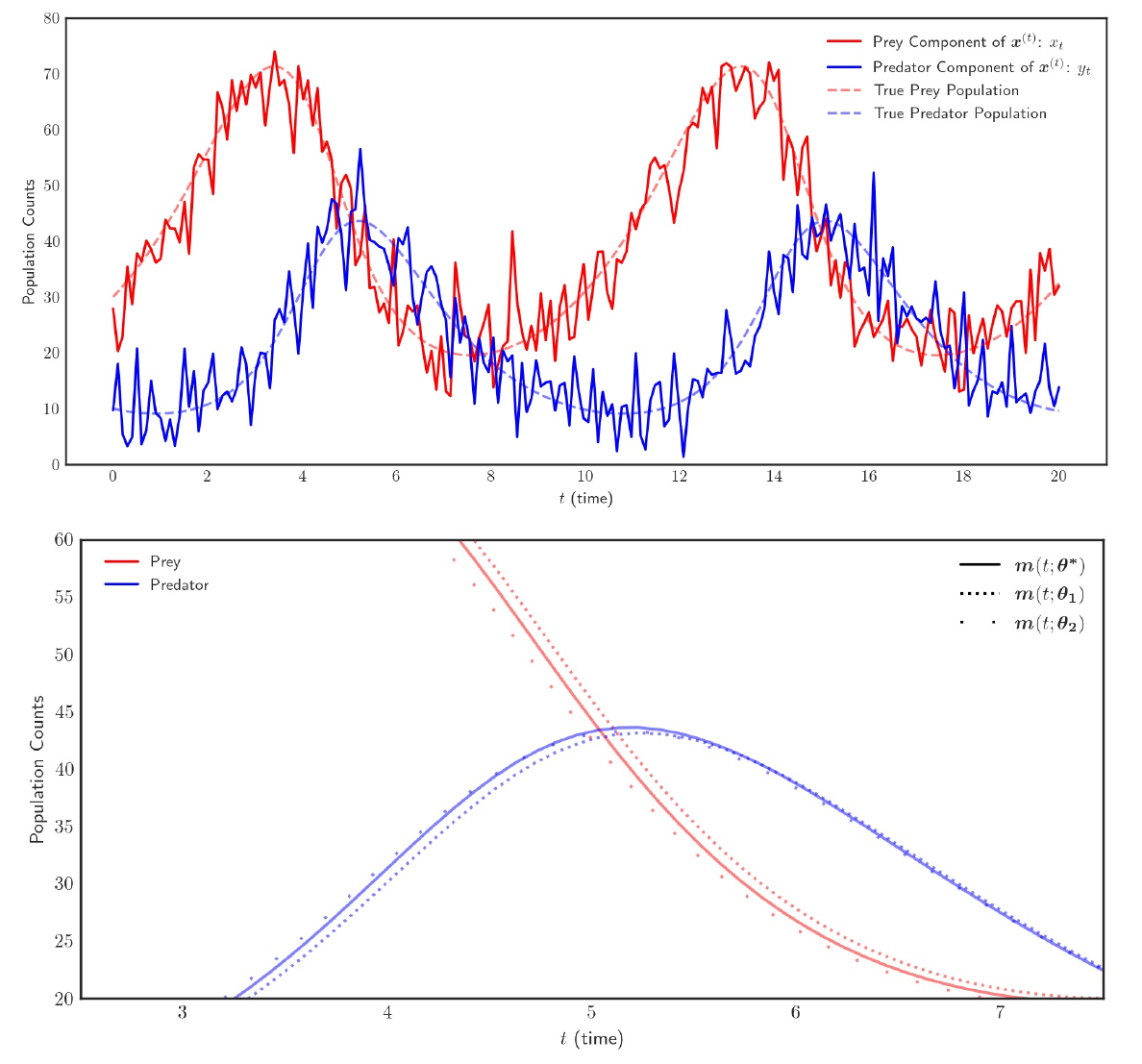}
    \caption{\textbf{(Top image)} Data $\mathcal{D}:=\boldsymbol{x}^{(t)}, t \in T$, generated,
 according to the Gaussian noise model \eqref{lot_vol_stat_model}:
 solid red
 is the prey component and solid blue the predator component of the generated data
$\boldsymbol{x}^{(t)}$,  the  dotted red and dotted blue are  the prey-predator
 components of the Lotka-Volterra solution
$\boldsymbol{m}(t,\boldsymbol{\theta}^{*})$ for  $t \in T$. \textbf{(Bottom image)} Uncertainty in the population dynamics corresponding to the worst-case measure:
  (1) red is prey and blue is predator,
(2) solid line  is the Lotka-Volterra evolution $ \boldsymbol{m}(t,
\boldsymbol{\theta}^{*}), t \in T $, fine dots  $ \boldsymbol{m}(t,
\boldsymbol{\theta}_{1}), t \in T,$ and coarse dots  $ \boldsymbol{m}(t,
\boldsymbol{\theta}_{2})$, where $S:=\{\boldsymbol{\theta}_{1}, \boldsymbol{\theta}_{2}\}$ is the set of support points of the worst case (posterior) measure (located on the boundary of the minimum enclosing ball).}
    \label{fig:lot_vol_uq_plot}
\end{figure}

\begin{figure}[htp]
    \centering
    \includegraphics[width= \textwidth]{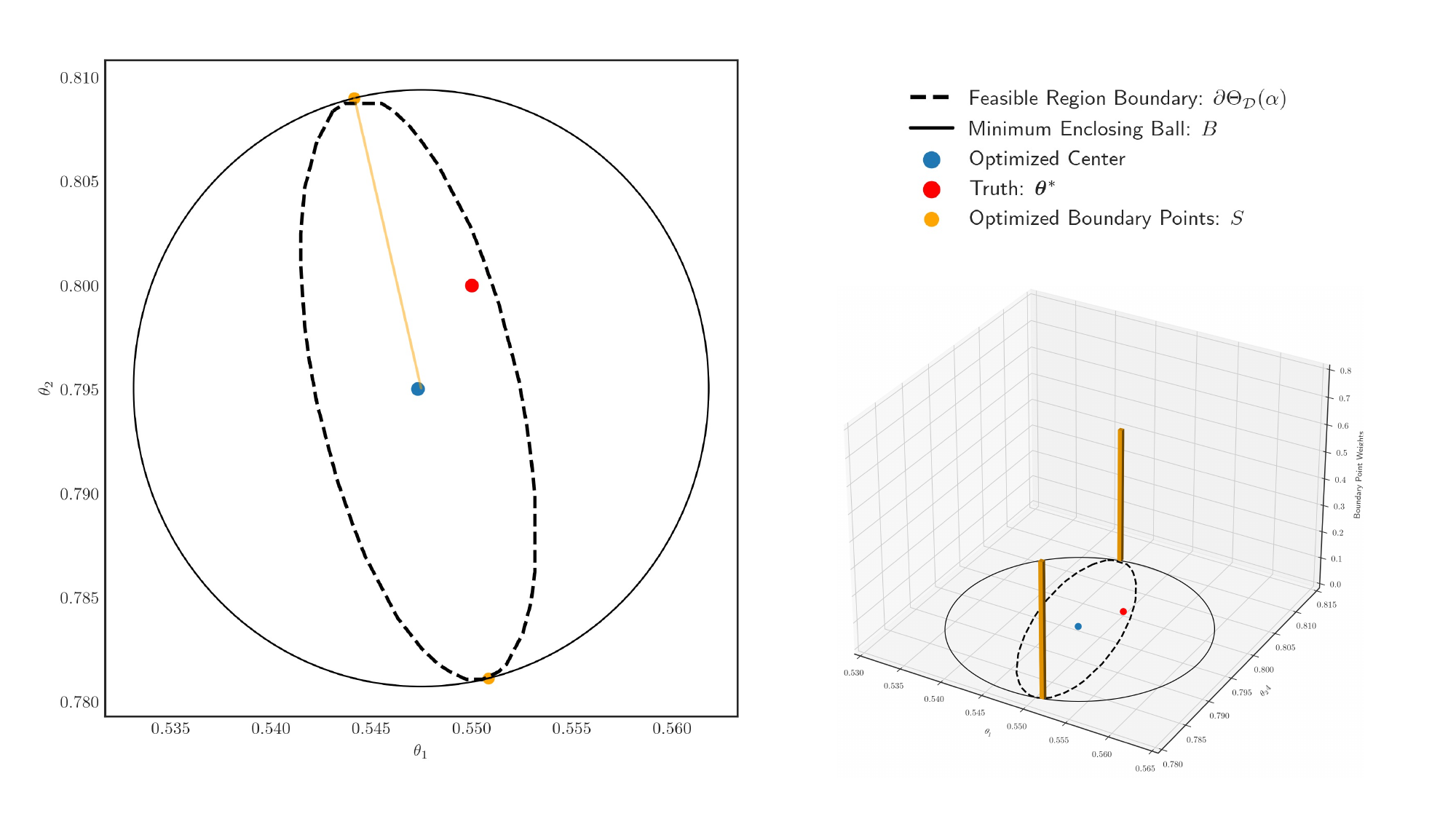}
    \caption{Minimum enclosing ball for the  Lotka-Volterra Model: (left)
the dashed line indicates the boundary
of the likelihood region $\Theta_{\mathcal{D}}(\alpha)$,  the solid circle its minimum enclosing ball,
the red point the data generating value $\boldsymbol{\theta}^{*}$ and the blue point the center of the minimum enclosing ball and the optimal estimate of $\boldsymbol{\theta}^{*}$. The two yellow points
 comprise the set $S:=\{\boldsymbol{\theta}_{1}, \boldsymbol{\theta}_{2}\}$. (right) a projected view with the yellow columns indicating the weights $(.5,.5)$ of the set $S$ in their determination of a worst-case (posterior) measure.}
    \label{fig:lotka_volterra_min_enclosing_result}
\end{figure}
\begin{figure}[htp]
    \centering
    \includegraphics[width= 0.6\textwidth]{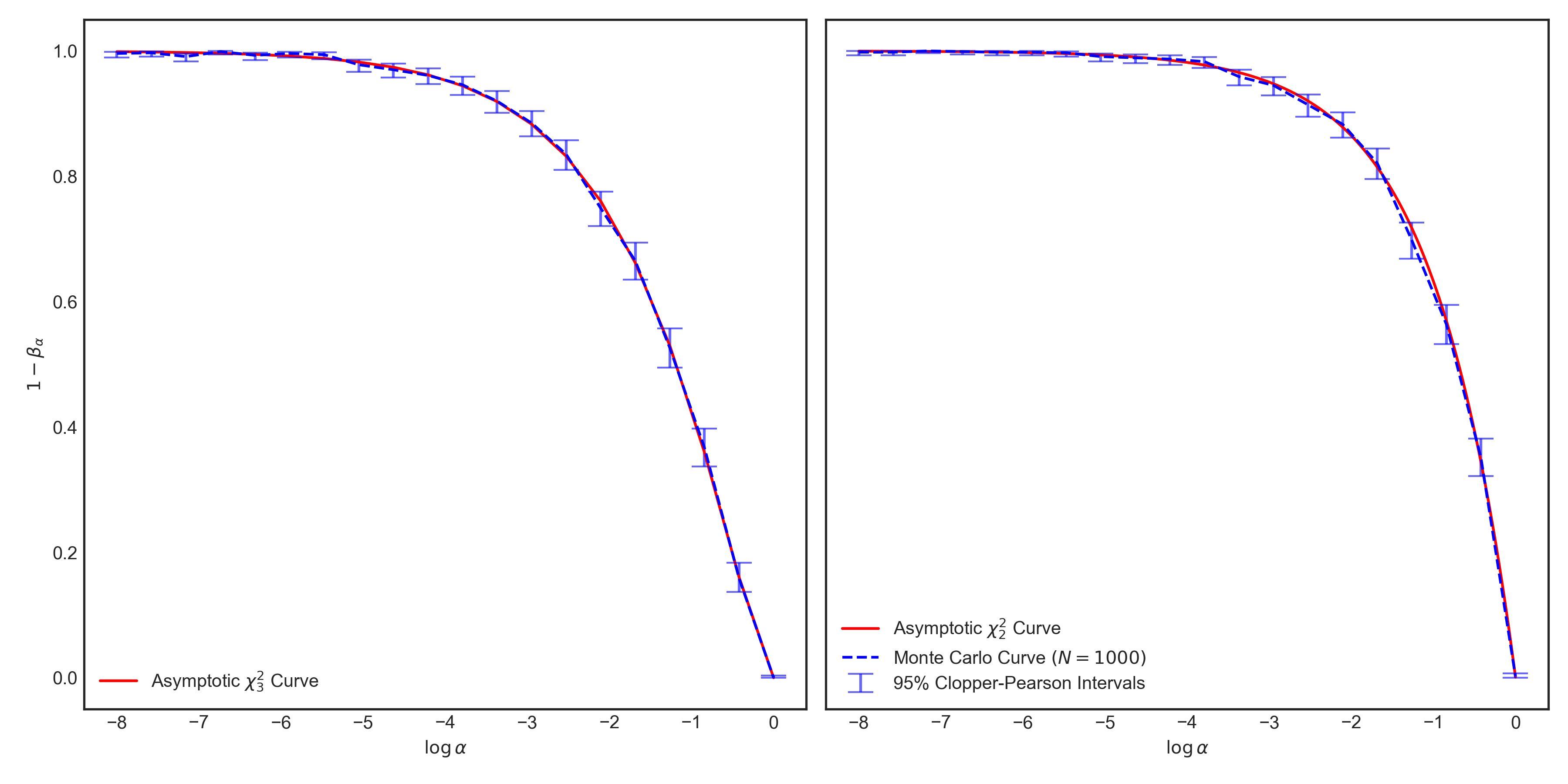}
    \caption{Monte Carlo numerically confirms the result from Theorem \ref{thm_asympt} in both the quadratic function estimation \textbf{(left image)} and Lotka-Volterra \textbf{(right image)} examples.}
    \label{fig:chisq_mc_comparison}
\end{figure}

To get a sense of the output uncertainty of $\boldsymbol{m}(\cdot)$ with these optimized results, we plot the predator and prey population dynamics associated with each optimized boundary point of
 $\Theta_{\mathcal{D}}(\alpha)$ in Figure \ref{fig:lot_vol_uq_plot}. This figure shows that with significance value $\beta_{\alpha} = 0.05$, the optimized boundary points create population dynamics in a tight band around the true population dynamics.

\section{General loss functions and rarity assumptions}
\label{sec_generalloss}
Here we generalize the framework introduced in Section \ref{sec_intro} to allow more general loss functions than the $\ell^{2}$ loss,
in for example Equations \eqref{worstcase}, \eqref{risk_bayes}, and \eqref{LalphadefW0l2}, and more general rarity assumptions than
\eqref{rarityass}.

The discussions of worst-case, robust Bayes, and Wald's statistical decision theory generalize in a straightforward manner, so we focus
on generalizing the current UQ of the 4th kind.
Let
$\ell:V \times V \rightarrow \R_{+}$ be a loss function.
 In addition to the pointwise rarity assumption \eqref{rarityass}, consider an {\em integral rarity assumption}
 $\mathcal{P}^{\Psi}_{x}(\alpha) \subset \mathcal{P}(\Theta)$ determined by a real-valued function $\Psi$, defined by
\begin{equation}
\label{PalphaH0}  \mathcal{P}^{\Psi}_{x}(\alpha):=
\Bigl\{\pi \in  \mathcal{P}(\Theta): \int_{\Theta}{\Psi\bigl(\bar{p}(x|\theta)\bigr)d\pi(\theta)} \geq \alpha\Bigr\}
 \end{equation}
generalizing \eqref{Palpha1},
Note that for $\Psi$ the identity function $ \mathcal{P}_{x}(\alpha) \subset  \mathcal{P}^{\Psi}_{x}(\alpha)$ and by comparison with  \eqref{Palpha1}
it follows from the following remark that such integral rarity assumptions  can also alleviate the brittleness of Bayesian inference.
\begin{Remark}
Jensen's inequality implies
that
\[\Psi\Bigl(\int_{\Theta}{\bar{p}(x|\theta)d\pi(\theta)}\Bigr) \leq \int_{\Theta}{\Psi\bigl(\bar{p}(x|\theta))\bigr)d\pi(\theta)}
\]
when the function $\Psi$ is convex,
and
\[\Psi\Bigl(\int_{\Theta}{\bar{p}(x|\theta)d\pi(\theta)}\Bigr)\geq
 \int_{\Theta}{\Psi\bigl(\bar{p}(x|\theta))\bigr)d\pi(\theta)}
\]
when the function $\Psi$ is concave, such as when $\Psi$ is a logarithm.
Consequently, when $\Psi$ is concave and strictly increasing
the assumption
\[\int_{\Theta}{\Psi\bigl(\bar{p}(x|\theta)\bigr)d\pi(\theta)} \geq \alpha\]
implies that the denominator
in the conditional measure \eqref{eq_conditional} satisfies
\[\int_{\Theta}{\bar{p}(x|\theta)d\pi(\theta)} \geq \Psi^{-1}( \alpha)\, .\]
Consequently, such constraints, by keeping the denominator in the conditional measure bound away from zero stabilize the numerical computation of the conditional measure in the numerical computation of a worst-case measure.
\end{Remark}

The following applies equally as well for the pointwise rarity assumption
\eqref{rarityass} and the integral rarity assumption \eqref{PalphaH0}. For simplicity of exposition, we restrict
to the pointwise rarity assumption.
 Generalizing \eqref{Lalphadefa}, consider playing a game
using the loss
\begin{equation}
\label{Lalphadef0}
 \L(\pi,d):=\E_{\theta \sim \pi_{x}}\bigl[\ell(\varphi(\theta),d)\bigr], \quad    \pi \in
  \mathcal{P}_{x}(\alpha), d \in V\,  .
\end{equation}
For the pointwise rarity assumption, the same logic following \eqref{Lalphadefa} implies that this game is equivalent to
the generalization \eqref{Lalphadefa2}
to
a game
using the loss
\begin{equation}
\label{Lalphadef}
 \L(\pi,d):=\E_{\theta \sim \pi}\bigl[\ell(\varphi(\theta),d)\bigr], \quad    \pi \in
  \mathcal{P}_{x}(\alpha), d \in V\,  .
\end{equation}
For the integral rarity assumption, we maintain the form \eqref{Lalphadef0}.
$\beta_{\alpha}$ is defined before as in \eqref{betadef}
and  the selection of
  $\alpha \in [0,1]$ is as before.

Under  the mild conditions of Theorem \ref{thm_minmax}, we can show that, for each $\alpha \in [0,1]$,  a maxmin optimal solution $\pi^{\alpha}$ of $ \max_{\pi \in \mathcal{P}_{x}(\alpha) }\min_{d \in V}{\L(\pi,d)}$
can be computed.
Moreover, by Theorem \ref{thm_minmax}, it also follows that the saddle function $\L$ in \eqref{Lalphadef}
satisfies the conditions of Sion's minmax theorem \cite{sion1958general}, resulting
 in a minmax
result for the  game \eqref{Lalphadef};
\begin{equation}
\label{Lalphaminmax}
\min_{d \in V}\max_{\pi \in \mathcal{P}_{x}(\alpha)}{\L(\pi,d)}=\max_{\pi \in \mathcal{P}_{x}(\alpha)}
\min_{d\in V}{\L(\pi,d)}\,.
\end{equation}
Consequently if, for each $\pi \in \mathcal{P}_{x}(\alpha)$,  we select
\begin{equation}
\label{dpi}
 d_{\pi} \in  \arg \min_{d \in V}{\L(\pi,d)},
\end{equation}
then a worst-case measure for the game \eqref{Lalphadef}, solving the maxmin problem on the right-hand side of \eqref{Lalphaminmax},  satisfies
\begin{equation}
\label{pistar}
 \pi^{\alpha}\in \arg \max_{\pi \in \mathcal{P}_{x}(\alpha)}{\L(\pi,d_{\pi})} \, .
\end{equation}
Moreover let $d^{\alpha} \in \arg \min_{d \in V}\max_{\pi \in \mathcal{P}_{x}(\alpha)}{\L(\pi,d)}$ denote any solution to
the minmax problem on the left-hand side of \eqref{Lalphaminmax}.
Then it is well known that the minmax equality \eqref{Lalphaminmax} implies that  the pair $(\pi^{\alpha},d^{\alpha})$ is a saddle
point of $\L$ in that
\begin{equation}
    \label{eq_saddle3}
\L(\pi, d^{\alpha}) \leq \L(\pi^{\alpha},d^{\alpha}) \leq \L(\pi^{\alpha},d),\qquad  \pi  \in \mathcal{P}_{x}(\alpha), \,\,  d \in V\, .
\end{equation}
Since the solution  to
\begin{equation}
\label{dpistar}
d_{\pi^{\alpha}}:=\arg \min_{d \in V}{\L(\pi^{\alpha},d)}\,
\end{equation}
is uniquely defined under the conditions of Theorem \ref{thm_red} (identical to those of Theorem \ref{thm_minmax}), it follows from the right-hand side of the saddle equation \eqref{eq_saddle}
that $d_{\pi^{\alpha}}=d^{\alpha}$ and $(\pi^{\alpha},d_{\pi^{\alpha}})$ is a saddle
point of $\L$, that is we have
\begin{equation}
    \label{eq_saddle2}
\L(\pi, d_{\pi^{\alpha}}) \leq \L(\pi^{\alpha},d_{\pi^{\alpha}}) \leq \L(\pi^{\alpha},d),\qquad  \pi  \in \mathcal{P}_{x}(\alpha), \,\,  d \in V\, .
\end{equation}
Moreover, its associated risk is the value
\begin{equation}
\label{Rstar}
\mathcal{R}(d_{\pi^{\alpha}}):=\L(\pi^{\alpha},d_{\pi^{\alpha}})\,
\end{equation}
in \eqref{eq_saddle2} of the  two person game defined  in \eqref{Lalphadef}, which is the same for all saddle points of $\L$.

\subsection{Finite-dimensional reduction}
\label{sec_reductionstrict}
Let $\Delta^{m}(\Theta)$ denote the set of convex sums of $m$ Dirac measures located in $\Theta$ and,
 let $\mathcal{P}^{m}_{x}(\alpha) \subset \mathcal{P}_{x}(\alpha)$ defined by
\begin{equation}
\label{PalphaHm0}
 \mathcal{P}^{m}_{x}(\alpha):=
\Bigl\{ \pi  \in  \Delta^{m}(\Theta) \cap \mathcal{P}_{x}(\alpha) \Bigr\}
\, ,
 \end{equation}
denote the  finite-dimensional subset of the  rarity assumption set consisting of
the convex combinations of $m$ Dirac measures in $ \mathcal{P}_{x}(\alpha)$. Then the following reduction Theorem \ref{thm_red} asserts that
\begin{equation}
\label{opt_ODAD202}
\max_{\pi \in \mathcal{P}_{x}(\alpha)} \min_{d\in V}{\L(\pi,d)}=\max_{\pi \in  \mathcal{P}^{m}_{x}(\alpha)}\min_{d\in V}{\L(\pi,d)}
,
\end{equation}
for any $m\geq dim(V)+2$. We note that the improvement to $m\geq dim(V)+1$  when the loss is the $\ell^{2}$ loss follows from
the Lagrangian duality, Theorem \ref{thm_duala},  of the maximization problem \eqref{opt_ODAD202} with the minimum enclosing ball and Caratheodory's theorem.

In the following theorem, applicable to both pointwise and integral rarity assumptions, we provide sufficient conditions that the computation of a worst-case measure in the optimization problem \eqref{pistar}  can be reduced to a finite-dimensional one.
\begin{Theorem}
\label{thm_red}
Let $\Theta$  be compact, $X$ be a measurable space, $\P$ be a positive dominated model such that, for each $x$, its  likelihood function is continuous, and let $\varphi:\Theta \rightarrow \R$ be continuous.
Let $V$ be a finite-dimensional Euclidean space and
let $\ell:V\times V \rightarrow \R_{+}$ be  continuously differentiable,  strictly convex and coercive  in its second variable  and vanishing along the diagonal.
Let $\Delta^{m}(\Theta)$ denote the set of convex sums of $m$ Dirac measures located in $\Theta$ and,
 for $\Psi:(0,1]\rightarrow \R$ upper semicontinuous and $x \in  X$,  consider both the
pointwise rarity assumption subset $  \mathcal{P}^{m}_{x}(\alpha)\subset \mathcal{P}_{x}(\alpha)$
 defined by
\begin{equation}
\label{PalphaHm2}
 \mathcal{P}^{m}_{x}(\alpha):=
\Bigl\{ \pi  \in  \Delta^{m}(\Theta) \cap \mathcal{P}_{x}(\alpha) \Bigr\}
\, ,
 \end{equation}
 and the integral rarity assumption subset
$  \mathcal{P}^{\Psi,m}_{x}(\alpha)\subset \mathcal{P}^{\Psi}_{x}(\alpha)$, defined by
\begin{equation}
\label{PalphaHmI}
 \mathcal{P}^{\Psi,m}_{x}(\alpha):=
\Bigl\{ \pi  \in  \Delta^{m}(\Theta): \int_{\Theta}{\Psi\bigl(p(x|\theta)\bigr)d\pi(\theta)} \geq \alpha \Bigr\}
\, ,
 \end{equation}
 where $p(x|\cdot):\Theta \rightarrow (0,1]$ is the relative likelihood.
Then we have
\begin{equation}
\label{opt_ODAD20}
\max_{\pi \in \mathcal{P}_{x}(\alpha)} \min_{d\in V}{\L(\pi,d)}=\max_{\pi \in  \mathcal{P}^{m}_{x}(\alpha)}\min_{d\in V}{\L(\pi,d)}
,
\end{equation}
for any $m\geq dim(V)+2$, and
\begin{equation}
\label{opt_ODAD2}
\max_{\pi \in \mathcal{P}^{\Psi}_{x}(\alpha)} \min_{d\in V}{\L(\pi,d)}=\max_{\pi \in  \mathcal{P}^{\Psi,m}_{x}(\alpha)}\min_{d\in V}{\L(\pi,d)}
,
\end{equation}
for any $m\geq dim(V)+3$, unless $\Psi$ is the identity function, when
 $m\geq dim(V)+2$.
\end{Theorem}
\begin{Remark}
Theorem \ref{thm_red} easily generalizes to vector integral functions $\Psi$ of more general form than \eqref{PalphaHmI}, where the number of Diracs required is then
$m\geq dim(V)+\dim(\Psi)+1$ and $m\geq dim(V)+\dim(\Psi)+2$ for the pointwise and integral cases, respectively.
 \end{Remark}

Shapiro and Kleywegt \cite[Thm.2.1]{Shapiro_minimax} implies one can generalize
 Theorem \ref{thm_red} to the more general class of loss functions $\ell$ which are coercive and convex in the second argument, but  requiring more,
$3\bigl(\dim(V)+1)\bigr)$, Dirac measures in the integral case. See \cite[Prop.~3.1]{Shapiro_minimax}.

The following theorem generalizes
 the duality Theorem \ref{thm_duala} to more
 general convex loss functions.
\begin{Theorem}
\label{thm_minball2}
Let $\varphi:\Theta_{x}(\alpha) \rightarrow V$ be continuous and suppose that the loss function satisfies
function $\ell(v_{1},v_{2})=W(v_{1}-v_{2})$, where $W:V \rightarrow \R_{+}$ is  non-negative,  convex and coercive.
For $x \in X$ and $\alpha \in [0,1]$,
suppose that $\Theta_{x}(\alpha)$ is compact. Then $\varphi\bigl(\Theta_{x}(\alpha)\bigr) \subset V$ is compact. Let
  $\lambda \geq 0$ be the smallest value such that there exists a $z\in V$ with
\[ \varphi(\Theta_{x}(\alpha)) +z \subset W^{-1}([0,\lambda]).\]
Then $\lambda$ is the value of the maxmin problem defined by the game \eqref{Lalphadef}
\begin{equation}
\label{maxminopt}
\lambda= \max_{\pi \in \mathcal{P}_{x}(\alpha)} \min_{d \in V}{\E_{\theta \sim \pi}\bigl[\ell(\varphi(\theta),d)\bigr]}\,,
   \end{equation}
and $\pi^{*}$ is maxmin optimal for it if and only if there exists
\[ d_{*} \in   \argmin_{d \in V}{\E_{\theta \sim \pi^{*}}\bigl[\ell(\varphi(\theta),d)\bigr]}
    \]
with
\[ \supp(\varphi_{*}\pi^{*}) \subset W^{-1}(\lambda)-d_{*}. \]
\end{Theorem}

\begin{Remark}
Pass \cite{pass2020generalized} generalizes these results to the case where the decision space $V$
 is a not-necessarily affine metric space and the $\ell^{2}$ distance is replaced by the metric.
\end{Remark}

\section{Supporting theorems and proofs}\label{sec:thmproofs}

\subsection{Minmax theorem}

\begin{Theorem}
\label{thm_minmax}
Consider the saddle function  \eqref{Lalphadef}  for the pointwise and \eqref{Lalphadef0} for the integral rarity assumptions, respectively.
Given the assumptions of Theorem \ref{thm_red}
we have
\[\min_{d \in V}\max_{\pi \in \mathcal{P}_{x}(\alpha) }\L(\pi,d)
= \max_{\pi \in \mathcal{P}_{x}(\alpha) }\min_{d \in V} \L(\pi, d)\, .
\]
and
\[\min_{d \in V}\max_{\pi \in \mathcal{P}^{\Psi}_{x}(\alpha) }\L(\pi,d)
= \max_{\pi \in \mathcal{P}^{\Psi}_{x}(\alpha) }\min_{d \in V} \L(\pi, d)\, .
\]
for the pointwise and integral rarity assumptions, respectively.
\end{Theorem}
\begin{proof}
We prove the result for the integral rarity assumption case only, the pointwise case being much simpler.
The assumptions imply $\Psi(p(x|\cdot))$ is upper semicontinuous, implying  that
$\mathcal{P}^{\Psi}_{x}(\alpha) \subset \mathcal{P}(\Theta)$ is closed, see e.g.~\cite[Thm.~15.5]{Aliprantis2006}. Since
$\mathcal{P}(\Theta)$ is a compact subset of  the space of signed measures in the weak topology,
 see e.g.~Aliprantis and Border\cite[Thm.~15.22]{Aliprantis2006}, it follows that
$\mathcal{P}^{\Psi}_{x}(\alpha) \subset \mathcal{P}(\Theta)$ is compact.
Consequently, to apply  Sion's minmax theorem \cite{sion1958general}
it is sufficient to establish that
the map $\L(\cdot, d):\mathcal{P}(\Theta) \rightarrow \R $ is upper semicontinuous and quasiconcave for each
$d \in V$  and
the map $\L(\pi, \cdot): V\rightarrow \R $ is lower semicontinuous and quasiconvex for each
$\pi  \in  \mathcal{P}(\Theta)$.
To that end, observe that since $\phi$ is continuous and $\Theta$ compact,
the function $\ell(\phi(\cdot),d)$ is bounded and continuous for each $d \in V$. Fixing $d$,
observe  the positivity of the likelihood function implies that the set
{\footnotesize  $$\bigl\{\pi \in \mathcal{P}(\Theta): \E_{\theta \sim \pi_{x}}[\ell(\phi(\theta),d)]\geq r\bigr\}
=\bigl\{\pi \in \mathcal{P}(\Theta): \E_{\theta \sim \pi}[p(x|\cdot)\ell(\phi(\theta),d)]\geq r
\E_{\theta \sim \pi}[p(x|\cdot)]\bigr\}$$}
is closed by the continuity of $\ell(\phi(\cdot),d)$ and $p(x|\cdot)$, see e.g.~\cite[Thm.~15.5]{Aliprantis2006}. Moreover, one can show that the reverse inequality also produces a closed set. Moreover, since it is a linear condition it is convex and therefore
the function $\L(\cdot, d):\mathcal{P}(\Theta) \rightarrow \R $ is upper and lower semicontinuous and quasiconcave for each
$d \in V$. Moreover, fixing $\pi \in \mathcal{P}(\Theta)$, since the function $\ell$ is continuous and $\Theta$ is compact and $\phi$ is continuous, it follows that the function
$\L(\pi, \cdot): V\rightarrow \R $ is continuous and convex and therefore lower semicontinuous and quasiconvex.
Consequently, Sion \cite[Cor.~3.3]{sion1958general}  implies that
\[\inf_{d \in V}\sup_{\pi \in \mathcal{P}^{\Psi}_{x}(\alpha) }\L(\pi,d)
= \sup_{\pi \in \mathcal{P}^{\Psi}_{x}(\alpha) }\inf_{d \in V} \L(\pi, d)
\]
Since, for fixed $d$,  inner optimization  $\sup_{\pi \in \mathcal{P}^{\Psi}_{x}(\alpha) }\L(\pi,d)$
is over of an upper semicontinuous function over a compact set, it achieves its supremum.
Since have established in the proof of Theorem \ref{thm_red} that the inner optimization
$\inf_{d \in v} \L(\pi, d)$  achieves its infimum, we can write
\[\inf_{d \in V}\max_{\pi \in \mathcal{P}^{\Psi}_{x}(\alpha) }\L(\pi,d)
= \sup_{\pi \in \mathcal{P}^{\Psi}_{x}(\alpha) }\min_{d \in V} \L(\pi, d)
\]
Since  the inner minimum $\min_{d \in v} \L(\pi, d)$ is the minimum of a family upper semicontinuous
functions, it produces an upper semicontinuous function, see e.g.~\cite[Lem.~2.41]{Aliprantis2006}, since the maximization of the outer loop is over the compact set $\mathcal{P}^{\Psi}_{x}(\alpha)$,
  we conclude that the supremum on the righthand side is attained. Moreover, since the inner maximum
$\max_{\pi \in \mathcal{P}^{\Psi}_{x}(\alpha) }\L(\pi,d)$ is the maximum over a family of continuous functions, therefore lower semicontinuous functions, it follows that it produces a lower semicontinuous function. Restricting to the compact subset $V^{*}$ from the proof of Theorem \ref{thm_red} we conclude the outer infimum is attained, thus establishing the assertion.
\end{proof}
\subsection{Proof of Theorem \ref{thm_saddles}}

Since the $\ell^{2}$ loss  $(v_{1},v_{2}) \mapsto \|v_{1}-v_{2}\|^{2}$ is strictly convex and coercive in its second argument and vanishes on the diagonal, the assumptions imply that the saddle function
 $\L(\pi,d):=\E_{\theta \sim \pi}\bigl[\|\varphi(\theta)-d\|^2\bigr], \,    \pi \in
  \mathcal{P}_{x}(\alpha),\, d \in V  ,
$
\eqref{Lalphadefa2}
satisfies the conditions of Theorem \ref{thm_red}, so that it follows from Theorem \ref{thm_minmax} that $\L$ satisfies the minmax equality, in particular
establishes the existence
 of a worst-case measure
\begin{equation}
    \label{wcmeas}
\pi^{*} \in \arg \max_{\pi \in \mathcal{P}_{x}(\alpha)}\inf_{d\in V}{\L(\pi,d)}
\end{equation}
and a worst-case decision
\begin{equation}
    \label{wcdec}
d^{*} \in \arg \min_{d\in V}\sup_{\pi \in \mathcal{P}_{x}(\alpha)}{\L(\pi,d)}.
\end{equation}
In addition to establishing the existence of saddle points,
where a pair $\bigl(\pi^{*},d^{*}\bigr) \in  \mathcal{P}_{x}(\alpha) \times V$ is a saddle point of $\L$ if we have
\begin{equation}
    \label{def_saddle2}
\L(\pi,d^{*}) \leq \L(\pi^{*},d^{*}) \leq  \L(\pi^{*},d), \quad \pi \in \mathcal{P}_{x}(\alpha),\, d \in V,\,
\end{equation}
observe that
 Bertsekas et al.~\cite[Prop.~2.6.1]{bertsekas2003convex} asserts
that a pair $(\pi^{*},d^{*}) \in  \mathcal{P}_{x}(\alpha) \times V$ is a saddle point
if and only if they are a worst-case measure and worst-case decision, respectively, as defined in
\eqref{wcmeas} and \eqref{wcdec}.

Now let $(\pi^{*},d^{*}) $ be a saddle point.  As demonstrated in the proof of Theorem  \ref{thm_red}, since the function $d \mapsto  \|\varphi(\theta)-d\|^{2}$  is strictly convex
for all $\theta$, it follows that its expectation $ d \mapsto \L(\pi,d):=\E_{\theta \sim \pi}\bigl[\|\varphi(\theta)-d\|^2\bigr]$
is strictly convex. Moreover a minimizer
\[ d^{**}\in \arg \min_{d \in V}{\L(\pi^{*},d)}\]
exists and by strict convexity it is necessarily unique, see e.g.~\cite{Rockafellar1998}.
Consequently, by the right-hand side of the definition \eqref{def_saddle2} of a saddle point it follows that
$d^{**}=d^{*}$. Since $\pi^{*}$ satisfying \eqref{wcmeas} is equivalent to it satisfying  \eqref{dopt}
and $d^{*}$ satisfying the right-hand side of the definition \eqref{def_saddle2} of a saddle point is equivalent to
it satisfying \eqref{piopt},  the assertion regarding the form \eqref{dopt} and \eqref{piopt} for saddle points is proved.

Let  $(\pi_{1},d_{1})$ and $(\pi_{2},d_{2})$ be two saddle points.
Then by the saddle relation  \eqref{def_saddle2} we have
\[ L(\pi_{2},d_{2}) \leq\L(\pi_{2},d_{1}) \leq \L(\pi_{1},d_{1}) \leq  \L(\pi_{1},d_{2}) \leq  \L(\pi_{2},d_{2}) \]
establishing equality of the value of the risk \eqref{Rstar2}  for all saddle points.

Finally, since $ d \mapsto \L(\pi,d)$ is strictly convex it follows that
its maximum  $d \mapsto \sup_{\pi \in \mathcal{P}_{x}(\alpha)}\L(\pi,d)$ is strictly convex demonstrating
the uniqueness of solutions to \eqref{wcdec}. Moreover, since $\mathcal{P}_{x}(\alpha)$ is convex, the mapping
$\pi \mapsto \L(\pi,d):=\E_{\theta \sim \pi}\bigl[\|\varphi(\theta)-d\|^2\bigr] $ is affine and therefore concave for all $d \in V$, and therefore
its minimum  $\pi \mapsto \inf_{d \in V}\L(\pi,d) $ is also concave. Consequently,  the set of all worst-case measures, that is,  maximizers of
\eqref{wcmeas}, is convex, establishing the final  assertion.

\subsection{Proof of Theorem \ref{thm_duala}}

Since the relative likelihood is continuous, the likelihood region $\Theta_{x}(\alpha)$ is closed and therefore compact, and since
$\varphi$ is continuous, it follows that $\varphi\bigl(\Theta_{x}(\alpha)\bigr)$ is compact and therefore measurable.
According to  Bonnans and Shapiro \cite[Sec.~5.4.1]{bonnans2013perturbation}, because the
 constraint function $(r,z,x) \mapsto \|x-z\|^{2} - r^{2}$ is  continuous, the Lagrangian of
the minimum enclosing ball problem \eqref{opt_minball} is
\begin{equation}
\label{def_lagr}
 L(r,z;\mu):= r^{2}+ \int_{\varphi(\Theta_{x}(\alpha))}{\bigl(\|x-z\|^{2} - r^{2}\bigr) d\mu(x)}, \quad r \in \R,\, z \in V,\,\, \mu \in  \mathcal{M}\bigl(\varphi(\Theta_{x}(\alpha))\bigr)\, .
\end{equation}
 Define
\[ \Psi(\mu):= \inf_{r \in \R,z \in V}{L(r,z;\mu)} \quad \]
and observe that
\[\inf_{r \in \R}{L(r,z;\mu)}=
\begin{cases}
\int_{\varphi(\Theta_{x}(\alpha))}{\|x-z\|^{2}  d\mu(x)},& \quad \int{d\mu}=1\\
-\infty,  & \quad \int{d\mu}\neq 1
\end{cases}
\]
so that
\[\Psi(\mu):= \inf_{r \in \R,z \in V}{L(r,z;\mu)}
=
\begin{cases}
\int_{\varphi(\Theta_{x}(\alpha))}{\bigl\|x-\E_{\mu}[x]\bigr\|^{2}  d\mu(x)},& \quad \int{d\mu}=1\\
-\infty,  & \quad \int{d\mu}\neq 1
\end{cases}
\]
and therefore the dual problem to the minimum enclosing ball problem \eqref{opt_minball}
is
\[ \max_{\mu \in \mathcal{M}(\varphi(\Theta_{x}(\alpha)))}{\Psi(\mu)}= \max_{\mu \in \mathcal{P}(\varphi(\Theta_{x}(\alpha)))}{\E_{\mu}\bigl[\|x-\E_{\mu}[x]\|^{2}\bigr] }\,,   \]
establishing    the Lagrangian duality assertion.

Moreover, since  $\Theta_{x}(\alpha)$ is compact it is Polish, that is Hausdorff and completely metrizable,
and since $\varphi:\Theta_{x}(\alpha)\rightarrow \varphi\bigl(\Theta_{x}(\alpha)\bigr) $ is continuous
\cite[Thm.~15.14]{Aliprantis2006} asserts that
$\varphi_{*}: \mathcal{P}\bigl(
   \Theta_{x}(\alpha)\bigr) \rightarrow  \mathcal{P}(\varphi(\Theta_{x}(\alpha)))$ is surjective.
The change of variables formula \cite[Thm.~13.46]{Aliprantis2006} establishes
that the
 objective function of \eqref{piopt} satisfies
\[\E_{\pi}\bigl[\|\varphi-\E_{\pi}[\varphi]\|^{2}\bigr] =
\E_{\varphi_{*}\pi}\bigl[\|v-\E_{\varphi_{*}\pi}[v]\|^{2}\bigr]\,, \]
so that the surjectivity of  $ \varphi_{*}$ implies  that
 the value  of \eqref{piopt} is  equal to
\begin{equation}
\label{opt57}
\max_{\nu \in \mathcal{P}(\varphi(\Theta_{x}(\alpha)))}  \E_{\nu}{\bigl[\|v-\E_{\nu}{[v]}\|^{2}\bigr]}\, .
\end{equation}
The primary assertions  then follow from  Lim and McCann's \cite[Thm.~1]{lim2021geometrical}
generalization of the one-dimensional result of Popoviciu \cite{Popoviciu} regarding the relationship between variance maximization and the   minimum enclosing ball of the domain  $\varphi(\Theta_{x}(\alpha)) \subset V$.

\subsection{Proof of Theorem \ref{pioptreduced}}

 The proof of Theorem \ref{thm_red} using rarity assumptions of the current form  \eqref{rarityass}
 does not require that the likelihood function $p(x',\cdot)$ be continuous for all $x' \in X$ but only at $x$.
 It asserts the finite-dimensional reduction  \eqref{eq_pioptreduced}  for
 $m\geq dim(V)+2$ Dirac measures. On the other hand, the duality Theorem \ref{thm_duala} implies that the optimality of such a measure
 $\pi:=\sum_{i=1}^{m}{w_{i}\updelta_{\theta_{i}}}$
 is equivalent to the images of these  Diracs  $\updelta_{\varphi(\theta_{i})}, i=1,\ldots, m$ lying on the intersection
 $\varphi(\Theta_{x}(\alpha))\cap \partial B$ of the image of the likelihood region and
 the boundary of its minimum enclosing ball $B$ and that the weights of these Diracs determine that the center of mass
 of this image measure $\varphi_{*}\pi$ is the center $d^{\alpha}$ of the ball $B$, expressed as the right-hand side of
 \eqref{eqjhegvdievd6}. Consequently, the center $d^{\alpha}$ is in the convex hull of the $m$ points
 $\varphi(\theta_{i}), i=1,\ldots, m$ and by Caratheodory's theorem, see e.g.~Rockafellar \cite{rockafellar1970convex},
 $d^{\alpha}$ is in the convex hull of $dim(V)+1$ of these points. Let $S \subset \{1, \ldots, m\}$  correspond to such a subset. Then by the if and only if characterization
 of duality Theorem \ref{thm_duala} it follows that  the subset $\varphi(\theta_{i}), i \in S$ of $dim(V)+1$ image points, using the weights
 $w'_{i}, i \in S$ defining this convex
 combination to be the center $d^{\alpha}$ corresponds to an optimal measure
 $\pi':=\sum_{i\in S}{w'_{i}\updelta_{\theta_{i}}}$, thus establishing the assertion.

\subsection{Proof of Theorem \ref{thm_asympt} }
The following are standard results in the statistics literature, see Casella and Berger \cite{casella2003statistical}.
The idea is to write the Taylor expansion of the log-likelihood around the MLE, then to consider properties of the MLE, and finally apply the law of large numbers and  Slutsky's theorem.

For simplicity, first let $\theta\in \Theta\subseteq \mathbb{R}$. Since the
  second term on the right-hand side in the Taylor expansion
\begin{multline*}
    \sum_{i=1}^N \ln p(x_i|\theta)=\sum_{i=1}^N \ln p(x_i|\hat{\theta}_N)+\sum_{i=1}^N\frac{\partial}{\partial\theta} \ln p(x_i|\theta)_{\theta=\hat{\theta}_N}(\theta - \hat{\theta}_N) \\
    +\frac{1}{2}\sum_{i=1}^N\frac{\partial^2}{\partial\theta^2} \ln p(x_i|\theta)_{\theta=\hat{\theta}_N}(\theta-\hat{\theta}_N)^2 + o_p(1)
\end{multline*}
of  the log-likelihood around the MLE, $\hat{\theta}_N$,
vanishes  by the first order condition of the MLE, we obtain
 $$\sum_{i=1}^N \ln p(x_i|\theta)-\sum_{i=1}^N \ln  p(x_i|\hat{\theta}_N)=\frac{1}{2}\sum_{i=1}^N\frac{\partial^2}{\partial\theta^2} \ln p(x_i|\theta)_{\theta=\hat{\theta}_N}(\theta-\hat{\theta}_N)^2 + o_p(1), $$
 which we write  as
 $$-2\sum_{i=1}^N \ln \big(\frac{p(x_i|\theta)}{p(x_i|\hat{\theta}_N)}\big)=-\frac{1}{N}\sum_{i=1}^N\frac{\partial^2}{\partial\theta^2} \ln p(x_i|\theta)_{\theta=\hat{\theta}_N}(\sqrt{N}(\theta-\hat{\theta}_N))^2 + o_p(1).$$
 By the law of large numbers  and the consistency of the MLE we have
\begin{equation}\label{eqcasa}
-\frac{1}{N}\sum_{i=1}^N\frac{\partial^2}{\partial\theta^2} \ln p(x_i|\theta)_{\theta=\hat{\theta}_N}\xrightarrow{\text{P}}-E_{X\sim P(\cdot|\theta)}\frac{\partial^2}{\partial\theta^2} \ln p(X|\theta)=I(\theta),
\end{equation}
where
 $I(\theta):=-E_{X\sim P(\cdot|\theta)}\frac{\partial^2}{\partial\theta^2} \ln p(X|\theta)$ is the
 Fisher information and $\xrightarrow{\text{P}}$ represents  convergence in probability. By the asymptotic efficiency of the MLE (under our regularity assumptions), the variance $V_{0}(\theta)$  of the MLE  $\hat{\theta}_N$ is the inverse of the Fisher information.  That is
$$I(\theta)=(V_0(\theta))^{-1}\, .$$
Moreover, under the regularity conditions of Casella and Berger \cite[Section.~10.6.2]{casella2003statistical},  we have
$$\frac{\sqrt{N}(\hat{\theta}_N-\theta)}{\sqrt{V_0(\theta)}}\xrightarrow{\text{d}} N(0,1),$$

where $\xrightarrow{\text{d}}$ represents  convergence in distribution, and
therefore Slutsky's theorem implies
$$-2\ln \bar{p}(\mathcal{D}|\theta)= -2\sum_{i=1}^N \ln \biggl(\frac{p(x_i|\theta)}{p(x_i|\hat{\theta}_N)}\biggr)\xrightarrow{\text{d}} \chi^2_1\, .$$

In the more general case $\Theta \subset \R^{k}$, under the regularity conditions of Casella and Berger \cite[Section.~10.6.2]{casella2003statistical},  we have
\begin{equation}\label{eqcasea2}
\sqrt{N}(\hat{\theta}_N-\theta)\xrightarrow{\text{d}} N(0,I(\theta)^{-1})\ ,
\end{equation}
where $I(\theta)$ is the Fisher information matrix, and
 the same argument goes through obtaining
$$-2\ln  \bar{p}(\mathcal{D}|\theta)\overset{a}{=} N(\theta-\hat{\theta}_N)I(\theta)^{-1} (\theta-\hat{\theta}_N)\xrightarrow{\text{d}}  \chi^2_k,$$
where $\overset{a}{=}$ represents  asymptotic equality. Therefore,
  for large sample sizes, one may use the following approximation
 \begin{equation*}
     -2\ln\bar{p}(\mathcal{D}|\theta)\approx \chi^2_k\ .
 \end{equation*}
Finally, the condition $\bar{p}(\mathcal{D}|\theta)<\alpha$ in the computation of $\beta_\alpha$ is the same as
\begin{equation*}
    -2\ln\bar{p}(\mathcal{D}|\theta)>2\ln(\frac{1}{\alpha})\,,
\end{equation*}
so that
consequently, for large sample sizes,  $\beta_\alpha$ may be approximated by
\begin{equation*}
    \beta_\alpha\approx 1-\chi^2_k\bigl(2\ln(\frac{1}{\alpha})\bigr)\, .
\end{equation*}

\subsection{Proof of Theorem \ref{thmasymcaseB} }
For case (A) the proof is an application of  Theorem \ref{thm_asympt}.
Case (B) follows from a direct adaptation of the proof of  Theorem \ref{thm_asympt}.
The first main step (in this adaptation) is to use the convergence of $Q_N$ and the Independence of the $y_i$ given the $t_i$ to
replace \eqref{eqcasa} by
\begin{equation}\label{eqcasab}
-\frac{1}{N}\sum_{i=1}^N\frac{\partial^2}{\partial\theta^2} \ln p(y_i|\theta, t_i)_{\theta=\hat{\theta}_N}\xrightarrow{\text{P}}-E_{t \sim Q, y\sim P(\cdot|\theta,t)}\frac{\partial^2}{\partial\theta^2} \ln p(y|\theta,t)\,.
\end{equation}
The second main step is to derive the  asymptotic consistency and normality \eqref{eqcasea2} of the MLE in case (B). This can be done by adapting the proofs of \cite[Lec.~5]{dudley200918}.

\subsection{Proof of Theorem \ref{thm_red}}
We prove the theorem for the integral rarity case only, the pointwise rarity case being much simpler.
First note that Theorem \ref{thm_minmax} asserts that the saddle function $\L$ satisfies a minmax equality, in particular,
that the inner loop $\min_{d\in V}{\L(\pi,d)}$ of the  primary assertion \eqref{opt_ODAD2}
indeed has a solution.
To analyze such a solution, recall, by \eqref{Lalphadef0}, that  $\L(\pi,d)$ is the expectation
\begin{equation}
 \L(\pi,d):=\E_{\theta \sim \pi_{x}}[\ell(\phi(\theta),d)]\, .
\end{equation}
It is easy to show that the expectation of a family of strictly convex functions is strictly convex, so that
 it follows, for fixed $\pi$, that
$\L(\pi,d)$ is strictly convex. Since $\phi$ is continuous and $\Theta$ is compact, it follows that
the image $\phi(\Theta) \subset V$ is compact and since
$\ell$ is coercive in its second variable, continuous in its first and
$\phi(\Theta)$ is compact, it follows that $\ell$ is uniformly coercive in $\theta$, that is,
for every  $y \in \R_{+}$, there exists an $R \in \R_{+}$ such that
$|d|\geq R \implies \ell(\phi(\theta),d) \geq y,\,\, \theta \in \Theta$.
It follows that $\L(\pi,d) \geq y, |d|\ge R, \pi \in \mathcal{P}(\Theta)$.
Since $y$ is arbitrary and $\L$ is convex, it follows that $\L$ achieves its minimum and
since it is strictly convex this minimum is achieved at a unique point $d_{\pi}$, see e.g.~\cite{Rockafellar1998}.
Since $\ell$ is continuous in its first  variable, $\phi$ is continuous and $\Theta$ compact, it follows that
$\ell(\phi(\theta),0)$ is uniformly bounded in $\theta$ and therefore implies a uniform bound on
$\L(\pi,0), \pi \in  \mathcal{P}(\Theta)$. Consequently, the coerciveness of $\L$ implies that we can
uniformly bound the unique optima $d_{\pi}, \pi \in \mathcal{P}(\Theta)$.

Let $V^{*} \supset \phi(\Theta)$ be a closed cube in $V$ containing  an open neighborhood of the  image $\phi(\Theta)$ and this feasible set of optima just discussed.
Since $\phi$ is continuous, $\Theta$ is compact,  and $\ell$ is  continuously differentiable, it
follows that $\nabla_{d}\ell(\phi(\theta),d)$
is uniformly bounded in both $\theta$ and $d \in V^{*}$. Consequently,
the Leibniz theorem for differentiation under the integral sign,
 see e.g.~Aliprantis and Burkinshaw \cite[Thm.~24.5]{aliprantis1998principles},
implies, for fixed $\pi$, that
\begin{equation}
\label{dL}
\nabla_{d} \L(\pi,d):=\E_{\theta \sim \pi_{x}}\bigl[\nabla_{d}\ell(\phi(\theta),d)\bigr],\quad d \in V^{*}\, .
\end{equation}
Consequently, the relation $0=\nabla_{d} \L(\pi,d_{\pi})$ at the unique minimum $d_{\pi}$  of
  $\min_{d\in V}{\L(\pi,d)}$  implies that
\begin{equation}
\label{oemiir}
\E_{\theta \sim \pi_{x}}\bigl[\nabla_{d}\ell(\phi(\theta),d_{\pi})\bigr]=0, \quad \pi \in \mathcal{P}(\Theta)\,.
\end{equation}

The formula \eqref{eq_conditional}
for the conditional measure  and the
 positivity of its denominator
imply that we can write \eqref{oemiir} as
\begin{equation}
\label{oemiir2}
\frac{1}{ \int_{\Theta}{p(x|\cdot)d\pi}}\E_{\theta \sim \pi}\bigl[p(x|\theta)\nabla_{d}\ell(\phi(\theta),d_{\pi})\bigr] =0\,
\end{equation}
which is equivalent to
\begin{equation}
\label{oemiir2a}
\E_{\theta \sim \pi}\bigl[p(x|\theta)\nabla_{d}\ell(\phi(\theta),d_{\pi})\bigr] =0\,
\end{equation}

Consequently, adding the constraint \eqref{oemiir2a}, equivalent to the   minimization problem,
 the maxmin problem   on the left-hand side of
\eqref{opt_ODAD2} can be written
\begin{equation}
\label{maxopt}
\begin{cases}
\text{Maximize } \L(\pi,d)
 \\
\text{Subject to }  \,\, \pi  \in \mathcal{P}_{x}^{\Psi}(\alpha), \, \, d \in V^{*}
\\
\E_{\theta \sim \pi}\bigl[p(x|\theta)\nabla_{d}\ell(\phi(\theta),d)\bigr] =0
\end{cases}
\end{equation}
which again using the conditional formula \eqref{eq_conditional} and the definition of
$ \L(\pi,d)$ can be written
\begin{equation}
\label{maxopt2}
\begin{cases}
\text{Maximize } \frac{1}{\int_{\Theta}{p(x|\theta)d\pi(\theta)}}\E_{\theta \sim \pi}\bigl[p(x|\theta)\ell(\phi(\theta),d)\bigr]
 \\
\text{Subject to }  \,\, \pi  \in \mathcal{P}_{x}^{\Psi}(\alpha), \, \, d \in V^{*}
\\
\E_{\theta \sim \pi}\bigl[p(x|\theta)\nabla_{d}\ell(\phi(\theta),d)\bigr] =0
\end{cases}
\end{equation}
which, introducing a new variable, can be written
\begin{equation}
\label{maxopt3}
\begin{cases}
\text{Maximize } \frac{1}{\epsilon}\E_{\theta \sim \pi}\bigl[p(x|\theta)\ell(\phi(\theta),d)\bigr]
 \\
\text{Subject to }  \,\, \pi  \in \mathcal{P}_{x}^{\Psi}(\alpha), \, \, d \in V^{*}, \epsilon >0
\\
\E_{\theta \sim \pi}\bigl[p(x|\theta)\nabla_{d}\ell(\phi(\theta),d)\bigr] =0,\quad
\epsilon =\E_{\theta \sim \pi}[p(x|\theta)]\, .
\end{cases}
\end{equation}
Now fix $\epsilon >0$ and $d \in V^{*}$ and consider the inner maximization loop
\begin{equation}
\label{maxopt4}
\begin{cases}
\text{Maximize } \frac{1}{\epsilon}\E_{\theta \sim \pi}\bigl[p(x|\theta)\ell(\phi(\theta),d)\bigr]
 \\
\text{Subject to }  \,\, \pi  \in \mathcal{P}_{x}^{\Psi}(\alpha), \, \,
\\
\E_{\theta \sim \pi}\bigl[p(x|\theta)\nabla_{d}\ell(\phi(\theta),d)\bigr] =0,\quad
\epsilon =\E_{\theta \sim \pi}[p(x|\theta)]\, .
\end{cases}
\end{equation}
Since this is linear optimization of the integration of a  non-negative, and thus integrable function, with possible integral value $+\infty$ for all $\pi$,  over the full simplex of probability measures subject to
$\dim(V)+1$ linear equality constraints defined by integration against measurable functions,
  plus one linear inequality constraint defined by integration against  a measurable function,
  \cite[Thm.~4.1]{owhadi2013optimal}, which uses von Weizsacker and Winkler \cite[Cor.~3]{weizsacker1979integral}, see also
Karr
\cite{karr1983extreme} which is applicable under more assumptions on the model $\mathbb{P}$, implies  this optimization problem can be reduced to optimization over the convex combination
of $dim(V)+3$ Dirac measures supported on $\Theta$. Since the full problem is
 the supremum of such problems,  using the compactness of the space $\mathcal{P}^{\Psi,m}_{x}(\alpha)$ in the weak topology, the  primary assertion follows. When $\Psi$ is the identity function
one of the constraints disappears, and the assertion in that case follows.

\subsection{Proof of Theorem \ref{thm_minball2}}

The proof follows from the invariance of the variance under $\varphi_{*}$ and the equality of their maximum variance problems established in the  proof of Theorem \ref{thm_duala}  and
    Lim and McCann's \cite[Thm.~2]{lim2021geometrical}
generalization of  their $\ell^{2}$ result \cite[Thm.~1]{lim2021geometrical}.

\subsection{Proof of Theorem \ref{thm_miniball}}
First consider the $\epsilon >0$ case.
Our proof will use results from B\u{a}doiu, Har-Peled and Indyk \cite{badoiu2002approximate}.
Consider the REPEAT loop.  As previously mentioned, the
 FOR loop always gets broken at Step \ref{forloopbreak} for some $x$ since,
by Theorem \ref{thm_duala}, the center of the ball must lie in the convex hull of the  $n+2$ points, but by  Caratheodory's theorem this center also lies in the convex hull of $n+1$ of those points, and Theorem \ref{thm_duala} then asserts that this ball is also the minimum enclosing ball of those $n+1$ points.
Clearly, the breaking of the step implies that the elimination of the point does not change the
 current ball.
Consequently, the only change in the  current ball is through the discovery
in Step \ref{distantpoint} of a distant point and its addition to the working set
$A_{k}$ followed by the calculation  in Step \ref{newball}  of a new minimum ball containing this
 enlarged working set. Let $B(A_{k})$ denote the minimum enclosing ball of $A_{k}$,
$R(A_{k})$ denote its radius and, overloading notation, let us  denote $R(A_{k}):=R(B_{k})$.
Then when a new point is added to $A_{k}$ to obtain $A_{k+1}$, it follows from  $A_{k} \subset A_{k+1}$
that $B(A_{k+1}) \supset A_{k+1} \supset  A_{k}$ and therefore  $R(A_{k+1}) \geq R(A_{k})$. Likewise
$A_{k} \subset K$ implies that  $R(A_{k}) \leq R$, the radius of the minimum enclosing  ball $B$ of $K$.
Consequently
 the  sequence  $R(A_{k}) \leq R$ of the radii  of the  balls is monotonically increasing and bounded by $R$.
Moreover,
\cite[Clm.~2.4]{badoiu2002approximate}, using \cite[Lem.~2.2]{badoiu2002approximate} from
Goel et al.~\cite{goel2001reductions},
implies  that, until the stopping criterion in Step \ref{untils}     is satisfied, we have
\begin{equation}
\label{oeokjiir}
  R(A_{k+1}) \geq \bigl(1+\frac{\epsilon^{2}}{16}\bigr) R(A_{k}),
\end{equation}
and when the stopping criterion is satisfied it follows from
Step \ref{distantpoint} that
the output in Step \ref{returns}  satisfies
\begin{equation}
\label{eiienuuururu}
B^{(1+\epsilon)(1+\delta)}_{k-1}\supset K.
\end{equation}
Observe that we have $R\leq \Delta$   where $\Delta:=\diam(K)$ and
  the initialization implies that $R(A_{0}) \geq \frac{1}{2(1+\delta)}\Delta$. Consequently
 \eqref{oeokjiir}  implies that the radius $ R(A_{k})$  increases by at least
$\frac{\epsilon^{2}}{32((1+\delta)}\Delta$ at each step.  Since the sequence is bounded by
$R \leq \Delta$ it follows that at most $\frac{16}{\epsilon^{2}}(1+2\delta)$ of the REPEAT loop can be taken before terminating at Step \ref{untils}. Upon termination
 the returned  ball $B^{*}:=B^{(1+\epsilon)(1+\delta)}_{k-1}$ in  Step \ref{returns}, by  \eqref{eiienuuururu}, satisfies
\[ B^{*} \supset K,\]
and since
\[R(B^{(1+\epsilon)(1+\delta)}_{k-1})= (1+\epsilon)(1+\delta)R(B_{k-1}) \leq (1+\epsilon)(1+\delta)R\]
we obtain
\[ R\bigl(B^{*}\bigr) \leq  (1+\epsilon)(1+\delta)R , \]
establishing the  primary assertion. Since each step in the REPEAT loop adds at most one new point to the
working set $A_{k}$, it follows that
the working set size is bounded by $2$ plus the number of steps in the REPEAT loop, that is
  $2+\frac{16}{\epsilon^{2}}(1+2\delta)$. Since the WHILE loop keeps the bound $\leq n+2$, the proof is finished.

Now consider the $\epsilon =0$ case. First let $\delta=0$.
By compactness of the set $K$, there exists a sub-sequence $(A_{t_1})$, indexed by $T_1\subseteq \mathbb{N}$, of $(A_t)$ such that $C(A_{t_1})\to C_1$, in our notation whenever the algorithm stops or we reach to a fixed point the sequence repeats the last set of points. For every $t_1\in T_1$, let $y^{t_1}$ be the point selected as a furthest point from $B(A_{t_1})$ in the algorithm to form $A_{t_1+1}$. Again by compactness of the set $K$, there exists a sub-sequence $(y^{t_2})$, indexed by  $T_2\subseteq T_1\subseteq \mathbb{N}$, of $(y^{t_1})$ such that $y^{t_2}\to y^*$.
By another application of compactness of $K$, there exists a sub-sequence $(A_{t_3})$, indexed by $T_3\subseteq T_2\subseteq T_1\subseteq \mathbb{N}$, of $(A_{t_2})$ such that $C(A_{t_3}\cup\{{y^{t_3}}\})\to C_2$.
Since $C(A_{t_3})\to C_1$ and $R(A_{t_3})\uparrow R_0\leq R(K)$ as $t_3\in T_3\to \infty$, it follows that $B(A_{t_3})\to B(C_1,R_0)$.

To complete the proof it is sufficient to show that $K\subseteq B(C_1,R_0)$,
  since then $R_0\leq R(K)$ implies that $B(C_1,R_0)= B(K)$. To that end, we demonstrate that
\[d:=\max_{x\in K} dist(x,B(C_1,R_0))=0.\]
Since $dist()$ is a continuous function in both of its arguments, $B(A_{t_3})\to B(C_1,R_0)$ as $t_3\in T_3\to \infty$, $y^{t_3^k}\in \operatorname{argmax}_{x\in K} dist(x,B(A_{t_3^k}))$ for every $t_3^k\in T_3$ and  $y^{t_3^k}\to y^*$, it follows that
\begin{equation}
\label{eooooe}
dist(y^*,B(C_1,R_0))=d\, .
\end{equation}
 By the choice of $T_3$, $y^{t_3}\to y^*$, $B(A_{t_3})\to B(C_1,R_0)$, and $B(A_{t_3}\cup\{{y^{t_3}}\})\to B(C_2,R_0)$. Therefore, for $\epsilon_1>0$, there exists a large number $N_{\epsilon_1}\in \mathbb{N}$, such that for all $t_3\in T_3$ with $t_3 \geq N_{\epsilon_1}$, we have
\begin{equation}
\label{eopokiier}
 A_{t_3}\subseteq B(C_1,R_0+\epsilon_1)\, \, \text {and} \,\,A_{t_3}\cup\{y^{t_3}\}\cup{y^*} \subseteq B(C_2,R_0+\epsilon_1).
\end{equation}
Consequently   \eqref{eooooe}, \eqref{eopokiier} and the triangle inequality imply
\[dist(C_1,C_2)\geq dist(y^*,C_1)-dist(y^*,C_2)\geq (d+R_0)-(R_0+\epsilon_1)=d-\epsilon_1\]
 and
\[A_{t_3}\subseteq B(C_1,R_0+\epsilon_1)\cap B(C_2,R_0+\epsilon_1),  \quad t_3\in T_3,
  t_3 \geq N_{\epsilon_1}.\]

 Let $\bar{C}=\frac{C_1+C_2}{2}$ and consider the hyperplane orthogonal to the vector $C_1-C_2$ passing through $\bar{C}$. By Pythagoras' Theorem, $B(C_1,R_0+\epsilon_1)\cap B(C_2,R_0+\epsilon_1)\subseteq B(\bar{C}, \sqrt{(R_0+\epsilon_1)^2-(\frac{d-\epsilon_1}{2})^2})$ and therefore, $A_{t_3}\subseteq B(\bar{C}, \sqrt{(R_0+\epsilon_1)^2-(\frac{d-\epsilon_1}{2})^2})$, which implies that $R(A_{t_3}) \leq   \sqrt{(R_0+\epsilon_1)^2-(\frac{d-\epsilon_1}{2})^2}$ for all $t_3\in T_3$ with $t_3 \geq N_{\epsilon_1}$.
By sending $\epsilon_1$ to zero, we obtain that $R(A_{t_3})\leq \sqrt{(R_0)^2-(\frac{d}{2})^2}$ as $t_3\in T_3\to \infty$, but $ R(A_{t_3})\uparrow R_0$ implies $d=0$, which completes the proof.

For the general case $\delta\geq 0$, we can use the same technique. Let $K'=(y^t)_{t\in\mathbb{N}}$ be a sequence of points selected as the furthest point in Step \ref{distantpoint} (with the relative error size of $\delta$) in one complete execution of the algorithm. Using the result and the language of the case $\delta=0$ applied to the set $K'$, we obtain that $B(K')=B(C_1,R_0)$, where $C_1$ and $R_0$ are the center and radius returned by the algorithm as $t \to \infty$, with the convention that whenever the algorithm stops we repeat the last set of points up to infinity.

Note that $\max_{x\in K}dist(C_1,x)\leq (1+\delta)R_0$, since otherwise the algorithm would have not  converged to $C_1$, and therefore
 $K\subseteq B(C_1, (1+\delta)R_0)$ which implies that $R(K)\leq(1+\delta)R_0$. Moreover, by $K'\subseteq K$ we have $R_0=R(K')\leq R(K)$ and therefore  $R_0\leq R(K)\leq (1+\delta)R_0$, completing the proof.

\subsection*{Acknowledgments}
Part of this research was carried out at the Jet Propulsion Laboratory, California Institute of Technology, under a contract with the National Aeronautics and Space Administration.
The authors gratefully acknowledge support
from Beyond Limits (Learning Optimal Models) through CAST (The Caltech Center for Autonomous Systems and Technologies) and partial support
from the Air Force Office of Scientific Research under award number FA9550-18-1-0271 (Games for Computation and Learning).

\noindent \copyright 2021. California Institute of Technology. Government sponsorship acknowledged.

\bibliographystyle{plain}
\bibliography{merged}

\end{document}